\documentclass[letterpaper,11pt]{article}

\input{preamble.sty}

\begin{document}

 \doparttoc 
\faketableofcontents 

	\algrenewcommand\algorithmicrequire{\textbf{Input:}}
	\algrenewcommand\algorithmicensure{\textbf{Output:}}
	\title{
	Robust Auction Design with Support Information 
	 }
	\ifx\blind\undefined
	\author{ 
	Jerry Anunrojwong\thanks{Columbia University, Graduate School of Business. Email: {\tt janunrojwong25@gsb.columbia.edu}} \and Santiago R. Balseiro\thanks{Columbia University, Graduate School of Business. Email: {\tt srb2155@columbia.edu}} \and Omar Besbes\thanks{ Columbia University, Graduate School of Business. Email: {\tt ob2105@columbia.edu}}
	}
	\fi




\date{\today}

\maketitle
\begin{abstract}

A seller wants to sell an item to $n$ buyers. Buyer valuations are drawn i.i.d.~from a distribution unknown to the seller; the seller only knows that the support is included in  $[a,b]$. To be robust, the seller chooses a DSIC mechanism that optimizes the \textit{worst-case} performance relative to the \jaedit{ideal expected revenue the seller could have collected with knowledge of buyers’ valuations}.  Our analysis unifies the \textit{regret} and the \textit{ratio} objectives.

For these objectives, we derive an optimal mechanism and the corresponding performance in quasi-closed form, as a function of the support information \jaedit{$[a,b]$} and the number of buyers $n$. Our analysis reveals three regimes of support information and a new class of robust mechanisms. i.) \jaedit{When $a/b$ is below a threshold}, the optimal mechanism is a second-price auction (SPA) with random reserve, a focal class in earlier literature. ii.) \jaedit{When $a/b$ is above another threshold}, SPAs are strictly suboptimal, and an optimal mechanism belongs to a  class of mechanisms we introduce, which we call  \textit{pooling auctions} (POOL);  whenever the highest value is above a threshold, the mechanism still allocates to the highest bidder, but otherwise the mechanism allocates to a \textit{uniformly random buyer}, i.e., pools low types. iii.) When \jaedit{$a/b$ is between two thresholds}, a randomization between SPA and POOL is optimal. 

We also characterize optimal mechanisms within nested central subclasses of mechanisms: standard mechanisms that only allocate to the highest bidder, SPA with random reserve,  and SPA with no reserve. We show strict separations in terms of performance across classes, implying that deviating from standard mechanisms is necessary for robustness. 

Lastly, we show that the same results hold under other distribution classes that capture ``positive dependence,'' namely: i.i.d., mixture of i.i.d., and exchangeable and affiliated distributions, as well as i.i.d. regular distributions.


\end{abstract}

\newpage

\setstretch{1.5}

\section{Introduction}\label{sec:intro}

The question of how to optimally sell an item underlies much of modern marketplaces, from online advertising and e-commerce to art auctions. Selling mechanisms are widely used in practice, and in turn they are studied in economics, computer science, and operations research under \textit{optimal mechanism design}, starting from the pioneering work of \citep{Myerson81}.  The literature often assumes that the seller knows the environment perfectly, but (i) this knowledge is often either not available or reliable, and (ii) the optimal mechanism prescribed by the theory is often too complicated or fine-tuned to the details of the environment, to be used in practice. There is therefore a need to develop mechanisms that depend less on market details, and this need is often referred to as the ``Wilson doctrine'' \citep{Wilson87}. 

The emerging literature on \textit{robust mechanism design}, in turn, aims to design mechanisms that perform ``well'' in the worst case against ``any'' environment. This line of work often leads to interesting insights but, taken literally, they can lead to mechanisms that are too conservative. In practice, while we do not have complete knowledge about the environment, we often do have \textit{partial} knowledge, and how to incorporate additional side information into the robust framework is essential to bring the robust theory closer to practice. In this paper, we make progress in this direction by analyzing the role of \textit{support information} of bidder valuations, as captured by \textit{lower bounds} and \textit{upper bounds} on bidder valuations. 

The motivation of the knowledge of such bounds  is that we operate in a world with minimal or no data. The bounds need not be learned from data but rather are derived from asking experts, using domain knowledge, or common sense. Examples include launching a new product, or auctioning rarely traded goods such as fine art, collectibles, and jewelry. In these contexts, the support information is a natural form of partial knowledge because it is easier and more intuitive to come up with a reasonable range of values than to guess something like the shape of the valuation distribution (either parametric or nonparametric like regularity or monotone hazard rate) or distributional parameters like the mean or the optimal monopoly price. 

More formally, consider a seller who wants to sell an item to $n$ bidders. The bidders' valuations are unknown to the seller and are assumed to be drawn from a joint distribution $\mathbf{F}$. The seller does not know $\mathbf{F}$, and knows only a lower bound $a$ and an upper bound $b$ on the support of $\mathbf{F}$, and that the valuations belong to the class $\mathcal{F}$ of  i.i.d. distributions. Similarly, the bidders also do not know $\mathbf{F}$. Therefore, we focus on mechanisms that are \textit{dominant strategy incentive compatible} (DSIC). Under such a mechanism, every bidder optimally reports her true value regardless of other bidders' valuations and strategies.

\obdelete{We will quantify the performance of mechanisms by the \textit{gap} between the benchmark oracle revenue and the mechanism revenue. The benchmark revenue is the ideal expected revenue the seller could have collected with knowledge of the buyers' valuations, while the mechanism revenue is the expected revenue garnered by the actual mechanism.}
\jaedit{We will quantify the performance of mechanisms by the \textit{gap} between the benchmark oracle revenue, the ideal expected revenue the seller could have collected with knowledge of the buyers' valuations, and  the expected revenue garnered by the actual mechanism.} Our framework will be general and apply to two classical notions of gaps considered in the literature: \obcomment{Our framework also applies to maxmin performance. it is just that the performance is degenerate/ the analysis is trivial in this case. Should we say so in passing?}  (i) the \textit{regret} (absolute gap) is the difference between these two revenues, and (ii) the \textit{approximation ratio} (relative gap) is the ratio of these two revenues. The seller selects a mechanism that performs well (minimizes regret or maximizes approximation ratio) in the \textit{worst case} against all admissible distributions.\jaedit{\footnote{Our framework also applies to the traditional maxmin (worst-case) revenue, but the worst-case revenue is trivial in the present case of support information.}}

The interval $[a,b]$ associated with the admissible distribution class captures the amount of uncertainty of the decision maker. We will parameterize this uncertainty through $a/b$, which we call the \textit{relative support information}, and which is a unitless quantity ranging from 0 to 1. When $a/b \sim 0$ (either because $a \sim 0$ or $b \gg a$), we have minimal relative support information while when $a/b \sim 1$ we have maximal support information as the endpoints are close. 

\jadelete{For the special case of pricing, i.e., when there is one bidder, the problem is well understood. It was analyzed concurrently in \cite{BergemannSchlag08} for the regret objective and  in \cite{ErenMaglaras10} for the approximation ratio objective. The one-bidder case is simpler because it is sufficient for the seller to optimize over randomized pricing. In the multiple-bidders case, however, the space of feasible mechanisms is much richer and more unstructured, making the seller's problem challenging: only the case $a=0$, which corresponds to minimal relative support information, has been studied. \cite{our-paper-ec} shows that a second-price auction with appropriately randomized reserve price is an optimal  minimax regret mechanism. It is also possible to show that with minimal relative support information, no mechanism can guarantee positive worst-case approximation ratio.  The understanding of the interplay of support information and robust auctions is very limited outside of these special cases, leading to the following question: how does support information affect the structure of optimal robust auctions and achievable performance? \obcomment{This seems to again motivate the setting, which we did already before... although not from the literature angle...} 
}

\jaedit{The understanding of the interplay of support information and robust auctions is very limited outside of very particular special cases ($n = 1$ \citep{BergemannSchlag08,ErenMaglaras10} and $a=0$ \citep{our-paper-ec}), leading to the following question: how does support information affect the structure of optimal robust auctions and achievable performance?} We study optimal  performance and associated mechanisms across the relative support information spectrum and establish richness in the structure of the resulting robust mechanisms with three distinct information regimes corresponding to three mechanism types. In particular, our work subsumes and unifies the three studies mentioned above, characterizing an optimal mechanism and the associated performance for an arbitrary number of bidders $n$ and any support information $[a,b]$, for both the regret and ratio objectives. See Table~\ref{table:previous-work} for a high level summary of known results and the results we develop in this paper.

\begin{table}[ht]
\centering
\begin{center}
\begin{tabular}{ ccc|c } 
 Problem & Information & \multicolumn{2}{c}{Objective} \\ 
 \cline{3-4}
Type & level & Regret & Ratio \\ 
\hline
\hline
pricing ($n=1$) & all $a/b$& \cite{BergemannSchlag08} & \cite{ErenMaglaras10} \\ 
\hline
auctions ($n\ge1$) & $a/b=0$ & \cite{our-paper-ec} & 0 \\ 
 \hline
 \hline
auctions ($n\ge1$) & all $a/b$ & \multicolumn{2}{c}{---\textbf{This work}---}  \\
 \hline
\end{tabular}
\end{center}
\caption{Comparison with the closest previous studies along the dimension of the number of buyers (pricing ($n=1$) vs. auctions (arbitrary $n\ge 1$) and the level of relative support information $a/b$.   }
  \label{table:previous-work}
\end{table}

\subsection{Summary of Main Contributions}

We develop a unified framework for regret and approximation ratio through a single quantity, the minimax $\lambda$-regret, where the $\lambda$-regret is the difference between $\lambda$ times the benchmark revenue and the mechanism revenue, and $\lambda \in (0,1]$ is a constant. \obedit{It is clear that when $\lambda=1$, the $\lambda$-regret reduces to the regret. The fact that the $\lambda$-regret can be used to characterize the maximin ratio relies on an epigraph reformulation of the latter problem, which is fairly standard in the context of optimization with fractional objectives}. Our main contribution, however, is the full characterization of a minimax optimal mechanism and its associated performance for $\lambda$-regret for \textit{any} value of $\lambda \in (0,1]$, \textit{any} number of buyers $n$ and \textit{any} support information $[a,b]$. Since we are primarily interested in the effect of the support $[a,b]$, we initially assume  that the valuations are $n$ i.i.d.~\jaedit{random variables} given the canonical nature of this setting. Our family of optimality results across this spectrum brings to the foreground a very rich structure of optimal mechanisms, and establishes how relative support information critically impacts the structure of optimal mechanisms.  

\paragraph{Novel mechanism class.} A natural candidate for an optimal mechanism is a second-price auction with appropriate random reserve. \jaedit{Previous work \citep{our-paper-ec} shows that this is optimal with zero relative support information, i.e., for $a=0$.}  Suppose for a moment that relative support information is high (i.e., $a \sim b$) and we are restricted to the class of \jaedit{second-price auctions (SPAs)}. Setting any nontrivial reserve is \textit{risky} because when the highest buyer's value is below the reserve the seller does not allocate and gets zero revenue. At the same time, the benefits of a reserve price are limited since the highest and lowest values are close. \obedit{Indeed}, the seller can guarantee a revenue of $a$ with no reserve, which is close to the maximal revenue achievable of $b$. Hence, it should be  intuitive that when relative support information is high, a SPA with no reserve is optimal among the class of SPAs. (A formal result is presented in Section~\ref{subsec:main-spa-rand-mech}.) A natural question is then whether there are mechanisms that can outperform a SPA with no reserve from a robust perspective, and what \obedit{structure they take}\obdelete{their structures are}. 

We define a new mechanism class, with the aim of softening the trade-offs associated with reserve pricing in second-price auctions. These mechanisms, that we dub ``pooling auctions'' \jaedit{($\pool$)}, have an associated threshold.  When the highest bid is above the threshold,  the mechanism allocates to the highest bidder, as in a SPA when the highest bid is above the reserve price. However, when the highest bid is below the threshold, rather than  not allocating as a SPA would do, the seller allocates uniformly at random to any of the bidders. In other words, this auction pools the low types and the lowest bidder may get the item. \jaedit{By increasing the allocation at low values, the mechanism increases the revenue derived from lower-valued bidders, but in doing so, the mechanism can extract less revenue from higher-valued bidders due to incentive compability. In this sense, the pooling auction makes allocation and payment more ``uniform'' across values and softens the tradeoffs from reserve pricing.}

\obcomment{Jerry: this ends abruptly and feels incomplete after the edit. Maybe you want to add a sentence or two about how it softens the trade-offs: allocates more often to low types, so can extract more. at the same time cannot extract as much from high types...  RESPONSE: we already have the intuitive explanation in Section 3 and we don't want to repeat ourselves. I  }

\jadelete{
\paragraph{Characterization of an optimal mechanism.}
Our main theorem fully characterizes a minimax optimal mechanism for any relative support information level, for the $\lambda$-regret (and hence for the minimax regret and maximin ratio). We present an abridged version of our main result here. The full version is available in Theorem~\ref{thm:char-all-mech-full-main}.

\begin{theorem}[Main Theorem, Succinct]\label{thm:char-all-mech-full-main}
Fix $n$ and $\lambda \in (0,1]$. Then, there exists constants $k_l < k_h$, depending only on $n$ and $\lambda$, such that the problem admits an optimal minimax $\lambda$-regret mechanism $m^*$, depending on $a/b$, as follows.
\begin{itemize}
    \item (Low Relative Support Information) For $a/b \leq k_l$, there is a probability distribution of reserves $\Phi$ such that $m^* = \textnormal{SPA}(r)$ with $r \sim \Phi$.
    \item (High Relative Support Information) For $a/b \geq k_h$, there is a probability distribution of thresholds $\Psi$ such that $m^* = \textnormal{POOL}(\tau)$ with $\tau \sim \Psi$.
    \item (Moderate Relative Support Information) For $k_l \leq a/b \leq k_h$, then there is $v^* \in [a,b]$, such that $m^*$ is a randomization over $\textnormal{SPA}$s and $\textnormal{POOL}$s in the following sense. There is a probability distribution $\mathcal{D}$ on $[a,b]$ such that when we draw a sample $r \sim \mathcal{D}$, if $r \leq v^*$, the mechanism is $\textnormal{SPA}(r)$, otherwise the mechanism is $\textnormal{POOL}(r)$. 
\end{itemize}

\end{theorem}
}

\paragraph{Characterization of an optimal mechanism.} Our main result, \obedit{Theorem \ref{thm:char-all-mech-full-main}}, establishes that there always exists an optimal mechanism that is \jaedit{a randomization over second-price auctions (SPA) with different reserves and pooling auctions (POOL) with different thresholds.}  Therefore, an optimal mechanism can be implemented in terms of a random instance of one of these ``base'' mechanisms. Furthermore, three fundamental relative support information regimes emerge. \jaedit{There are thresholds $k_l$ and $k_h$ such that: if $a/b \leq k_l$ (low information regime), $\spa$ with random reserves is optimal; if $a/b \geq k_h$ (high information regime), $\pool$ with random thresholds is optimal; if $k_l \leq a/b \leq k_h$ (moderate information regime), a randomization over $\spa$ and $\pool$ (i.e. interpolation between the two extremes) is optimal.}

\obdelete{The resulting optimal mechanism inherits qualitative features from the base mechanisms $\textnormal{SPA}$ and $\textnormal{POOL}$, and so their properties depend critically on the relative support information.}
We note that $\textnormal{SPA}$ is a \jaedit{``standard''} mechanism, meaning that it never allocates to non-highest bidders, but $\textnormal{POOL}$ is not. Therefore, the optimal mechanism we have identified is standard if and only if $a/b \le k_l$.\obcomment{Jerry, at this stage it is unclear if this is true as we have only identified AN optimal mechanism and did not show uniqueness...} Secondly, $\textnormal{POOL}$ always allocates, meaning it allocates with probability one, whereas $\textnormal{SPA}$ does not (because it does not allocate below the reserve). Therefore, the optimal mechanism always allocates if and only if $a/b \ge k_h$. \obcomment{Jerry, at this stage it is unclear if this is true as we have only identified AN optimal mechanism and did not show uniqueness...} \jacomment{I think it is okay, because we are specifically talking about the optimal mechanism that we have identified.}

While the result above applies for any $\lambda$, we note that for the maximin ratio problem \obedit{(the problem of maximizing the worst-case ratio of revenue to the benchmark)}, the value of $\lambda$ is endogenous, and it is not clear a priori in which information regime one falls. Quite interestingly, we can prove that the optimal maximin ratio mechanism is never in the $\spa$ regime and thus some amount of pooling is always necessary in this case (see Section~\ref{subsec:structure-mech-all} and Proposition~\ref{prop:maximin-ratio-regime}).

\paragraph{Methodology and closed-form characterization.} We characterize the optimal mechanism and worst-case distribution in closed form via a saddle-point argument. In particular, if we assume that a saddle point exists and the optimal mechanism has the form outlined in the previous paragraph, we derive necessary conditions for Nature's worst-case distribution (cf. Section~\ref{subsec:proof-main-thm}) as well as the distributions of random reserve $r$ and threshold $\tau$, under a few fairly mild technical conditions. We then prove that the resulting mechanism is optimal without any additional assumptions. Our methodology provides a unified treatment across all support information levels, and objectives (regret and approximation ratio) in one framework. We also characterize Nature's worst-case distribution as part of our analysis, which takes the following form:  for $a/b \le k_h$, the worst-case distribution is an isorevenue distribution (i.e., zero virtual value), whereas for $a/b > k_h$, the worst-case distribution has a constant positive virtual value in the interior of the support.\footnote{For a distribution with CDF $F$ and density $f$, the virtual value at $v$ is defined by $v - (1-F(v))/f(v)$.}

\paragraph{Quantifying the value of scale information and competition.} Using the machinery we develop, we can exactly compute the minimax regret and maximin ratio for any support information $[a,b]$ and number of buyers $n$ \jaedit{(cf. Figure~\ref{fig:maximin-ratio-intro}).} We show that even a small amount of knowledge can lead to nontrivial guarantees on revenue. For example, even when we only know that values can vary over a full order of magnitude ($a/b = 0.10$), we can guarantee $40.38\%$ of the \jaedit{ideal benchmark} with only $2$ buyers. When the knowledge of the scale is more precise, say, if we know the value up to a factor of two ($a/b = 0.50$), we get a guarantee of 74.63\% with 2 buyers.  With more agents, the guarantees improve (around 5\% and 3\% more, respectively, for an additional buyer). 

\obcomment{Jerry, we may want to reference the table/figure in the main text later... Did you remove these from later sections?}

\jadelete{
Approximation ratio values are shown in Table~\ref{table:maximin-ratio-intro}.

\begin{table}[h!]
\centering

\begin{tabular}{c || c | c | c | c | c | c | c | c | c | c } 
 \hline
 $a/b$ & $10^{-4}$ & 0.01 & 0.05 & 0.10 & 0.20 & 0.25 & 0.30 & 0.50 & 0.75 & 0.99 \\ [0.5ex] 
 \hline
 $n=1$ & 0.0979 & 0.1784 & 0.2503 & 0.3028 & 0.3832 & 0.4191 & 0.4537 & 0.5906 & 0.7766 & 0.9900 \\
 $n=2$ & 0.1086 & 0.2158 & 0.3228 & 0.4038 & 0.5197 & 0.5660 & 0.6077 & 0.7463 & 0.8841 & 0.9957 \\ 
 $n=3$ & 0.1148 & 0.2406 & 0.3673 & 0.4529 & 0.5668 & 0.6110 & 0.6504 & 0.7779 & 0.9001 & 0.9963 \\
 $n=4$ & 0.1194 & 0.2582 & 0.3884 & 0.4743 & 0.5869 & 0.6302 & 0.6684 & 0.7909 & 0.9066 & 0.9966 \\
 $n=8$ & 0.1310 & 0.2836 & 0.4175 & 0.5035 & 0.6139 & 0.6556 & 0.6922 & 0.8080 & 0.9150 & 0.9969 \\
 \hline
\end{tabular}

\caption{Maximin ratio as a function of relative support information $a/b$ for various numbers of buyers $n$.}
\label{table:maximin-ratio-intro}
\end{table}

This table provides quantitative evidence that even a small amount of knowledge can lead to nontrivial guarantees on revenue. For example, Table~\ref{table:maximin-ratio-intro} shows that, even when we only know that values can vary over a full order of magnitude ($a/b = 0.10$), we can guarantee $40.38\%$ of the first best with only $2$ buyers. When the knowledge of the scale is more precise, say, if we know the value up to a factor of two ($a/b = 0.50$), we get 74.63\% with 2 buyers. With more agents, the guarantees improve (around 5\% and 3\% more, respectively, for an additional buyer). 

}

\paragraph{Quantifying the power of mechanism features.} We have identified an optimal mechanism that is a randomization over base mechanisms in the $\textnormal{SPA}$ and $\textnormal{POOL}$ classes. A distinguishing feature of the latter  mechanism is that it is \textit{non-standard}, i.e., it allocates to non-highest bidders. We show that this feature is \textit{necessary for optimality} by characterizing the minimax optimal mechanism and performance within the class of all standard mechanisms and showing that the optimal mechanism strictly improves over optimal standard mechanisms.  More broadly, \jaedit{in Section~\ref{sec:other-mech-classes},} we quantify the value of different features in the mechanism class by computing the worst-case $\lambda$-regret (and thus, regret and ratio) for different nested mechanism subclasses of all DSIC mechanisms: all DSIC mechanisms ($\mathcal{M}_{\textnormal{all}}$), all standard mechanisms ($\mathcal{M}_{\textnormal{std}}$), SPA with random reserve ($\mathcal{M}_{\textnormal{SPA-rand}}$), SPA with deterministic reserve ($\mathcal{M}_{\textnormal{SPA-det}}$), and SPA with no reserve ($\mathcal{M}_{\textnormal{SPA-a}}$). These results are also of independent interest, as they characterize the worst-case performance of commonly used mechanisms. In terms of maximin ratio, we find strict separation for all subclasses except $\mathcal{M}_{\textnormal{SPA-det}}$ versus $\mathcal{M}_{\textnormal{SPA-a}}$ (cf. Figure~\ref{fig:ratio-mech-classes-vary-n}). \jaedit{These results show that} introducing some features (such as non-standardness) can lead to significant performance improvements.

\obcomment{Jerry, I think in the above, it is good to refer the results in the main text Cf .ZZZ} \jacomment{I already cite Figure~\ref{fig:ratio-mech-classes-vary-n} which is the figure version of the original table as well as Section~\ref{sec:other-mech-classes} which is }

\jadelete{

We present in Table \ref{table:maximin-ratio-mech-class-intro-n-4} the maximin ratio across mechanism classes and levels of relative support information, and observe that introducing some features (such as non-standardness) can lead to significant performance improvements. \obcomment{In the spirt of shortening the intro and being to the point, Maybe sufficient to highlight separation across classes in intro, and leave actual numbers to the section in the main text?}

\begin{table}[h!]
\centering

\begin{tabular}{c || c | c | c | c | c | c | c  } 
 \hline
 $a/b$ & 0.10 & 0.20 & 0.25 & 0.30 & 0.50 & 0.75 & 0.99 \\ [0.5ex] 
 \hline
all mechanisms & 0.4743 & 0.5869 & 0.6302 & 0.6684 & 0.7909 & 0.9066 & 0.9966 \\
standard mechanisms & 0.4137 & 0.5236 & 0.5684 & 0.6092 & 0.7471 & 0.8853 & 0.9958 \\
SPA with random reserve & 0.3918 & 0.5045 & 0.5517 & 0.5951 & 0.7424 & 0.8849 & 0.9958 \\
SPA with no reserve & 0.3586 & 0.4933 & 0.5457 & 0.5923 & 0.7424 & 0.8849 & 0.9958 \\
 \hline
\end{tabular}

\caption{Maximin ratio for each mechanism class as a function of $a/b$, with $n = 4$ buyers.}
\label{table:maximin-ratio-mech-class-intro-n-4}
\end{table}

}

\subsection{Related Work}

\jacomment{REVIEWER AND AE ASK: Specifically, the review team would like the authors to better compare with \cite{our-paper-ec} by including a technical comparison of saddle techniques and also a more detailed comparison on the insights.}

\paragraph{Auction Design and Mechanism Design}  \cite{Vickrey61}, \cite{Myerson81} and \cite{RileySamuelson81} pioneered a long line of work on the design of auctions and other economic mechanisms with strategic agents. In particular, \cite{Myerson81} shows that if agent valuation distributions are known, i.i.d. and \textit{regular}, then the optimal (expected-revenue-maximizing) mechanism is a \textit{second-price auction with reserve}. This is the classical paradigm of Bayesian mechanism design. \jaedit{However, once we go beyond the simplest settings, this paradigm quickly leads to very complicated ``optimal'' mechanisms that are too detail-dependent and potentially fragile. In response, part of the algorithmic game theory literature instead focuses on proving approximation guarantees for specific ``simple'' mechanisms \citep{RougTalg2019}.  This line of work still assumes that the value distribution is known to both the designer and all players, and the players play a Bayes-Nash equilibrium. In our setting, however, the value distribution is \textit{not} known, the performance is evaluated in the \textit{worst case} rather than the ``Bayesian'' average case, and we require a dominant strategy equilibrium. }

\obcomment{JERRY, don't these two papers also speak about robustness/prior-free... In this case it is awakward to have them at a point when contrast based on the next sentence RESPONSE: \cite{RougTalg2019} only has Bayesian, whereas \cite{Hartline-approx-md-book} has both Bayesian (Chapter 3,4,6,8) and worst-case approximation (Chapter 5,7). I rewrote it as above. I also say ``part of'' the AGT literature rather than just ``AGT literature''}

\paragraph{Robust Mechanism Design} \obedit{The closest line of work to ours is how to robustly sell an item with non-Bayesian uncertainty on valuation distributions. This question has been studied in the ``prior-independent approximation'' literature in algorithmic game theory \citep[Chapter 5]{Hartline-approx-md-book}, often assuming the shape of the distribution (such as regular or monotone hazard rate) is known but not the ``scale'' of the distribution. In the present paper,  we assume we know the scale of the distribution (as captured by the bounds $[a,b]$) but not the shape and we derive both regret and ratio guarantees.}   The one-agent case reduces to a pricing problem; \cite{BergemannSchlag08} and \cite{ErenMaglaras10} provide exact characterization for minimax regret and maximin ratio pricing, respectively. \cite{KocyigitIKW20-distributionally-robust-md,KocyigitRK21-old} analyze minimax regret against any number $n$ of agents whose valuation distributions are arbitrarily correlated with a known upper bound on the support. They show that their problem ``reduces'' to the one-agent case because Nature can choose the worst-case distribution to only have one effective bidder. 

\jaedit{Technically and conceptually, \cite{our-paper-ec} is the closest to our work. They show that the second-price auction is robustly optimal for any number $n$ of agents when only the upper bound of the valuations \jaedit{is} known, whereas our work assumes that both the lower bound $a$ and the upper bound $b$ are known. This allows us to capture the entire spectrum of support information. Whereas both their work and ours share the guess-and-verify saddle point framework, the main difficulty of this framework is to identify the \textit{form} of the optimal mechanism in the first place. This makes our departure from the case $a=0$ challenging: it requires us to explore the space of DSIC mechanisms beyond second price auctions. As soon as one departs from SPAs, the space of mechanisms is much larger and it is not clear what should be a good candidate class a priori. We identify new focal mechanisms (namely, pooling auctions) and show that qualitatively different forms of optimal mechanisms emerge, depending on the amount of information $a/b$. Lastly, they only focus on regret, whereas we unify both regret and ratio objectives in a single framework.

}

\jaedit{Previous works that study robust mechanism design  tend to identify second-price auctions (SPA) as optimal \citep{our-paper-ec,BachrachTalgamCohen22,KocyigitKR20-new,Zhang-auctioning-multiple-goods-without-priors,Zhang22-corr-robust-auction,Che22,AllouahBesbes20}.}  One of the main contributions of this paper is to show why SPA fails to be optimal when we have sufficient relative support information and propose a new \textit{building block} for robust mechanism design, the \textit{pooling auction} mechanism. Other than the fact that the optimal mechanism in our setting is composed of these new mechanisms, this new class  may also be of independent interest in other robust mechanism design problems.  
\obcomment{JERRY, should the above paragraph be connected to the robust mechanisms paragraph?}

\jaedit{
Our work is also related to a broader literature on robustness in mechanism design and contracting \citep{Carroll19-robust-survey}. In particular, while we highlight the work on robustness to distributions here because they are most related to our work, there are other forms of robustness as well, e.g., robustness to higher-order beliefs \citep{BergemannMorris05-robust-md-ecta,BergemannMorris13-robust-md-survey}, robustness to collusion and renegotiation \citep{che-kim-06,che-kim-09,carroll-segal-resale}, and robustness to strategic behavior that is weaker than dominant strategy \citep{ChungEly07,babaioff-undominated,arya2009robust}. Robust mechanism design also has conceptual links to robust and distributionally robust optimization; see \cite{bertsimas-robust-survey} and \cite{dro-review} for overviews.

}

\paragraph{Optimal Mechanisms with Partial Information}  Our work is also related to the design of robustly optimal pricing and mechanisms with partial information about the distribution. Some works assume access to samples drawn from the i.i.d. distribution \citep{ColeRoughgarden14,DhangwatnotaiRoughgardenYan15,AllouahBahamouBesbes22-pricing-with-samples,FengHartlineLi21-revelation-gap-pricing-samples,FuHHK21} while others assume that summary statistics of distributions are known \citep{AzarDMW13,Suzdaltsev20-dr-auction,Suzdaltsev-dr-pricing,BachrachTalgamCohen22,AllouahBahamouBesbes21-pricing-single-point}. \jacomment{Should we move this under the Robust Mechanism Design header also? But that header is already very long. }


\jadelete{\cite{our-paper-ec} show that SPA minimizes worst-case regret for an auction with $n$ bidders and minimal support information for a wide range of distribution classes. \cite{BachrachTalgamCohen22} show that SPA maximizes worst-case revenue for an auction with two i.i.d. bidders when only the mean and the upper bound of the support are known. \cite{KocyigitKR20-new} and \cite{Zhang-auctioning-multiple-goods-without-priors} show that separate SPAs minimize worst-case regret for an auction with multiple goods and multiple bidders when only the upper bounds are known.
\cite{Zhang22-corr-robust-auction} shows that if the seller knows only the marginal distribution of each bidder but not the joint distribution, then SPA maximizes worst-case revenue over all DSIC mechanisms in the case of two bidders and over all standard DSIC mechanisms in the case of $n \geq 3$ bidders. \cite{Che22} assumes only the upper bound and the mean are known but considers a different notion of robustness and shows that SPA is optimal among a class of mechanisms he calls competitive. \cite{AllouahBesbes20} shows that SPA (without reserve) achieves exactly the optimal worst-case fraction of the second best benchmark revenue when there are two bidders and the distribution has monotone hazard rate.

 The upshot of this discussion is that SPA has been the focal candidate for a robust mechanism.
}

\paragraph{Pooling in Auctions} While the specific form of the pooling auction $\pool$ that we propose is new, the more general notion of pooling in auctions has appeared in the literature, starting from the pioneering work on revenue-maximizing auctions of \cite{Myerson81}. When $F$ is not regular, it is shown the distribution must be ``ironed'' such that the all bidders in the same ironing interval have the same allocation, that is, their types are pooled. The main difference is that in \cite{Myerson81}, there is a known focal distribution $F$ to iron, whereas there is no single distribution in our problem, and the pooling emerges naturally from the worst-case analysis over all feasible distributions. \jadelete{The distribution of pooling thresholds depends only on $a,b,n,\lambda$ and not on any specific distribution.\footnote{This is why we name our auctions $\pool$ rather than IRON to distinguish it from standard Myerson ironing. }} In fact, we can see from the \jaedit{proof of the} main theorem that the worst case distributions are all regular. Therefore, \textit{even if we only consider the worst case over all regular distributions, the robustly optimal mechanism will still pool}, whereas if we know the true distribution to be any specific regular distribution, the Bayesian optimal mechanism will not pool by Myerson. The robust auction framework therefore gives qualitatively different prescriptions. \obedit{Beyond \cite{Myerson81},  pooling in auctions has been shown to be optimal in a variety of settings \citep{bergemann-pooling, feldman-lookahead-pooling,laffont1996optimal,pai2014optimal} but to the best of our knowledge, these work all operate in the Bayesian setting and as such the driver of pooling appears different.} 




\jadelete{
Beyond \cite{Myerson81}, there are also some recent works featuring the pooling element. For example, motivated by conflation in digital advertising, \cite{bergemann-pooling} shows that the optimal information disclosure in auctions is to reveal low types and pool high types. The setting is different from ours in that in their setting, the distribution $F$ is known and bidders are Bayesian and the goal is to optimize the \textit{information structure}. Their mechanism  pools high values, whereas our pooling auctions pools low values. \cite{feldman-lookahead-pooling} proposes lookahead auctions with pooling and prove revenue guarantees in the Bayesian setting. Pooling also arises in the design of optimal mechanisms for financially constrained buyers even in the presence of standard regularity conditions on valuations~\citep{laffont1996optimal,pai2014optimal}.
}

\section{Problem Formulation}\label{sec:problem-formulation}

The seller wants to sell an indivisible object to one of $n$ buyers. The $n$ buyers have valuations drawn from a joint cumulative distribution $\mathbf{F}$. The seller does not know $\mathbf{F}$, and only knows a lower bound $a$ and an upper bound $b$ of the valuation of each buyer. That is, the seller only knows that the  support of the buyers' valuations belongs to $[a,b]^{n}$. \obdelete{As discussed in the introduction, the seller chooses a mechanism to minimize the worst-case ``gap'' (either absolute or relative) between the mechanism revenue and benchmark revenue; we will formalize this later in this section.}

\paragraph{Seller's Problem.} We model our problem as a game between the seller and Nature, in which the seller first selects a selling mechanism from a given class $\mathcal{M}$ and then Nature may counter such a mechanism with any distribution from a given class $\mathcal{F}$. Buyers' valuations are then drawn from the distribution chosen by Nature and they participate in the seller's mechanism.

We will now consider the choice of the mechanism class $\mathcal{M}$. A selling mechanism $m = (\mathbf{x}, \mathbf{p})$ is characterized by an allocation rule $\mathbf{x}$ and a payment rule $\mathbf{p}$, where $\mathbf{x}: [a,b]^n \to [0,1]^n$ and $\mathbf{p}: [a,b]^n \to \mathbb{R}$. Given buyers' valuations $\mathbf{v} \in [a,b]^n$, $x_i(\mathbf{v})$ gives the probability that the item is allocated to buyer $i$, and $p_i(\mathbf{v})$ his expected payment to the seller. In our main result, we will consider the class $\mathcal{M}_{\textnormal{all}}$ of all dominant strategy incentive compatible (DSIC) direct mechanisms. A mechanism is DSIC if and only if it is optimal for every buyer to report her true valuation (IR) and participate in the mechanism (IC), regardless of the realization of valuations of the other buyers, and  the seller can allocate at most one item (AC).
More formally, we require that the mechanism $m = (\mathbf{x}, \mathbf{p})$ satisfies the following constraints:
\begin{align*}
    v_i x_i(v_i, \mathbf{v}_{-i}) - p_i(v_i,\mathbf{v}_{-i}) &\geq 0, \quad \forall i, v_i, \mathbf{v}_{-i} &\text{(IR)}\\
    v_i x_i(v_i, \mathbf{v}_{-i}) - p_i(v_i,\mathbf{v}_{-i}) &\geq v_i x_i(\hat{v}_i, \mathbf{v}_{-i}) - p_i(\hat{v}_i,\mathbf{v}_{-i}) \quad \forall i,v_i,\mathbf{v}_{-i},\hat{v}_i &\text{(IC)} \\
    \sum_{i=1}^{n} x_i(v_i,\mathbf{v}_{-i}) &\leq 1 \quad \forall \mathbf{v} \,. &\text{(AC)}
\end{align*}
Note that we allow the seller's mechanism to be randomized. We can now define the class of all DSIC mechanisms
\begin{align}\label{mech-ir-ic-ac}
    \mathcal{M}_{\textnormal{all}} = \left\{ (\mathbf{x},\mathbf{p}): \text{(IR), (IC), (AC)} \right\}.
\end{align}

\paragraph{Seller's Objective.} Informally, the seller seeks to minimize the ``gap'' between the expected revenue $\bE_{\mathbf{v} \sim \mathbf{F} } \left[\sum_{i=1}^{n} p_i(\mathbf{v}) \right]$ relative to the benchmark associated with the revenues that could be collected when the valuations of the buyers are known $\bE_{\mathbf{v} \sim \mathbf{F}} \left[ \max(\mathbf{v}) \right]$.\footnote{The benchmark we use, maximum revenue with known \textit{valuations}, is called the first-best benchmark. Another plausible benchmark we can use is the second-best benchmark, maximum revenue with known \textit{distributions}. Both benchmarks are extensively used in the economics, operations, and computer science literatures; examples of papers using the first-best benchmark include  \citep{BergemannSchlag08,CaldenteyLiuLobel17,robust-monopoly-regulation,kleinberg-yuan}.  The first-best benchmark is also reminiscent of the \textit{offline optimum} benchmark, which is extensively used in the analysis of algorithms \citep{BorodinElYaniv}.} We consider two notions of gaps. First is the \textit{absolute gap}, or \textit{regret}, defined by
\begin{align}
    \text{Regret}(m,\mathbf{F}) &= \bE_{\mathbf{v} \sim \mathbf{F}} \left[ \max(\mathbf{v}) - \sum_{i=1}^{n} p_i(\mathbf{v}) \right]. \label{eq:regret} 
\end{align}

\noindent Second is the \textit{relative gap} or \textit{approximation ratio}, defined by
\begin{align}
    \text{Ratio}(m,\mathbf{F}) &= \frac{\bE_{\mathbf{v} \sim \mathbf{F}} \left[  \sum_{i=1}^{n} p_i(\mathbf{v}) \right]}{\bE_{\mathbf{v} \sim \mathbf{F}} \left[ \max(\mathbf{v})  \right]}.  \label{eq:ratio} 
\end{align}

After the seller chooses a mechanism $m$, Nature then chooses a distribution $\mathbf{F}$ from a given class of distributions $\mathcal{F}$ such that the valuation of the $n$ agents $\mathbf{v} \in \mathbb{R}_{+}^n$ are drawn from $\mathbf{F}$. The seller aims to select the mechanism $m$ to either minimize the worst-case regret or maximize the worst-case approximation ratio. Our goal, therefore, is to characterize the minimax regret and maximin ratio for different classes of mechanisms $\mathcal{M}$ and classes of distributions $\mathcal{F}$:
\begin{align}
  \textnormal{MinimaxRegret}(\mathcal{M},\mathcal{F}) &:=  \inf_{m \in \mathcal{M}} \sup_{\mathbf{F} \in \mathcal{F}} \text{Regret}(m, \mathbf{F})\,, \label{eq:minimaxregret} \\
  \textnormal{MaximinRatio}(\mathcal{M},\mathcal{F}) &:=  \sup_{m \in \mathcal{M}} \inf_{\mathbf{F} \in \mathcal{F}} \text{Ratio}(m, \mathbf{F}).\label{eq:maximinratio}
\end{align}
 
For  $\lambda \in (0,1]$, define the $\lambda$-regret
\begin{align*}
R_\lambda(m,\mathbf{F}) &= \bE_{\mathbf{v} \sim \mathbf{F}} \left[ \lambda \max(\mathbf{v}) - \sum_{i=1}^{n} p_i(\mathbf{v}) \right].
\end{align*}
To unify the minimax regret and maximin ratio objectives, we will focus on the minimax $\lambda$-regret defined by
\begin{align}\label{eqn:minimax-lmbd-regret-def}
R_\lambda(\mathcal{M},\mathcal{F}) &:= \inf_{(x,p) \in \mathcal{M}} \sup_{\mathbf{F} \in \mathcal{F}} R_\lambda(m,\mathbf{F}).
\end{align}
 The following proposition, whose proof is given in Appendix~\ref{app:sec:problem-formulation}, formalizes that the values of problems \eqref{eq:minimaxregret} and \eqref{eq:maximinratio} can be obtained from a characterization of the problem in  \eqref{eqn:minimax-lmbd-regret-def}.

\begin{proposition}\label{prop:minimax-regret-exp}
$\textnormal{MinimaxRegret}(\mathcal{M},\mathcal{F}) = R_1(\mathcal{M},\mathcal{F})$ and $\textnormal{MaximinRatio}(\mathcal{M},\mathcal{F})$ is the largest constant $\lambda \geq 0$ such that $R_\lambda(\mathcal{M},\mathcal{F}) \leq 0$.
\end{proposition}

\jaedit{We note that the traditional notion of worst-case performance can also be obtained from $\textnormal{MaximinRevenue}(\mathcal{M},\mathcal{F}) = -R_0(\mathcal{M},\mathcal{F})$. However, with only support information $[a,b]$, the problem is trivial: without the $F$-dependent benchmark counteracting, Nature will simply put all the weight of the worst-case distribution $F$ at $a$.}

\paragraph{Admissible distributions.} Lastly, we consider the choice of the class of admissible distributions  $\mathcal{F}$. This class can be seen as capturing the ``power'' of Nature: the larger the class, the more powerful/adversarial Nature becomes. To simplify exposition, we will assume for most of the paper that $\mathcal{F}$ is a class of independently and identically distributed (i.i.d.) distributions $\mathcal{F}_{\textnormal{iid}}$, defined formally as follows.
\begin{definition}
The class  $\mathcal{F}_{\textnormal{iid}}$ consists of all distributions 
such that there exists a distribution $F$ with support on $[a,b]$, referred to as the marginal, such that $\mathbf{F}(\mathbf{v}) = \prod_{i=1}^{n} F(v_i)$ for every $\mathbf{v} \in [a,b]^n$.
\end{definition}

\jadelete{
\obcomment{I am not sure why we introduce here the saddle point approach? It seems strange to introudce it in the model section.}
\paragraph{Saddle point approach.}

For each $\mathcal{M}$ and $\mathcal{F}$ considered in this work, we establish the optimality of a mechanism and characterize the associated performance via a saddle point approach. 

\begin{definition}\label{def:saddle-point}
$(m^*,\mathbf{F}^*)$ is a \emph{saddle point} of $R_\lambda(m,\mathbf{F})$ defined in (\ref{eqn:minimax-lmbd-regret-def}) if and only if
\begin{align*}
R_\lambda(m^*,\mathbf{F}) \leq R_\lambda(m^*,\mathbf{F}^*) \leq R_\lambda(m,\mathbf{F}^*) \quad \text{ for all } m \in \mathcal{M}, \mathbf{F} \in \mathcal{F}.
\end{align*}
\end{definition}

Note that if $(m^*,\mathbf{F}^*)$ is a saddle point then $\inf_{m \in \mathcal{M}} \sup_{\mathbf{F} \in \mathcal{F}} R_\lambda(m,\mathbf{F}) \leq \sup_{\mathbf{F} \in \mathcal{F}} R_\lambda(m^*,\mathbf{F}) \leq  R_\lambda(m^*,\mathbf{F}^*) \leq \inf_{m \in \mathcal{M}} R_\lambda(m,\mathbf{F}^*) \leq \inf_{m \in \mathcal{M}} \sup_{\mathbf{F} \in \mathcal{F}} R_\lambda(m,\mathbf{F})$. Therefore, $R_\lambda(m^*,\mathbf{F}^*) = \inf_{m \in \mathcal{M}} \allowbreak \sup_{\mathbf{F} \in \mathcal{F}} R_\lambda(m,\mathbf{F})$ is the minimax $\lambda$-regret (the optimal performance), $m^*$ is an optimal mechanism, and $\mathbf{F}^*$ is a corresponding worst-case distribution. Therefore, it is sufficient to exhibit a saddle point and verify the inequalities in Definition~\ref{def:saddle-point} to obtain an optimal mechanism and its associated performance.
}

\section{Optimal mechanisms over the class of all DSIC mechanisms}\label{sec:main-all-mech}

In this section, we characterize \jaedit{an} optimal mechanism over the class of all DSIC mechanisms $\mathcal{M}_{\textnormal{all}}$ for the minimax $\lambda$-regret problem against i.i.d. distributions $\mathcal{F}_{\textnormal{iid}}$ for any $\lambda \in (0,1]$, support information $[a,b]$, and number of bidders $n$. Our main theorem presents an optimal mechanism for each uncertainty regime. We will then use the main theorem to gain insights into the structure and performance of the optimal mechanism.

\jadelete{As discussed in the introduction, we consider the mechanism classes $\spa$ and $\pool$, which are defined formally as follows.}

\jaedit{Let $v^{(1)}$ and $v^{(2)}$ be the highest and second-highest entry of $\mathbf{v}$ and let  $k$ denote the number of buyers with value equal to $v^{(1)}$.}
\jaedit{We define the $\spa$ and $\pool$ mechanism classes as follows.}

\begin{definition}[second-price and pooling auctions]\label{def:thres-mech-default} 


A \emph{second-price auction} with reserve $r$, denoted $\spa(r)$, is defined by the allocation rule $x: [a,b]^n \to [0,1]^n$ \jaedit{and the payment rule $p: [a,b]^n \to [0,1]^n$} given by, for each $i \in [n]$,
\jaedit{
\begin{align*}
x_i(\mathbf{v}) = \begin{cases}
\frac{\1(v_i = v^{(1)})}{k}  &\text{ if } v^{(1)} \geq r \,, 
\\
0 &\text{ if }  v^{(1)} < r \,,
\end{cases}
\quad \textnormal{ and } \quad 
p_i(\mathbf{v}) = \begin{cases}
v^{(2)} \frac{\1(v_i = v^{(1)})}{k}  &\text{ if } v^{(2)} \geq r \,, 
\\
r \frac{\1(v_i = v^{(1)})}{k}  &\text{ if } v^{(1)} \geq r > v^{(2)} \,, 
\\
0 &\text{ if }  v^{(1)} < r \, ,
\end{cases}
\end{align*}
}%
\jaedit{Given a distribution of reserve prices $\Phi$, we denote by $\spa(\Phi)$ a second-price auction with random reserve prices drawn from $\Phi$.}

A \emph{pooling auction} with threshold $\tau$, denoted $\pool(\tau)$, is defined by the allocation rule $x: [a,b]^n \to [0,1]^n$ \jaedit{and the payment rule $p: [a,b]^n \to [0,1]^n$} given by, for each $i \in [n]$,
\jaedit{
\begin{align*}
x_i(\mathbf{v}) = \begin{cases}
\frac{\1(v_i = v^{(1)})}{k}  &\text{ if } v^{(1)} \geq r \,, 
\\
\frac{1}{n} &\text{ if }  v^{(1)} < r \,,
\end{cases}
\quad \textnormal{ and } \quad 
p_i(\mathbf{v}) = \begin{cases}
v^{(2)} \frac{\1(v_i = v^{(1)})}{k}  &\text{ if } v^{(2)} \geq r \,, 
\\
\frac{(n-1)\tau + a}{n } \frac{\1(v_i = v^{(1)})}{k}  &\text{ if } v^{(1)} \geq r > v^{(2)} \,, 
\\
\frac{a}{n} &\text{ if }  v^{(1)} < r \, 
\end{cases}
\end{align*}
}%
\jaedit{Given a distribution of thresholds $\Psi$, we denote by $\pool(\Psi)$ a pooling auction with random thresholds drawn from $\Psi$.}
\end{definition}


\jaedit{
Note that in any DSIC mechanism (including $\spa$ and $\pool$ defined above), the payment rule is uniquely determined from the allocation rule via Myerson's envelope formula. 
In a pooling auction, we can still use the threshold to differentiate between bidders with different values and potentially extract more revenue, without risking the zero payoff that comes from not allocating the item. Of course, this has implications for payments.  We illustrate this interplay in 
 Figure~\ref{fig:spa-pool-alloc-pay}, where we depict, for the case of two agents, the allocation rule $x(\mathbf{v})$ and revenue $p_1(\mathbf{v})+p_2(\mathbf{v})$ at each valuation vector $\mathbf{v} = (v_1,v_2)$ for three mechanisms:  $\textnormal{SPA}$ (no reserve),  $\textnormal{SPA}(r)$, and  $\textnormal{POOL}(r)$. 

 We can see intuitively that pooling low types indeed softens the tradeoff associated with reserve pricing.  By increasing the allocation for the low types, we increase the payment accrued from lower-value bidders but, at the same time, we decrease the payment accrued from higher-value bidders to guarantee  incentive compatibility (so higher-value bidders do not pretend to be lower-value ones). When the relative support information is high (i.e., $a \sim b$), the lower-value and higher-value bidders are not too different, and this softer tradeoff has the potential to lead to  more robust mechanisms. 
 \begin{figure}[h!]
	\begin{subfigure}{0.32\linewidth}
    	\centering
    	\begin{tikzpicture}[>=triangle 45,xscale=3.5,yscale=3.5]
\draw[->] (0,0) -- (1.2,0) node[below] {$v_1$};
\draw[->] (0,0) -- (0,1.2) node[left] {$v_2$};
\draw[thick,-] (1,0)--(1,1);
\draw[thick,-] (0,1)--(1,1);
\node[below] at (1,0){$b$};
\node[left,below] at (0,0){$a$};
\node[left] at (0,1){$b$};

\draw[thick,-] (0,0) --(1,1);
\node at (0.7,0.35){P1};
\node at (0.35,0.7){P2};

\end{tikzpicture}
    	\caption{$\textnormal{SPA}(a)$ allocation rule}
    \end{subfigure}%
	\begin{subfigure}{0.32\linewidth}
    	\centering
    	\begin{tikzpicture}[>=triangle 45,xscale=3.5,yscale=3.5]
\draw[->] (0,0) -- (1.2,0) node[below] {$v_1$};
\draw[->] (0,0) -- (0,1.2) node[left] {$v_2$};
\draw[thick,-] (1,0)--(1,1);
\draw[thick,-] (0,1)--(1,1);
\node[below] at (0.6,0){$r$};
\node[below] at (1,0){$b$};
\node[left,below] at (0,0){$a$};
\node[left] at (0,0.6){$r$};
\node[left] at (0,1){$b$};

\draw[thick,-] (0.6,0)--(0.6,0.6);
\draw[thick,-] (0,0.6)--(0.6,0.6);
\draw[thick,-] (0.6,0.6) --(1,1);
\node at (0.8,0.35){P1};
\node at (0.35,0.8){P2};
\node at (0.3,0.3){P$\emptyset$};

\end{tikzpicture}
    	\caption{$\textnormal{SPA}(r)$ allocation rule}
    \end{subfigure}%
	\begin{subfigure}{0.32\linewidth}
    	\centering
    	\begin{tikzpicture}[>=triangle 45,xscale=3.5,yscale=3.5]
\draw[->] (0,0) -- (1.2,0) node[below] {$v_1$};
\draw[->] (0,0) -- (0,1.2) node[left] {$v_2$};
\draw[thick,-] (1,0)--(1,1);
\draw[thick,-] (0,1)--(1,1);
\node[below] at (0.6,0){$r$};
\node[below] at (1,0){$b$};
\node[left,below] at (0,0){$a$};
\node[left] at (0,0.6){$r$};
\node[left] at (0,1){$b$};

\draw[thick,-] (0.6,0)--(0.6,0.6);
\draw[thick,-] (0,0.6)--(0.6,0.6);
\draw[thick,-] (0.6,0.6) --(1,1);
\node at (0.8,0.35){P1};
\node at (0.35,0.8){P2};
\node at (0.35,0.3){$\small \begin{cases} \textnormal{P1} \text{ w.p.} \frac{1}{2} \\ \textnormal{P2} \text{ w.p.} \frac{1}{2} \end{cases}$};

\end{tikzpicture}
    	\caption{$\textnormal{POOL}(r)$ allocation rule}
    \end{subfigure}%
    
	\begin{subfigure}{0.32\linewidth}
    	\centering
    	\begin{tikzpicture}[>=triangle 45,xscale=3.5,yscale=3.5]
\draw[->] (0,0) -- (1.2,0) node[below] {$v_1$};
\draw[->] (0,0) -- (0,1.2) node[left] {$v_2$};
\draw[thick,-] (1,0)--(1,1);
\draw[thick,-] (0,1)--(1,1);
\node[below] at (1,0){$b$};
\node[left,below] at (0,0){$a$};
\node[left] at (0,1){$b$};

\draw[thick,-] (0,0) --(1,1);
\node at (0.7,0.35){$v_2$};
\node at (0.35,0.7){$v_1$};

\end{tikzpicture}
    	\caption{$\textnormal{SPA}(a)$ revenue}
    \end{subfigure}%
	\begin{subfigure}{0.32\linewidth}
    	\centering
    	\begin{tikzpicture}[>=triangle 45,xscale=3.5,yscale=3.5]
\draw[->] (0,0) -- (1.2,0) node[below] {$v_1$};
\draw[->] (0,0) -- (0,1.2) node[left] {$v_2$};
\draw[thick,-] (1,0)--(1,1);
\draw[thick,-] (0,1)--(1,1);
\node[below] at (0.6,0){$r$};
\node[below] at (1,0){$b$};
\node[left,below] at (0,0){$a$};
\node[left] at (0,0.6){$r$};
\node[left] at (0,1){$b$};

\draw[thick,-] (0.6,0)--(0.6,0.6);
\draw[thick,-] (0,0.6)--(0.6,0.6);
\draw[thick,-] (0.6,0.6) --(1,1);
\draw[thick,-] (0.6,0.6)--(1,0.6);
\draw[thick,-] (0.6,0.6)--(0.6,1);

\node at (0.3,0.3){0};
\node at (0.3,0.8){$r$};
\node at (0.8,0.3){$r$};
\node at (0.85,0.7){$v_2$};
\node at (0.7,0.85){$v_1$};

\end{tikzpicture}
    	\caption{$\textnormal{SPA}(r)$ revenue}
    \end{subfigure}%
	\begin{subfigure}{0.32\linewidth}
    	\centering
    	\begin{tikzpicture}[>=triangle 45,xscale=3.5,yscale=3.5]
\draw[->] (0,0) -- (1.2,0) node[below] {$v_1$};
\draw[->] (0,0) -- (0,1.2) node[left] {$v_2$};
\draw[thick,-] (1,0)--(1,1);
\draw[thick,-] (0,1)--(1,1);
\node[below] at (0.6,0){$r$};
\node[below] at (1,0){$b$};
\node[left,below] at (0,0){$a$};
\node[left] at (0,0.6){$r$};
\node[left] at (0,1){$b$};

\draw[thick,-] (0.6,0)--(0.6,0.6);
\draw[thick,-] (0,0.6)--(0.6,0.6);
\draw[thick,-] (0.6,0.6) --(1,1);
\draw[thick,-] (0.6,0.6)--(1,0.6);
\draw[thick,-] (0.6,0.6)--(0.6,1);

\node at (0.3,0.3){$a$};
\node at (0.3,0.8){$\frac{r+a}{2}$};
\node at (0.8,0.3){$\frac{r+a}{2}$};
\node at (0.85,0.7){$v_2$};
\node at (0.7,0.85){$v_1$};


\end{tikzpicture}
    	\caption{$\textnormal{POOL}(r)$ revenue}
    \end{subfigure}%
	\caption{Allocation rules and revenue of SPA without reserve, $\textnormal{SPA}(r)$, and $\textnormal{POOL}(r)$. In the allocation rule, P1 stands for allocating to player 1, P2 for allocating to player 2, and P$\emptyset$ for not allocating.}
\label{fig:spa-pool-alloc-pay}
\end{figure}
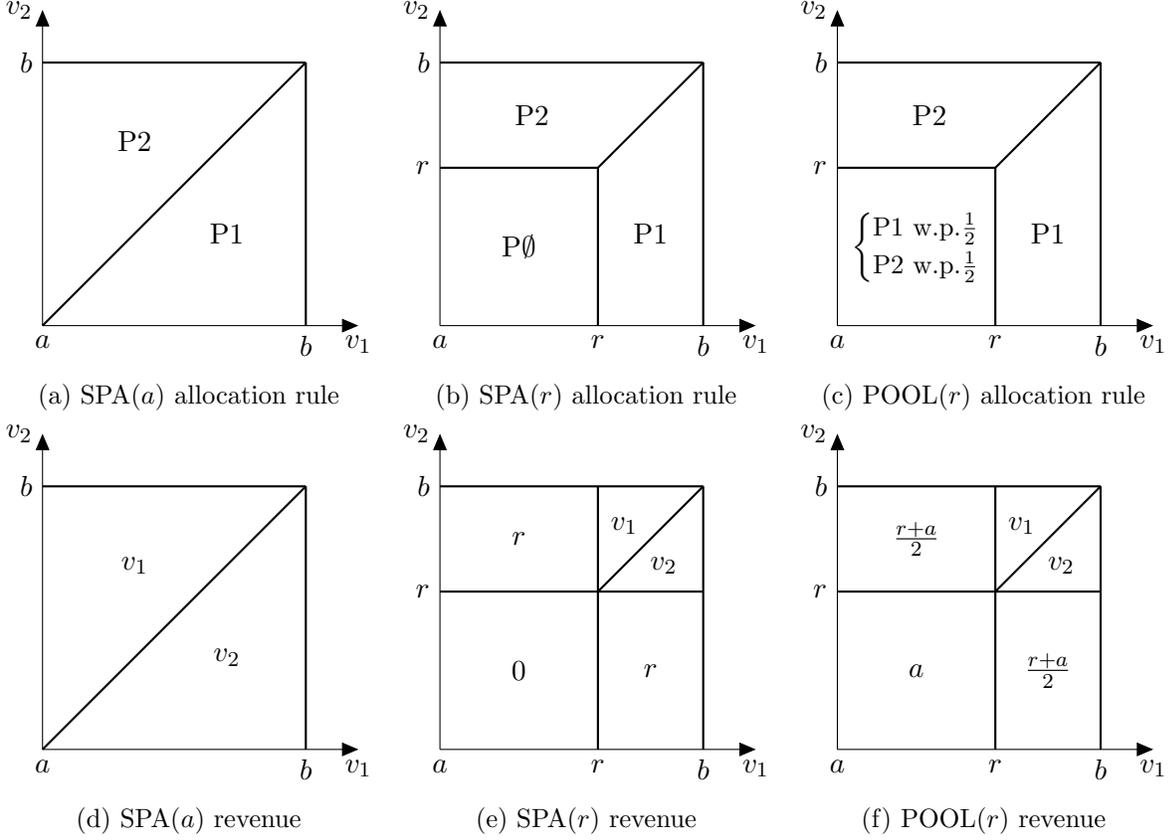

}


\jaedit{Before we state our main theorem, we first define two thresholds that will demarcate three information regimes arising in our analysis.

\begin{definition}
Fix $n$ and $\lambda \in (0,1]$. Define $k_l \in (0,1)$ as a unique solution to 
\begin{align*}
\lambda \int_{t=k_l}^{t=1} \frac{(t-k_l)^{n-1}}{t^n} dt =  (1-k_l)^{n-1}.
\end{align*}
Define $k_h$ to be $1$ for $n = 1$, and if $n \ge 2$, define $k_h \in (0,1)$ to be the  unique solution to
\begin{align*}
\int_{t=k_h}^{t=1} \left[ \frac{(t-k_h)^{n-1}}{t^n} - (1-\lambda) \frac{(t-k_h)^{n}}{t^{n+1}} \right] dt = (1-k_h)^{n}
\,.
\end{align*}   
\end{definition}
}

\jaedit{We are now ready to state our main theorem.}

\jaedit{
\begin{theorem}[Main Theorem]\label{thm:char-all-mech-full-main}
Depending on the value of $a/b$, the minimax $\lambda$-regret problem admits the following mechanism as robustly optimal.

\begin{itemize}
\item Suppose $a/b \leq k_l$, and let $r^* = k_l b$. An optimal mechanism is $\spa(\Phi^*)$ with
\begin{align*}
\Phi^*(v) = \lambda  \frac{v^{n-1}}{(v-r^*)^{n-1}} \int_{t=r^*}^{t=v} \frac{(t-r^*)^{n-1}}{t^n} dt \, .
\end{align*}
\item Suppose $a/b \geq k_h$, and let $\phi_0 = (a-k_h b)/(1-k_h) $. An optimal mechanism is $\pool(\Psi^*)$ with
\begin{align*}
    \Psi^*(v) = \frac{n\lambda}{n-1} \left( \frac{v-\phi_0}{v-a} \right)^{n} \int_{t=a}^{t=v} \frac{(t-a)^{n}}{(t-\phi_0)^{n+1}} dt\, .
\end{align*}
\item Suppose $k_l \leq a/b \leq k_h$. Let  $(v^*,\alpha)$ be the unique solution to 
\begin{align*}
\frac{(v^*-a)^{n-1}}{(v^*)^{n-1}} ( 1 - n \alpha ) &= \lambda \int_{t=a}^{t=v^*} \frac{(t-a)^{n-1}}{t^n} dt\,, \\
\frac{(b-a)^{n}}{b^{n}}  - \frac{(v^*-a)^{n}}{(v^*)^{n}} (1 - (n-1)\alpha) &= \int_{t=v^*}^{t=b} \left[ \frac{(t-a)^{n-1}}{t^n} - (1-\lambda) \frac{(t-a)^{n}}{t^{n+1}}  \right] dt\,.
\end{align*}
An optimal mechanism is based on a \emph{unified threshold} distribution $\mathcal{D}$ such that if we draw a sample $r \sim \mathcal{D}$, if $r \leq v^*$ the mechanism is $\spa(r)$, otherwise the mechanism is $\pool(r)$, where the CDF of $\mathcal{D}$ is given by
\begin{align*}
    \mathcal{D}(v) = \begin{cases}
\lambda \left( \frac{v}{v-a} \right)^{n-1} \int_{t=a}^{t=v} \frac{(t-a)^{n-1}}{t^{n}} dt
&\text{ for } v \in [a,v^*]\,, \\
-\frac{1}{n-1} +  \frac{n}{n-1} \left( \frac{v}{v-a} \right)^{n} \left[ \left( \frac{b-a}{b} \right)^{n} - \int_{t=v}^{t=b} \left[ \frac{(t-a)^{n-1}}{t^n} - (1-\lambda) \frac{(t-a)^n}{t^{n+1}}  \right]dt   \right] &\text{ for } v \in [v^*,b]\,.
\end{cases}
\end{align*}
\end{itemize}

\end{theorem}

We present the proof of Theorem \ref{thm:char-all-mech-full-main} in \S\ref{subsec:proof-main-thm}, with some computational details deferred to Appendix~\ref{app:sec:main-all-mech}. We then use our closed form characterization to gain insights into the structure and performance of the optimal mechanism in \S\ref{subsec:structure-mech-all} and show how our results reduce to the known pricing case with $n = 1$ in \S\ref{subsec:n-1-remark}.

}

\subsection{Proof of the Main Theorem}\label{subsec:proof-main-thm}

\jaedit{
The key idea for the proof of our main theorem (Theorem~\ref{thm:char-all-mech-full-main}) is to explicitly exhibit a saddle point of the zero-sum game between seller and Nature, defined as follows.

\begin{definition}\label{def:saddle-point}
$(m^*,\mathbf{F}^*)$ is a \emph{saddle point} of $R_\lambda(m,\mathbf{F})$ defined in (\ref{eqn:minimax-lmbd-regret-def}) if and only if
\begin{align*}
R_\lambda(m^*,\mathbf{F}) \leq R_\lambda(m^*,\mathbf{F}^*) \leq R_\lambda(m,\mathbf{F}^*) \quad \text{ for all } m \in \mathcal{M}, \mathbf{F} \in \mathcal{F}.
\end{align*}
\end{definition}

As discussed in the related work section, the saddle point guess-and-verify technique has been used before in robust auction design problems \citep{our-paper-ec,BachrachTalgamCohen22}, but these all identify second-price auctions as robustly optimal. \sbdelete{If we assumed that $m^*$ is $\spa(\Phi)$ and tried to derive conditions on the reserve distribution $\Phi$, we would get a contradiction, suggesting that $\spa$ is not the right mechanism class. We later formalize this in \S\ref{sec:other-mech-classes} by showing that SPAs are strictly suboptimal and precisely quantify the separation.} Because the space of DSIC mechanisms is so big and unstructured, it is unclear a priori what the candidate  mechanism family should be. Part of our technical contribution is identifying the right mechanism class depending on the support information regime. Importantly, this mechanism class is parameterized by a ``one-dimensional distribution'' (e.g. the $\Phi$ in $\spa(\Phi)$ or $\Psi$ in $\pool(\Psi)$ or a ``unified threshold distribution'' combining both $\spa$ and $\pool$), which enables us to apply first-order condition techniques to establish the saddle point conditions.

}

\sbedit{
We proceed in three steps. Firstly, we introduce a class of $(g_u,g_d)$ mechanisms, which unifies the optimal mechanisms in all three regimes (second-price auctions with random reserves, pooling auctions with random thresholds, and randomization between $\spa$ and $\pool$). Secondly, we derive sufficient conditions for a particular $(g_u,g_d)$ mechanism and a particular distribution $F^*$ to be a saddle point. The proof of the main theorem is then reduced to checking these sufficient conditions, which we defer to the Appendix. Thirdly, we illustrate how we can use our framework to guess and verify an optimal mechanism in the high information regime. While this step is not strictly necessary for verifying the optimality of a candidate saddle point, it offers valuable insights into the construction of our mechanism. Such intuition may prove beneficial for future researchers exploring related robust mechanism design problems.
}

\jadelete{

\subsection{Self-Contained Proof of the High Support Information Regime}\label{subsec:proof-high-info}

Assume that $\mathbf{F}$ is i.i.d. with marginal $F$ \obedit{and suppose $a/b \geq k_h$. The proof proceeds in two steps. We first identify a candidate pair  $(\Psi^*,F^*)$ to form a saddle point. Then we proceed with a verification argument.}

\paragraph{Identifying a candidate saddle point} We can compute the $\lambda$-regret $R_\lambda(\Psi,F)$ of $\pool(\Psi)$ and show that it has two representations using an ad-hoc integration by parts result that handles distributions that are not necessarily smooth (cf. Proposition~\ref{prop:reg-exp-g}). Suppose that $\Psi$ is arbitrary, while $F$ has a density in the interior with point masses $F(a)$ and $f_b = 1-F(b^-)$ at $a$ and $b$, then
    \begin{align*}
R_\lambda(\Psi,F)&= \lambda b - a  F(a)^n - b + b(1-f_b)^n \\
&+ \int_{v=a}^{v=b} \Bigg[ \left\{- \lambda F(v)^n + n F(v)^{n-2} F'(v) (v-a) - n F(v)^{n-1} v F'(v) \right\} \\
&+ \underbrace{n F(v)^{n-2} \big\{F(v) - F(v)^2 - (v-a) F'(v)\big\}}_{\text{coefficient of $\Psi(v)$}}  \mathbf{\color{blue} \Psi(v)} \Bigg] dv \, .
\tag{Regret-$\Psi$}  \label{eqn:regret-psi-main}
\end{align*}

Note that (\ref{eqn:regret-psi-main}) depends on $\Psi$ only through $\mathbf{\color{blue} \Psi(v)}$ and is linear in $\Psi$. This representation is useful for the seller's saddle $\inf_m R_\lambda(m,F^*)$, because we do not make any assumption on the seller's choice $\Psi$. By the first-order conditions, under the worst-case distribution $F^*$, the coefficient of each $\Psi^*(v)$ is zero wherever \jaedit{$\Psi^*(v) > 0$}. \obcomment{Is it obvious that there is always a feasible deviation? What if $\psi^*$ corresponds to a mass point?} Otherwise, the seller could decrease his regret by changing the distribution of reserves. Therefore,
\begin{align*}
    F^*(v) - F^*(v)^2 - (v-a) (F^*)'(v) = 0 \Rightarrow \frac{d}{dv} \left( v - \frac{v-a}{F^*(v)} \right) = 0 \Rightarrow v - \frac{v-a}{F^*(v)} = \phi_0\,.
\end{align*}
This pins down Nature's candidate distribution as $F^*(v) = (v-a)/(v-\phi_0)$, a distribution with constant virtual value $\phi_0$. \jaedit{We note that this part of the argument (guessing $F^*$) is technically not needed in the formal proof, because it is sufficient to simply verify the saddle point to confirm that our mechanism is robustly optimal. Therefore, we can make certain smoothness assumptions to derive this $F^*$ which is \textit{not needed} for the formal saddle verification proof. Nevertheless, we think it is useful for the readers to understand where our candidate saddle point comes from. From this point onward, we do not need any extra assumptions.}

Alternatively, let $F$ be arbitrary but assume that $\Psi$ is differentiable, then $R_\lambda(\Psi,F) $ can be written as
\begin{align*}
R_\lambda(\Psi,F) &= a(\lambda-1) + \int_{v=a}^{v=b} \left[ \lambda - \Psi(v) - (v-a) \Psi'(v) \right] dv  \\
&+  \int_{v=a}^{v=b} \underbrace{\left( -\lambda - (n-1) \Psi(v) \right) \mathbf{\color{red} F(v)}^n +   \left( n \Psi(v) + (v-a) \Psi'(v) \right)   \mathbf{\color{red} F(v)}^{n-1}}_{:=I(F(v),v)}  dv \, . \tag{Regret-$F$} \label{eqn:regret-F-main}
\end{align*}

Note that (\ref{eqn:regret-F-main}) depends on $F$ only through $\mathbf{\color{red} F(v)}$ and is ``separable'' (different $F(v)$ terms do not interact). This representation is useful for Nature's saddle $\sup_F R_\lambda(m^*,F)$, because we do not make any assumption on Nature's choice $F$ and we can optimize the integrand pointwise as an expression in $F(v)$, i.e., we can equivalently solve $\sup_{F(v)} I(F(v),v)$ for each $v$. Taking derivatives with respect to $F(v)$, the first-order condition on $F(v)$ gives 
\begin{align*}
    \left( -\lambda - (n-1) \Psi(v) \right) n F(v)^{n-1} +   \left( n \Psi(v) + (v-a) \Psi'(v) \right) (n-1) F(v)^{n-2} = 0\,,
\end{align*}
because, at the saddle, this first derivative must be zero. Substituting $F^*(v) = (v-a)/(v-\phi_0)$ gives
\begin{align*}
    (\Psi^*)'(v) + \frac{n(a-\phi_0)}{(v-\phi_0)(v-a)} \Psi^*(v) = \frac{n \lambda}{(n-1)(v-\phi_0)} \, , \tag{ODE-$\Psi$} \label{eqn:ode-psi-main}
\end{align*}
or
\begin{align*}
    \frac{d}{dv} \left[ \left( \frac{v-a}{v-\phi_0} \right)^{n} \Psi^*(v) \right] = \frac{n\lambda}{n-1} \frac{(v-a)^{n}}{(v-\phi_0)^{n+1}} \, .
\end{align*}

Because $\left( \frac{v-a}{v-\phi_0} \right)^{n} \Psi^*(v)$ is zero at $v = a$, integrating from $a$ to $v$ gives
\begin{align*}
    \left( \frac{v-a}{v-\phi_0} \right)^{n} \Psi^*(v) = \frac{n\lambda}{n-1} \int_{t=a}^{t=v} \frac{(t-a)^{n}}{(t-\phi_0)^{n+1}} dt \, .
\end{align*}
Performing the integration, we can rewrite $\Psi^*(v)$ as
\begin{align*}
    \Psi^*(v) = \frac{n\lambda}{n-1} \sum_{k=n+1}^{\infty} \frac{(v-a)^{k-n}}{k(v-\phi_0)^{k-n}} \, .
\end{align*}
This expression makes it clear that $\Psi^*(a) = 0$, so $\Psi^*$ is well-behaved and does not have a point mass at $a$, and that $\Psi^*(v)$ is increasing in $v$. 
For $\Psi^*$ to be feasible as a CDF, the only additional condition we need is  $\Psi^*(b) = 1$, which gives the following equation that determines $\phi_0$:
\begin{align*}
    \left( \frac{b-a}{b-\phi_0} \right)^{n}  = \frac{n\lambda}{n-1} \int_{t=a}^{t=b} \frac{(t-a)^{n}}{(t-\phi_0)^{n+1}} dt\,.
\end{align*}
By definition of $k_h$, we see by inspection that this equation has an explicit solution
\begin{align*}
    \phi_0 = \frac{a - k_h b}{1-k_h }\,.
\end{align*}
We need $\phi_0 \geq 0$ for Nature's saddle to hold: this is why we need $a/b \geq k_h$ in the pure $\pool$ regime. 

\paragraph{Verification argument} Now that we have identified a candidate saddle-point, we need to formally verify its optimality. Seller's saddle (optimize over $m$ with fixed $F^*$) is a standard Bayesian mechanism design problem and the optimality of $\pool(\Psi^*)$ follows because every mechanism that always allocates is optimal  since $F^*$ has positive, constant virtual value. For Nature's saddle (optimize over $F$ with fixed $m^*$), we need to check that $F^*$ maximizes the regret given the mechanism $\pool(\Psi^*)$. It is sufficient to show that $F^*(v)$ globally maximizes the integrand $I(F(v),v)$, as a function of $F(v)$, pointwise. The expression $I(F(v),v)$ is polynomial in $F(v)$ with two nonzero degree terms $F(v)^{n-1}$ and $F(v)^{n}$. Two conditions are sufficient to guarantee global optimality: (i) the first-order optimality condition (FOC), (ii) nonnegativity of the coefficient of $F(v)^{n-1}$. To see why this is true, consider the expression $I(F(v),v) := \alpha(v) F(v)^{n-1} - \beta(v) F(v)^{n}$ as a function of $F(v)$. The first derivative of this expression is $\partial I/\partial F = F(v)^{n-2} ( (n-1) \alpha(v) - n \beta(v) F(v))$, and the sign of $\partial I/\partial F$ is the same as the sign of $(n-1) \alpha(v) - n \beta(v) F(v)$. From (i), FOC is equivalent to $(n-1) \alpha(v) - n \beta(v) F^*(v) = 0$, so if we have $\alpha(v) \geq 0$ from (ii), then we also have $\beta (v) \geq 0$. We can immediately see that $I(F(v),v)$ is unimodal with a peak at $F^*(v)$ (increasing in $F(v)$ for $F(v) < F^*(v)$ and decreasing in $F(v)$ for $F(v) > F^*(v)$), justifying our claim.

We conclude by checking that the coefficient of $F(v)^{n-1}$ is nonnegative:
\begin{align*}
    n \Psi^*(v) + (v-a) (\Psi^*)'(v) \geq 0 \, .
\end{align*}
Substituting the expression for $(\Psi^*)'(v)$ from (\ref{eqn:ode-psi-main}) reduces the above inequality to 
\begin{align*}
    \frac{n(v-a)}{(v-\phi_0)} \Psi^*(v) + \frac{n\lambda}{n-1} \frac{(v-a)}{(v-\phi_0)} \geq 0\,,
\end{align*}
which is trivially true. We therefore have established a saddle point $(\pool(\Psi^*),F^*)$ of the problem in the high relative support information case.

\jadelete{We give a more detailed outline of \textit{all cases} in the proof of the main theorem in Appendix~\ref{app:sec:proof-main-thm-outline}. Details omitted from Appendix~\ref{app:sec:proof-main-thm-outline} are given in Appendix~\ref{app:sec:proof-main-thm-detail}. }

}

\jaedit{



\subsubsection{A Unifying Class of Mechanisms}

We now introduce the class of $(g_u,g_d)$ mechanisms, parameterized by two functions $g_u,g_d$.

\begin{definition}[\emph{$(g_u,g_d)$} mechanisms]\label{def:gu-gd} 
Let $g_u,g_d: [a,b] \to [0,1]$ be given functions.  A mechanism  \emph{$(g_u,g_d)$} is defined by the allocation rule $x: [a,b]^n \to [0,1]^n$ given by, for each $i \in [n]$,
\begin{align*}
x_i(\mathbf{v}) = \begin{cases}
\frac{1}{k} g_u(v_{\max}) + \frac{k-1}{k} g_d(v_{\max}) &\text{ if } v_i = \max(\mathbf{v}) := v_{\max} \text{  and there are $k$ entries in $\mathbf{v}$ equal to $v_{\max}$} \\
g_d(v_{\max}) &\text{ if } v_i < \max(\mathbf{v}) := v_{\max}
\end{cases}
\end{align*} 
and the payment rule $p: [a,b]^n \to \mathbb{R}_{++}^n$ is determined uniquely from  Myerson's formula such that the resulting mechanism $(x,p)$ is dominant strategy incentive compatible. 
\end{definition}

 \begin{figure}[t!]
 \centering
 	\begin{subfigure}{0.32\linewidth}
		\centering
  \begin{center}
    	\begin{tikzpicture}[>=triangle 45,xscale=3.5,yscale=3.5]
\draw[->] (0,0) -- (1.2,0) node[below] {$v_1$};
\draw[->] (0,0) -- (0,1.2) node[left] {$v_2$};
\draw[thick,-] (1,0)--(1,1);
\draw[thick,-] (0,1)--(1,1);
\node[below] at (1,0){$b$};
\node[below] at (0.7,0){$v_1$};
\node[left,below] at (0,0){$a$};
\node[left] at (0,1){$b$};
\node[left] at (0,0.7){$v_2$};

\draw[thick,-] (0,0) --(1,1);
\node at (0.7,0.3){$g_u(v_1)$};
\node at (0.3,0.7){$g_d(v_2)$};
\draw[dashed,-] (0,0.7) --(0.14,0.7);
\draw[dashed,-] (0.45,0.7) --(0.7,0.7);
\draw[dashed,-] (0.7,0.7) --(0.7,0.35);
\draw[dashed,-] (0.7,0.25) --(0.7,0.0);

\node at (0.5,1.1){$x_1(v_1,v_2)$};

\end{tikzpicture}
  \end{center}
	\end{subfigure}
	\begin{subfigure}{0.32\linewidth}
		\centering 
  \begin{center}
    	\begin{tikzpicture}[>=triangle 45,xscale=3.5,yscale=3.5]
\draw[->] (0,0) -- (1.2,0) node[below] {$v_1$};
\draw[->] (0,0) -- (0,1.2) node[left] {$v_2$};
\draw[thick,-] (1,0)--(1,1);
\draw[thick,-] (0,1)--(1,1);
\node[below] at (1,0){$b$};
\node[below] at (0.7,0){$v_1$};
\node[left,below] at (0,0){$a$};
\node[left] at (0,1){$b$};
\node[left] at (0,0.7){$v_2$};

\draw[thick,-] (0,0) --(1,1);
\node at (0.7,0.3){$g_d(v_1)$};
\node at (0.3,0.7){$g_u(v_2)$};
\draw[dashed,-] (0,0.7) --(0.14,0.7);
\draw[dashed,-] (0.45,0.7) --(0.7,0.7);
\draw[dashed,-] (0.7,0.7) --(0.7,0.35);
\draw[dashed,-] (0.7,0.25) --(0.7,0.0);
\node at (0.5,1.1){$x_2(v_1,v_2)$};
\end{tikzpicture}
  \end{center}
	\end{subfigure}

  \caption{Allocation rules of a $(g_u,g_d)$ mechanism for $n = 2$. For the first buyer, we have $x_1(v_1,v_2) = g_u(v_1)$ if $v_1 > v_2$, and $x_1(v_1,v_2) = g_d(v_2)$ if $v_1 < v_2$. When $v_1 = v_2$ the allocation is not shown but equal to $x_1(v_1,v_2) = (g_u(v_1) + g_d(v_1))/2$.}
  \label{fig:x1-of-gu-gd}
  \end{figure}
  
In other words, the mechanism allocates $g_u(v_{\max})$ to the highest bidder(s) and $g_d(v_{\max})$ to the non-highest bidder(s). If there are $k$ highest bidders, the mechanism breaks ties symmetrically by selecting one of them to be the ``winner'' with $g_u$ allocation uniformly at random.  The $(g_u,g_d)$ mechanism class is a powerful abstraction in our settings for two reasons we delineate below. Figure~\ref{fig:x1-of-gu-gd} illustrates the allocation under a $(g_u,g_d)$ mechanism for the case of two buyers.

First, all mechanisms in the three different support information regimes (SPA with random reserves, POOL with random thresholds, and a randomization between SPA and POOL) can be represented in this form. Note that both $\spa(r)$ and $\pool(\tau)$ are $(g_u,g_d)$ with the following specification:
\begin{align*}
    &\spa(r): &&g_u(v) = \1(v \geq r) &&&&g_d(v) = 0 \\
    &\pool(\tau): &&g_u(v) = \frac{1}{n} + \frac{n-1}{n} \1(v \geq \tau) &&&&g_d(v) = \frac{1}{n} - \frac{1}{n} \1(v \geq \tau)
\end{align*}
Furthermore, this class of mechanisms is closed under randomization. A randomization over a family of $(g_u,g_d)$ mechanisms is still a $(g_u,g_d)$ mechanism with the resulting mechanism having $g_u$ and $g_d$ that are ``convex combinations'' over the base $g_u$ and $g_d$ functions. Therefore, all mechanisms having the form in Theorem~\ref{thm:char-all-mech-full-main} are $(g_u,g_d)$ mechanisms. The formal statement and the proof are deferred to Proposition~\ref{prop:gu-gd-rep} in the Appendix. 

Second, the $(g_u,g_d)$ mechanism, while general, still captures the sense in which our mechanisms are analytically tractable for saddle point calculations. More precisely, we have the following expression for the expected regret of a $(g_u,g_d)$ mechanism under an arbitrary distribution $F$.

\begin{proposition}[Expected Regret of a $(g_u,g_d)$ mechanism]\label{prop:reg-exp-g}
Let $R_{\lambda}(\mathbf{g},F) := R_{\lambda}((g_u,g_d), F)$ be the expected $\lambda$-regret of a $(g_u,g_d)$ mechanism under i.i.d. distribution $F$.

If we assume that $g_u$ and $g_d$ are continuous everywhere and differentiable everywhere except at a finite number of points, then
\begin{align*}
R_{\lambda}(\mathbf{g},F) &= a(\lambda-g_u(a)-(n-1)g_d(a)) + \int_{v=a}^{v=b} (\lambda-g_u(v)+ g_d(v) -vg_u'(v)+ (v-na) g_d'(v))  \\
&+  \int_{v=a}^{v=b} \left( -\lambda -(n-1) (g_u(v)-g_d(v)) + v(g_u'(v)+(n-1) g_d'(v)) \right) F(v)^n dv \\
&+ \int_{v=a}^{v=b} n(g_u(v)-g_d(v)-(v-a)g_d'(v))   F(v)^{n-1}  dv \, . \tag{Regret-$F$} \label{eqn:regret-F}
\end{align*}

If, instead, we let $g_u$ and $g_d$ be arbitrary but we assume that $F$ has a density $f = F'$ on $[a,b)$ (so it potentially has point masses only at $a$ and $b$ of size $F(a)$ and $F(\{b\}) := f_b$ respectively), then we have the (Regret-$\mathbf{g}$) expression
\begin{align*}
R_{\lambda}(\mathbf{g},F) &= \lambda b - a \left( g_u(a) + (n-1) g_d(a) \right) F(a)^n - \left( b g_u(b) + (n-1)a g_d(b) \right) \left( 1 -(1-f_b)^n \right) \\
&+ (b-a) g_d(b) \left( 1-(1-f_b)^{n-1} (1+(n-1)f_b) \right) \\
&+ \int_{v=a}^{v=b} -\lambda F(v)^n + g_u(v) n F(v)^{n-1} (1-F(v) - vF'(v) ) dv \\
&+ \int_{v=a}^{v=b} g_d(v) n(n-1) F(v)^{n-2} F'(v) \left\{ (v-a)(1-F(v)) - a F(v) \right\} dv \, . \tag{Regret-$\mathbf{g}$}  \label{eqn:regret-gbold}
\end{align*}

\noindent \eqref{eqn:regret-gbold} is valid for $n \geq 1$ if we take the expression $n(n-1)F(v)^{n-2}$ to be zero for $n = 1$.

\end{proposition}


\subsubsection{Verification of a Saddle Point}
Importantly, the \eqref{eqn:regret-F} shows that the expected regret can be written as an explicit polynomial function of the marginal CDF $F(v)$ and is ``separable'' as a function of $F(v)$. This allows us to maximize regret as a function of $F$ by maximizing each individual $F(v)$ expression ``pointwise'' independently for each $v$, subject only to the constraint that $F(v)$ is weakly increasing in $v$ (which is automatically satisfied for our specific $F$). By deriving first and second conditions using \eqref{eqn:regret-F}, we derive the following sufficient conditions for Nature's saddle.

\begin{proposition}\label{prop:nature-saddle-foc-soc}
Suppose that $\mathbf{g}^*$ is a $(g_u^*,g_d^*)$ mechanism and $F^*(v)$ is an increasing function that satisfy the following conditions:
\begin{align*}
 \left( -\lambda -(n-1) (g_u^*(v)-g_d^*(v)) + v(g_u'(v)+(n-1) g_d'^*(v)) \right)  F^*(v)  \\ + (n-1) (g_u^*(v)-g_d^*(v)-(v-a)g_d'^*(v))   = 0 \tag{FOC} \label{eqn:foc} \\
g_u^*(v)-g_d^*(v)-(v-a)g_d'^*(v)  > 0 \, . \tag{SOC} \label{eqn:soc}
\end{align*}
Then $R(\mathbf{g}^*,F^*) \leq R(\mathbf{g}^*,F)$ for any $F$.
\end{proposition}


Note that first-order and second-order conditions together do \textit{not} imply global optimality in general. It is only true in this case due to the special structure of the integrand, which has the form $\alpha F(v)^{n-1} - \beta F(v)^{n}$ for each $F(v)$, that we analyze directly.

To verify Nature's saddle it is sufficient to check \eqref{eqn:foc} and \eqref{eqn:soc} for the specific $g_u^*$ and $g_d^*$ and $F^*$ for each of the three regimes. We defer these calculations to Appendix~\ref{subsubsec:saddle-verify}. 

The seller's saddle is $R(m,F^*) \leq R(m^*, F^*)$. Optimizing over $m$ given i.i.d.~$F^*$  is a standard Bayesian mechanism design problem, and optimality of $m^*$ in each respective regime follows from applying the classical theory from \citet{Myerson81}.

\subsubsection{Derivation of the Saddle Point for the High Information Regime}

To give the reader a sense of how these calculations work, we will work out the guess-and-verify procedure in the high support information regime below. Throughout the rest of this subsection, we assume $a/b \geq k_h$. We first use the formula \eqref{eqn:regret-gbold} and the sellers' saddle to pin down the worst-case distribution $F^*$. In turn, plugging $F^*$ in the first-order condition of Nature's saddle in Proposition~\ref{prop:nature-saddle-foc-soc} gives us a differential equation involving $g_u$ and $g_d$ that can be used to  solve for the mechanism. We then conclude by formally verifying that the candidate saddle point is optimal using the approach delineated in the previous section.

In the high information regime, the pooling auction is conjectured to be optimal and, thus, we should always allocate the item. So, we must have $g_u(v) + (n-1)g_d(v) = 1$ for every $v$. Using \eqref{eqn:regret-gbold}, we can write the expected regret purely in terms of $g_u$ as
\begin{align*}
&\lambda b - a  F(a)^n - b + b(1-f_b)^n \\
&+ \int_{v=a}^{v=b} \Big[ - \lambda F(v)^n + n F(v)^{n-2} F'(v) \left\{ (v-a)(1-F(v)) - a F(v) \right\}  \\ &\quad\quad\quad\quad\quad + n F(v)^{n-2} \left\{ F(v) - F(v)^2 - (v-a) F'(v) \right\} g_u(v) \Big] dv \, .
\tag{Regret-$g$}  \label{eqn:regret-g}
\end{align*}

Note that \eqref{eqn:regret-g} depends on $g_u$ only through $g_u(v)$ and is linear in $g_u$. This is useful for the seller's saddle $\inf_{m} R_{\lambda}(m,F^*)$. If the seller maximizes over the $\pool$ mechanism parameterized by $g_u$, then by the first-order conditions, under the worst-case distribution $F^*$, the coefficient of each $g_u(v)$ should be zero. Otherwise, the seller could decrease her regret by changing the distribution of reserves. Therefore,
\begin{align*}
    F^*(v) - F^*(v)^2 - (v-a) (F^*)'(v) = 0 \, \Rightarrow \, \frac{d}{dv} \left( v - \frac{v-a}{F^*(v)} \right) = 0 \, \Rightarrow \, v - \frac{v-a}{F^*(v)} = \phi_0\,.
\end{align*}
This pins down Nature's candidate distribution as $F^*(v) = (v-a)/(v-\phi_0)$, a distribution with constant virtual value $\phi_0$. We note this part of the argument (guessing $F^*$) is technically not needed in the formal proof, because it is sufficient to simply verify the saddle point to confirm that our mechanism is robustly optimal. Therefore, we can make certain smoothness assumptions to derive this $F^*$ which are \textit{not needed} for the formal saddle verification proof. Nevertheless, we think it is useful for the readers to understand where our candidate saddle point comes from. 

We now derive the mechanism $g_u^*$. Proposition~\ref{prop:nature-saddle-foc-soc} gives sufficient conditions to imply Nature's saddle. The \eqref{eqn:foc} also gives an Ordinary Differential Equation (ODE) on $g_u^*$ that uniquely determines it as follows. Substituting $g_d^*(v) = (1-g_u^*(v))/(n-1)$ and $F^*(v) = (v-a)/(v-\phi_0)$ in \eqref{eqn:foc}, we get
\begin{align*}
    \frac{d}{dv} \left[ \frac{(v-a)^{n}}{(v-\phi_0)^{n}} g_u^*(v) \right] = \frac{(v-a)^{n-1}}{(v-\phi_0)^{n}} - (1-\lambda) \frac{(v-a)^{n}}{(v-\phi_0)^{n+1}} .
\end{align*}
We therefore get
\begin{align*}
    g_u^*(v) = \frac{(v-\phi_0)^{n}}{(v-a)^{n}} \int_{t=a}^{t=v} \left[ \frac{(t-a)^{n-1}}{(t-\phi_0)^{n}} - (1-\lambda) \frac{(t-a)^{n}}{(t-\phi_0)^{n+1}} \right] dt = \frac{1}{n} + \lambda \sum_{k=n+1}^{\infty} \frac{(v-a)^{k-n}}{k(v-\phi_0)^{k-n}} .
\end{align*}

The expression makes it clear that $g_u^*(a) = 1/n$ and $g_u^*(v)$ is increasing in $v$, the latter of which is necessary for the function to correspond to a true feasible mechanism. We also impose the condition that $g_u^*(b) = 1$, which gives an equation that $\phi_0$ must satisfy. By definition of $k_h$, we see by inspection that the resulting equation has an explicit solution
\begin{align*}
    \phi_0 = \frac{a-k_h b}{1-k_h} .
\end{align*}

We need $\phi_0 \geq 0$ for Nature's saddle to hold: this is why this mechanism and the corresponding saddle is valid only in the $a/b \geq k_h$ regime. 

The above ODE manipulation not only determines $g_u^*$, but also makes sure that the resulting $g_u^*$ satisfies \eqref{eqn:foc}. The only thing that remains for Nature's saddle is to verify \eqref{eqn:soc}. By substituting $g_d^*$ with $g_u^*$ and write $(g_u^*)'(v)$ in terms of $g_u^*(v)$ using the ODE from \eqref{eqn:foc}, \eqref{eqn:soc} reduces to $ng_u^*(v) - 1 + \lambda > 0$ which is true because $g_u^*(v) \geq g_u^*(a) = 1/n$. Lastly, verifying Seller's saddle is a standard Bayesian mechanism design problem. Fixing $F^*$, the optimality of a $\pool$ mechanism follows because under a constant positive virtual value $F^*$, every mechanism that always allocates is optimal.

}





\jaedit{

\subsection{Structure of Optimal DSIC Mechanisms}\label{subsec:structure-mech-all}

We next discuss the structure of optimal mechanisms under the minimax regret and maximin ratio objectives.

\subsubsection{Minimax Regret Objective}\label{subsubsec:structure-mech-regret}

The case of minimax regret is obtained by setting $\lambda = 1$ in  Theorem~\ref{thm:char-all-mech-full-main}. In Figure~\ref{fig:alloc-heatmaps-all}, we fix $b=1$, and depict optimal mechanisms for $a = 0$ (low information), $a = 0.25$ (moderate information), and $a = 0.5$ (high information). We show the allocation rule $x_1(v_1,v_2)$ to buyer 1, and the total allocation $x_1(v_1,v_2)+x_2(v_1,v_2)$ to both buyers. Note that the mechanism is symmetric, so $x_2(v_1,v_2) = x_1(v_2,v_1)$ and is therefore not explicitly shown. The corresponding $g_u$ and $g_d$ functions are shown in the last row of Figure~\ref{fig:alloc-heatmaps-all}. 

\newcommand{\figwidth}{2in}
 \begin{figure}[tbh!]
    \centering
    \begin{tabular}{ccc}        
    \includegraphics[width=\figwidth]{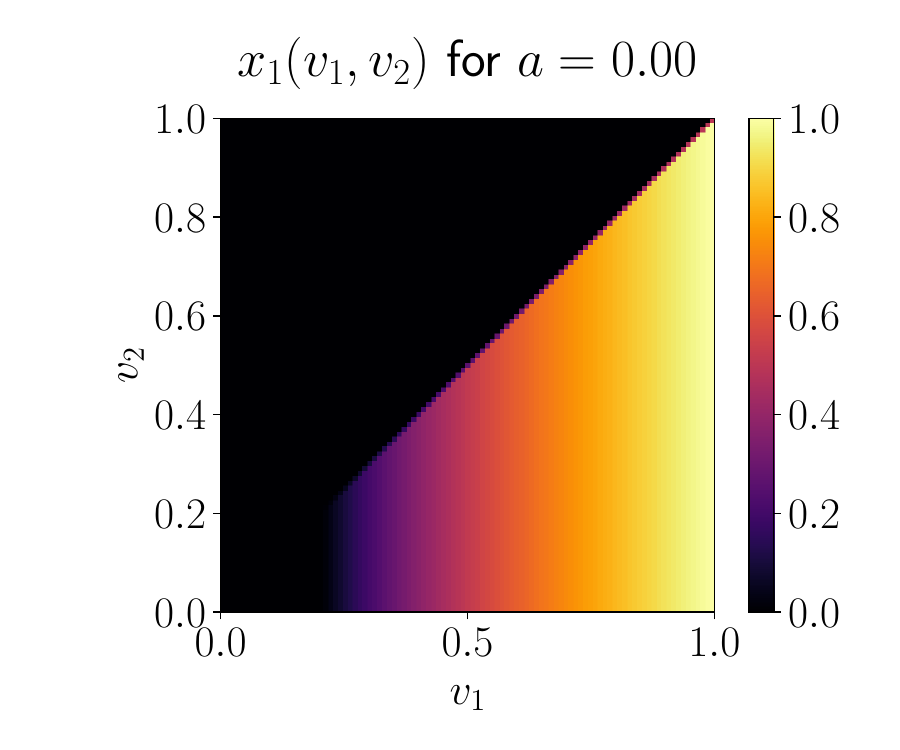} & 
    \includegraphics[width=\figwidth]{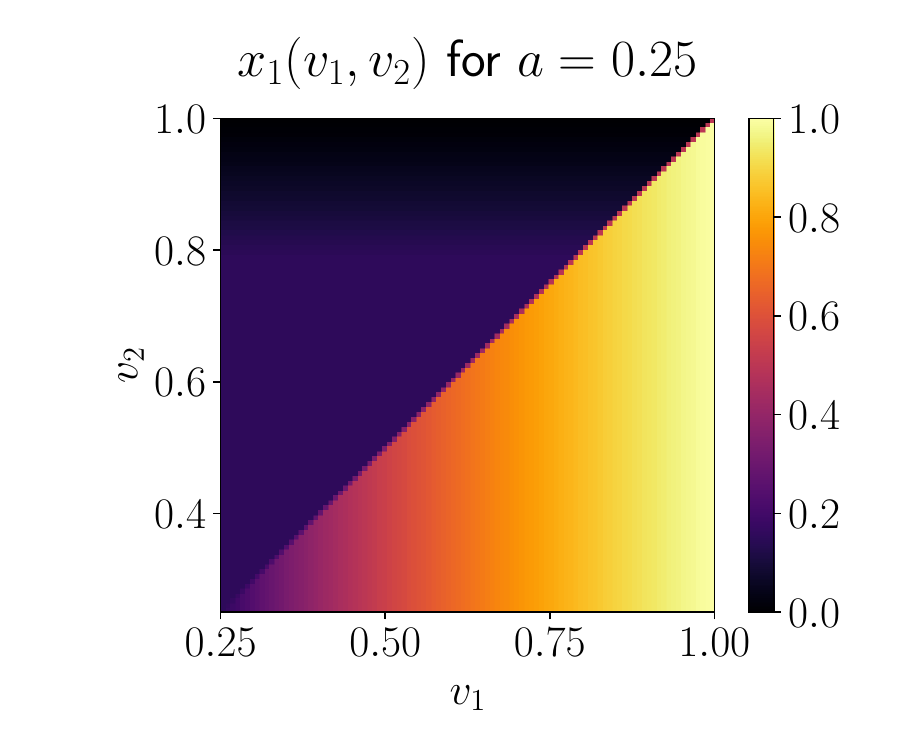} &
    \includegraphics[width=\figwidth]{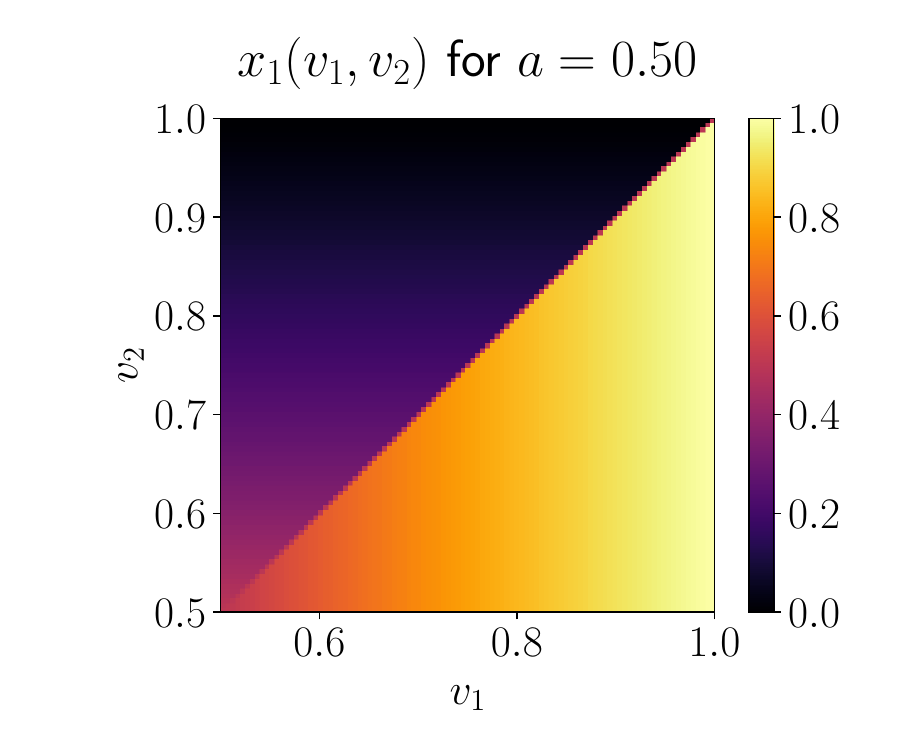} \\
    \includegraphics[width=\figwidth]{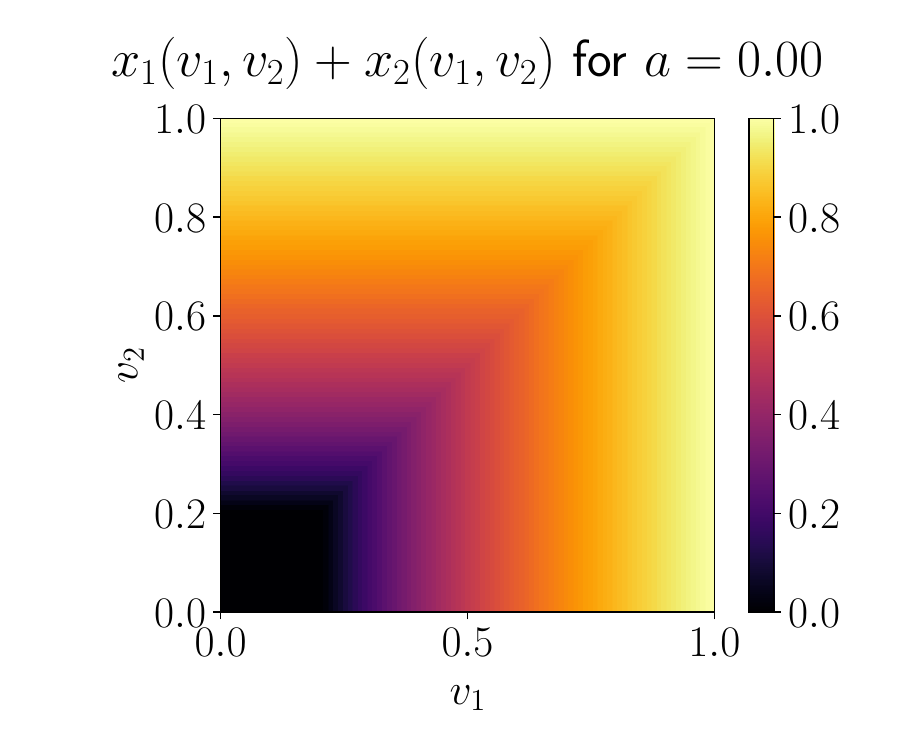} & 
    \includegraphics[width=\figwidth]{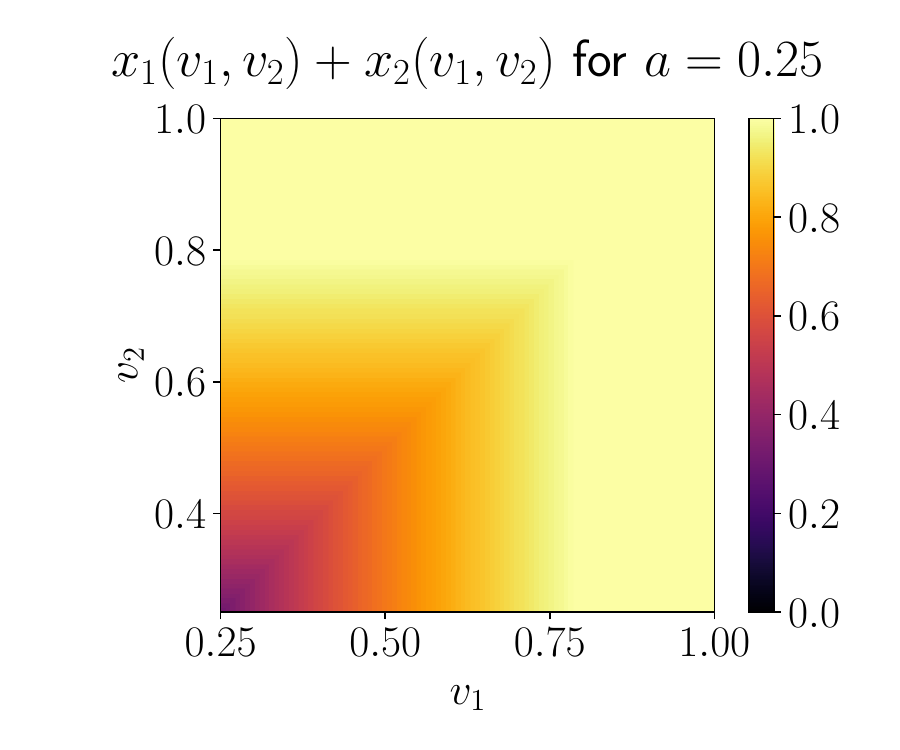} &
    \includegraphics[width=\figwidth]{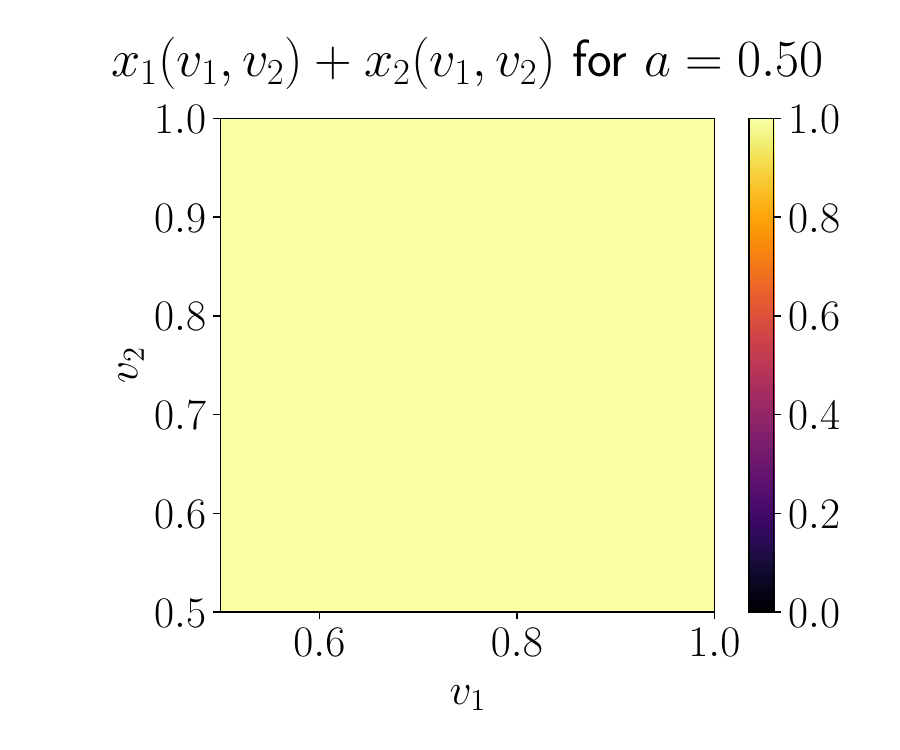} \\   \includegraphics[width=\figwidth]{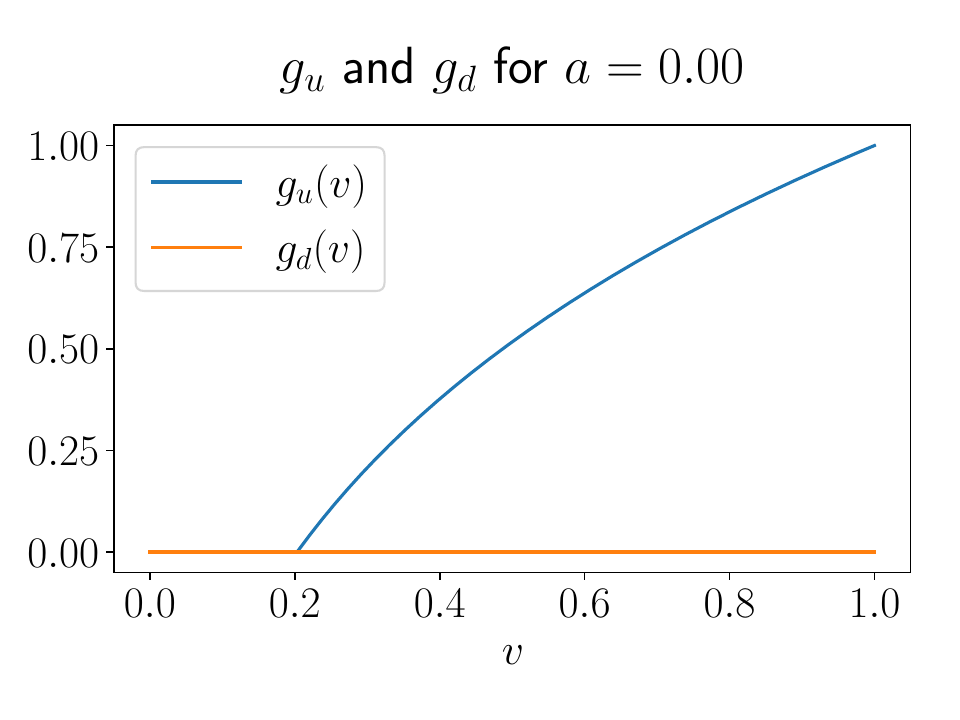} &
    \includegraphics[width=\figwidth]{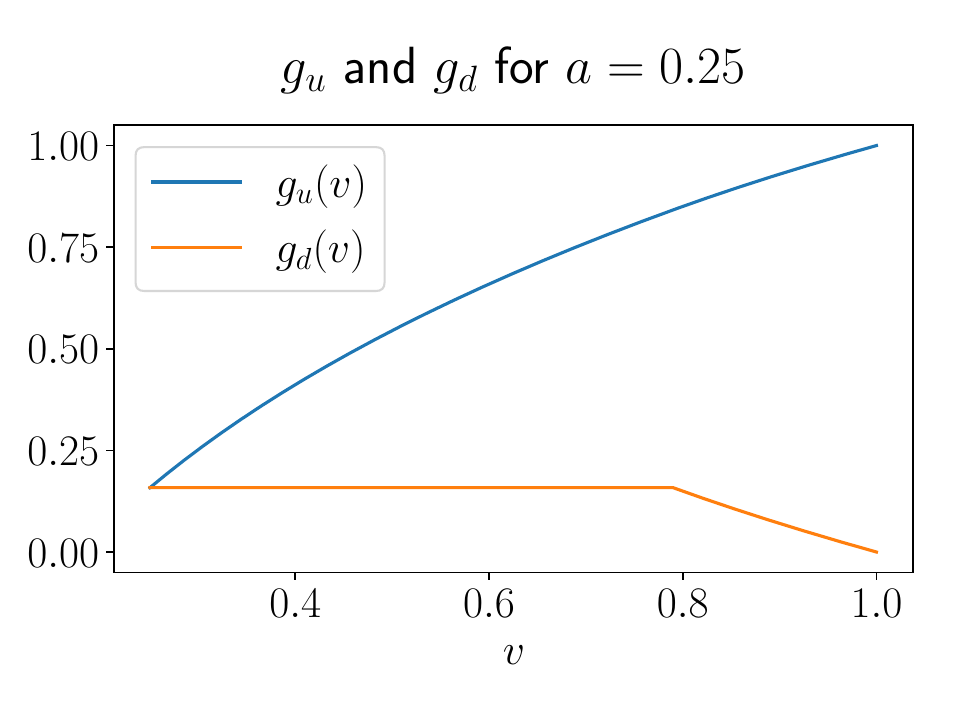} &
    \includegraphics[width=\figwidth]{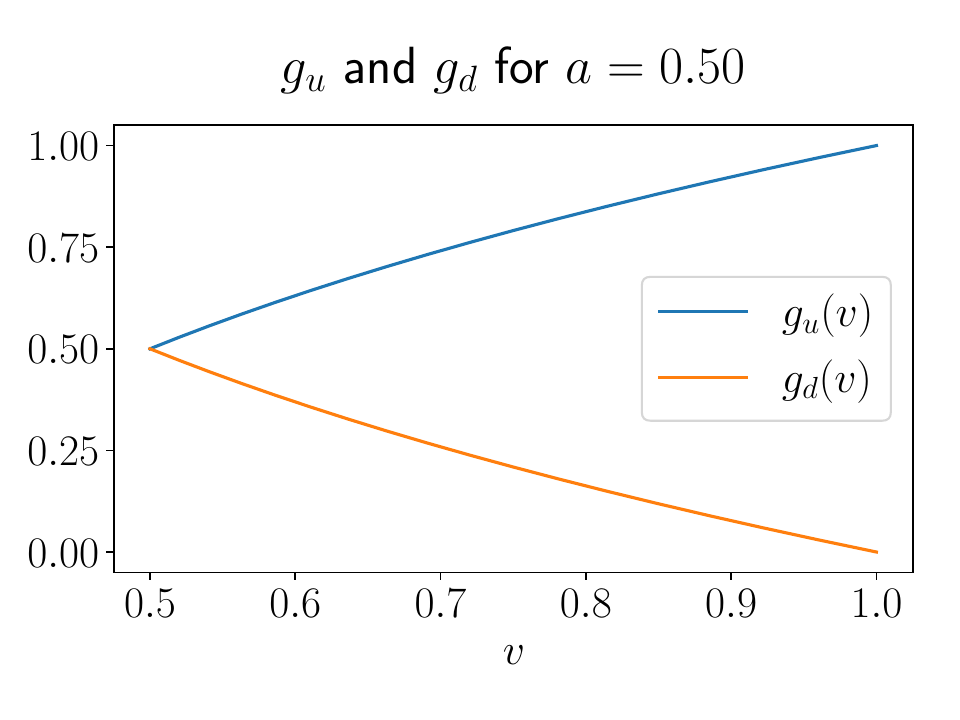} 
    \end{tabular}
	\caption{Optimal mechanisms for minimax regret. The top row shows the allocation rule $x_1(v_1,v_2)$, the second row shows the total allocation $x_1(v_1,v_2)+x_2(v_1,v_2)$, and the last row shows the $g_u$ and $g_d$ functions. The upper end of the support is $b=1$ and each columns shows a different value of the lower end $a \in \{0, 0.25, 0.50\}$ (low, moderate, and high support information regimes).}
\label{fig:alloc-heatmaps-all}
\end{figure}

The top heatmaps of Figure~\ref{fig:alloc-heatmaps-all} for $x_1(v_1,v_2)$ show that in the low information regime ($a=0$), the upper triangle is all zero because in that regime the optimal mechanism randomizes among SPAs, which never allocate to the non-highest buyer. The upper triangle is not zero for moderate (mixture of $\spa$ and $\pool$) and high information regimes ($\pool$). We can also see this from the bottom row of Figure~\ref{fig:alloc-heatmaps-all} because $g_d$ is zero for $a = 0$ but strictly positive for $a \in \{0.25,0.50\}$. 

The middle heatmaps of Figure~\ref{fig:alloc-heatmaps-all} show the total allocation $x_1(\mathbf{v}) + x_2(\mathbf{v}) = g_u(v_{\max}) + g_d(v_{\max})$. In the low information regime, it is always less than 1 because $\spa$ discards the item below the reserve. In the high information regime, it is always 1 because $\pool$ always allocates. In the moderate information regime, it is 1 in the $\pool$ region when $v_{\max} \geq v^* \approx 0.8$ and is less than 1 in the $\spa$ region when $v_{\max} < v^*$ in the lower left corner. We can also see this from the value of $g_u(v)+g_d(v)$ in the last row of Figure~\ref{fig:alloc-heatmaps-all}. 

In all heatmaps, we see that in the lower triangle ($v_1 \geq v_2$), the values are constant on each horizontal line because the allocations are $g_u(v_1)$ for bidder 1 and $g_d(v_1)$ for bidder 2 which depend on $v_1$, while in the upper triangle ($v_1 \leq v_2$), the values are constant on each vertical line because the allocations are $g_d(v_2)$ for bidder 1 and $g_u(v_2)$ for bidder 2 which depend only $v_2$, just as Figure~\ref{fig:x1-of-gu-gd} suggests. \obcomment{Jerry, you need to explain this more...}

}

\jadelete{
In the first row (low support information), $\spa(\Phi)$ is optimal. We can see this in $x_1(v_1,v_2)$. The upper triangle is all 0 because SPA is ``standard'' as in it never allocates to the non-highest buyer. The lower triangle has constant value vertically because the allocation probability is $g_u(v_1) = \Phi(v_1)$ depends only on the highest $v_1$ and not on $v_2$. It is zero from 0 up to $k_l \approx 0.2$ before increasing until reaching 1 at $v_1 = 1$, as expected. The total allocation $x_1+x_2$ is always below 1 except at $v_1 = 1$ or $v_2 = 1$ because SPAs do not allocate the item if the highest value is below the reserve.

In Figure~\ref{fig:phi-tilde-psi-tilde}a, we depict the CDF of the optimal $\textnormal{SPA}$-reserve distribution $\Phi$ as a function of $n$; in  Figure~\ref{fig:phi-tilde-psi-tilde}b, we plot the optimal $\textnormal{POOL}$-threshold distribution $\Psi$ as a function of $n$. \obcomment{Eventhough I understand why, it is a bit strange visually to have $n=1$ in one figure and not the other... also, it is unclear how the picture illustrate the structure of the mechanisms. Also, we already have Figure 2a) in our former paper, so we may want have new stuff only unless we feel it adds a lot. Maybe Figure 3 is sufficient. In addition, maybe we want something that explore a good visualisation of the characteristics of the mechanisms (e.g., a heatmap of allocation in the format of Figure 1c)? }

\begin{figure}[h!]
\captionsetup{justification=centering}
	\begin{subfigure}{0.48\linewidth}
    	\centering
    	\includegraphics[height=0.75\linewidth]{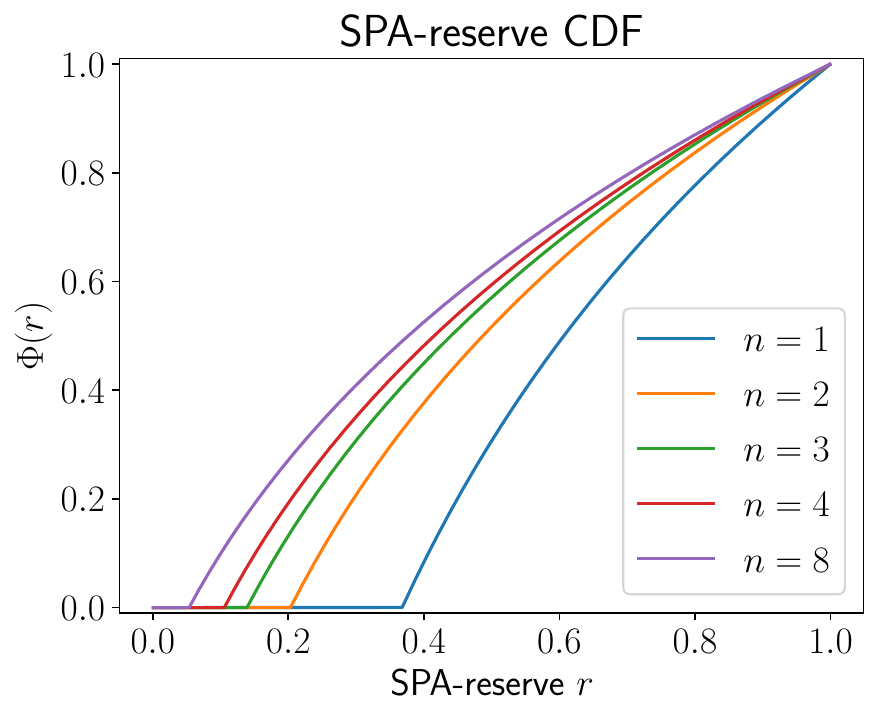}
    	\caption{Low support information \\ $a=0,b=1$}
    \end{subfigure}%
	\begin{subfigure}{0.48\linewidth}
		\centering
		\includegraphics[height=0.75\linewidth]{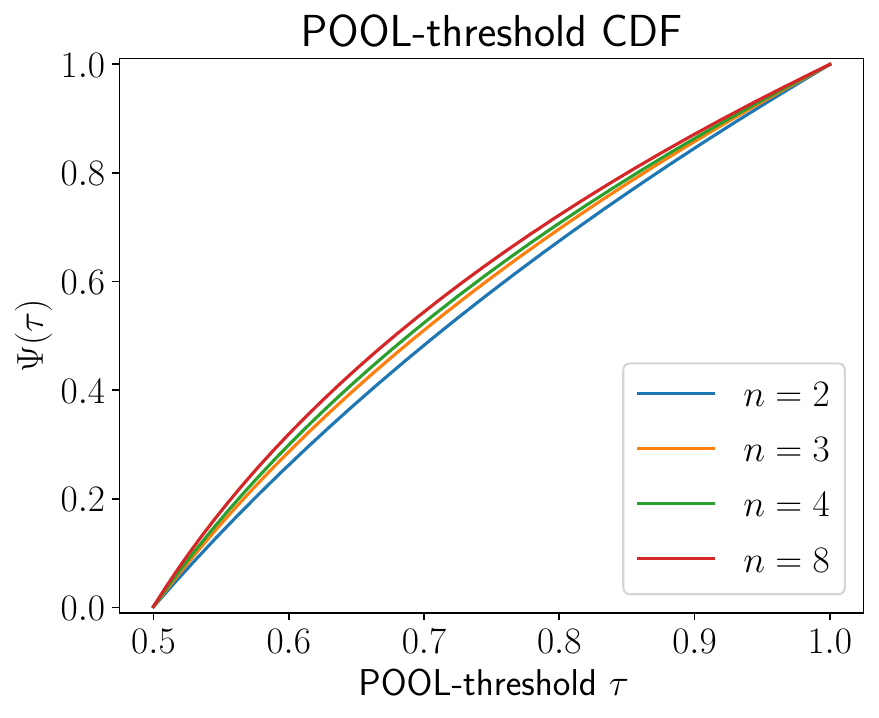}
		\caption{High support information \\ $a=0.5,b=1$}
	\end{subfigure}
	\caption{Structure of the optimal mechanism in the low and high support information regimes. (The intermediate support information regime is shown in Figure~\ref{fig:intermediate-regime-interpolation}.)}
\label{fig:phi-tilde-psi-tilde}
\end{figure}
\obcomment{In Figure 3b) above, the legend should not interfere with the curves.}

Theorem~\ref{thm:char-all-mech-full-main} states that in the moderate support information regime ($k_l \leq a/b \leq k_h$), the optimal mechanism is a randomization over $\textnormal{SPA}$s with random reserves (thresholds) in $[a,v^*]$ and $\textnormal{POOL}$s with random thresholds in $[v^*,b]$. We collectively call both types of thresholds ``unified thresholds'' such that the mechanism is characterized by the CDF of the random unified threshold $\mathcal{D}$. Figure~\ref{fig:spa-iron-a-025} shows the example of such a mechanism (the CDF $\mathcal{D}$) for $n = 2, \lambda = 1, a=0.25$, which is between $k_l \approx 0.2032$ and $k_h \approx 0.3162$. The unified threshold distribution $\mathcal{D}$ can be decomposed into a measure $\Phi$ over $\textnormal{SPA}$-reserves on $[a,v^*]$ with weight $1-n\alpha$ and a measure $\Psi$ over $\textnormal{POOL}$-thresholds on $[v^*,b]$ with weight $n\alpha$ for some $\alpha \in [0,1/n]$ as in Theorem~\ref{thm:char-all-mech-full-main}. The $\textnormal{SPA}$ and $\textnormal{POOL}$ parts are shown with solid and dashed lines, respectively. The two parts meet at boundary $v^*$ with CDF value $1-n \alpha$, the total measure of $\Phi$.

\begin{figure}[h!]
	\begin{subfigure}{0.48\linewidth}
    	\centering
    	\includegraphics[height=0.75\linewidth]{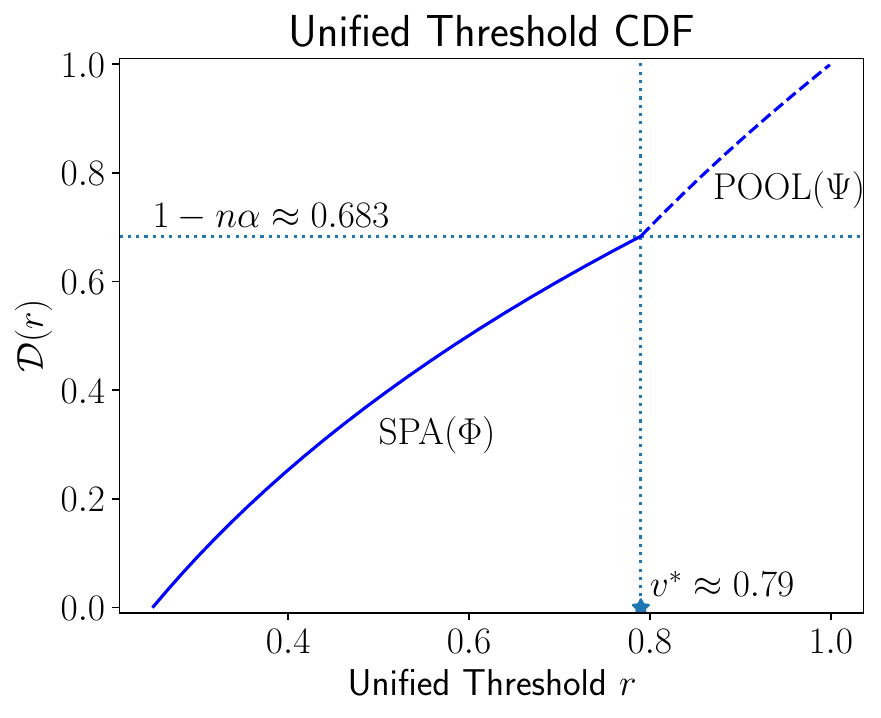}
    	\caption{Optimal mechanism in the intermediate regime}
    	\label{fig:spa-iron-a-025}
    \end{subfigure}%
	\begin{subfigure}{0.48\linewidth}
		\centering
		\includegraphics[height=0.75\linewidth]{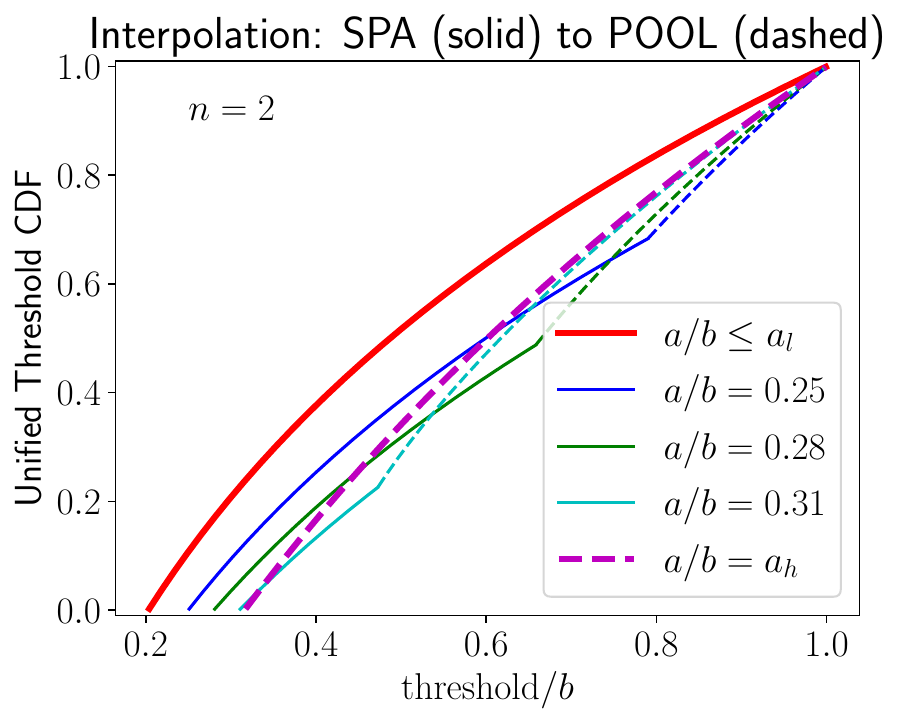}
		\caption{Continuous interpolation from $\textnormal{SPA}$ to $\textnormal{POOL}$}
		\label{fig:spa-iron-compare-a}
	\end{subfigure}
	\caption{Unified threshold CDFs with $\textnormal{SPA}$-reserve (solid lines) and $\textnormal{POOL}$-thresholds  (dashed lines) with $n = 2, \lambda = 1$. (a) highlights the important features of the mechanism with $a/b = 0.25$. (b) shows the continuous interpolation of the mechanisms for $a/b \in \{k_l,0.25,0.28,0.31,k_h\}$. }
\label{fig:intermediate-regime-interpolation}
\end{figure}

Lastly, we observe that as $a/b$ varies from $k_l$ to $k_h$, the optimal mechanism transitions \textit{smoothly} from low ($\textnormal{SPA}$) to high ($\textnormal{POOL}$) support information regimes. Figure~\ref{fig:intermediate-regime-interpolation}  shows the continuous interpolation as $a/b \in \{k_l,0.25,0.28,0.31,k_h\}$. As $a/b$ gets closer to $k_h$, the $\textnormal{SPA}$  part shrinks and the $\textnormal{POOL}$ part expands, and vice versa.

}

\subsubsection{Maximin Ratio Objective}\label{subsubsec:structure-mech-ratio}

\jaedit{

Figure~\ref{fig:maximin-ratio-intro} depicts the maximin ratio value as a function of relative support information $a/b$ for different numbers of buyers $n$. The ratio quantity gives us the fraction of ideal revenue that the mechanism can achieve. This figure provides quantitative evidence that even a small amount of knowledge can lead to nontrivial guarantees on revenue.

%

 \begin{figure}[h!]
		\centering
		\includegraphics[height=0.4\linewidth]{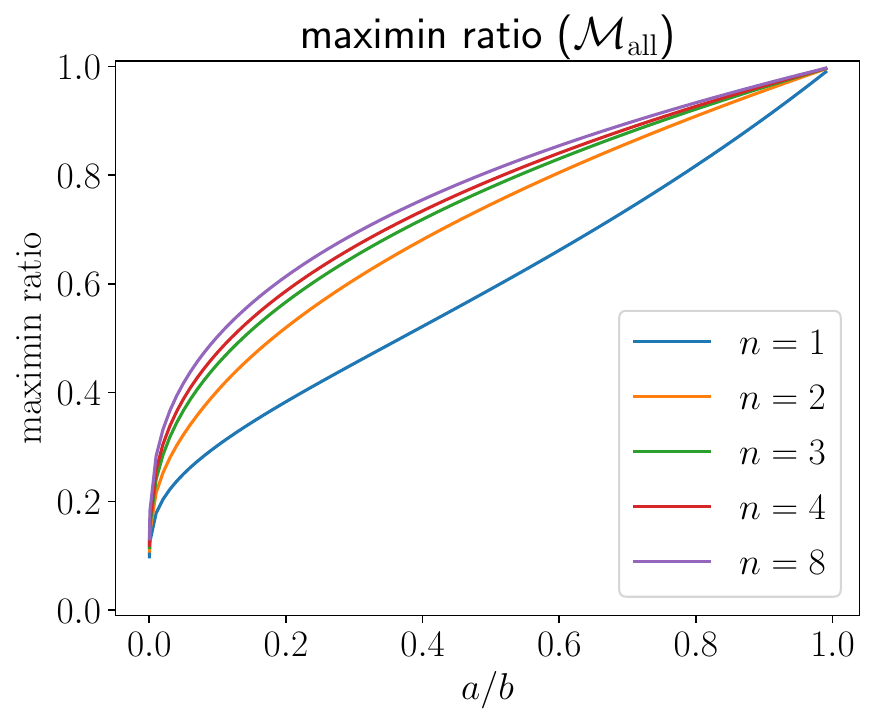}
	\caption{Maximin ratio as a function of relative support information $a/b$ and the number of buyers $n$.}
	\label{fig:maximin-ratio-intro}
\end{figure}

}

Unlike the minimax regret case where $\lambda = 1$ is set exogenously, here $\lambda$ is obtained from bisection search to find the value of $\lambda$ such that the minimax $\lambda$-regret is zero (cf. Proposition \ref{prop:minimax-regret-exp}) and $\lambda = \lambda^*(k,n)$ is a function of $k \equiv a/b$, i.e., the maximin ratio given $k$ that we computed earlier. As a result, the regime is determined \jaedit{endogenously}.  For each $n$, we then compare $k = a/b$ with $k_l = k_l(\lambda^*(k,n),n)$ and $k_h = k_h(\lambda^*(k,n),n)$ to determine the regime. \jaedit{The next result shows that the maximin ratio mechanism is either a pure $\pool$, or a mixture of $\spa$ and $\pool$.

\begin{proposition}\label{prop:maximin-ratio-regime}
The optimal mechanism identified in Theorem \ref{thm:char-all-mech-full-main}, when specialized to the maximin ratio objective, is never in the pure $\spa$ regime. 
\end{proposition}

The proof is given in Appendix~\ref{app:subsec:structure-mech-all}.  Numerically, we find that the mechanism is in a pure pooling auction regime for ``reasonable'' values of $a/b$, namely, $a/b \geq 0.0978$ for $n = 2$, $a/b \geq 0.0155$ for $n = 3$, and $a/b \geq 0.0035$ for $n = 4$.
We can visualize the pool threshold distribution $\Psi$ of $\pool(\Psi)$ as follows. Because $\Psi$ is supported on $[a,b]$, with varying parameters $a$ and $b$ we normalize the threshold $\tau$ by $\tilde{\tau} = (\tau-a)/(b-a)$ so the normalized thresholds are on the same scale $[0,1].$
}
For $n \in \{2,4\}$ and $a/b \in \{0.10,0.25,0.50,0.75,0.99\}$, we plot the normalized $\textnormal{POOL}$-threshold CDFs in Figure~\ref{fig:ratio-structure}.  \jaedit{We see that for low $a/b$, the distribution puts more weight on lower thresholds, and vice versa. However, the distributions are quite close for a wide range of $a/b$ and for reasonably high values of $a/b$, the normalized distribution is close to uniform.}

\jadelete{
We can prove that the pure $\spa$ regime is \textit{never} possible for the maximin ratio objective; we formalize and prove this in Proposition~\ref{prop:maximin-ratio-regime}.\obcomment{I think this is important and could be in the main text} For reasonable values of $a/b$,\footnote{For $a/b \geq 0.0978$ for $n = 2$, $a/b \geq 0.0155$ for $n = 3$, and $a/b \geq 0.0035$ for $n = 4$ we are in the pure $\pool$ regime. } the optimal mechanisms are in the regime of pure pooling auctions. For a $\pool$-threshold $\tau$, we define the $\textit{normalized}$ threshold as $\tilde{\tau} = (\tau-a)/(b-a)$. As $\tau \in [a,b]$, we have $\tilde{\tau} \in [0,1]$, so this normalization allows us to compare the shape of threshold distributions on the same scale. 
}

 \begin{figure}[h!]
	\begin{subfigure}{0.48\linewidth}
		\centering
	\vspace*{1em}	\includegraphics[height=0.75\linewidth]{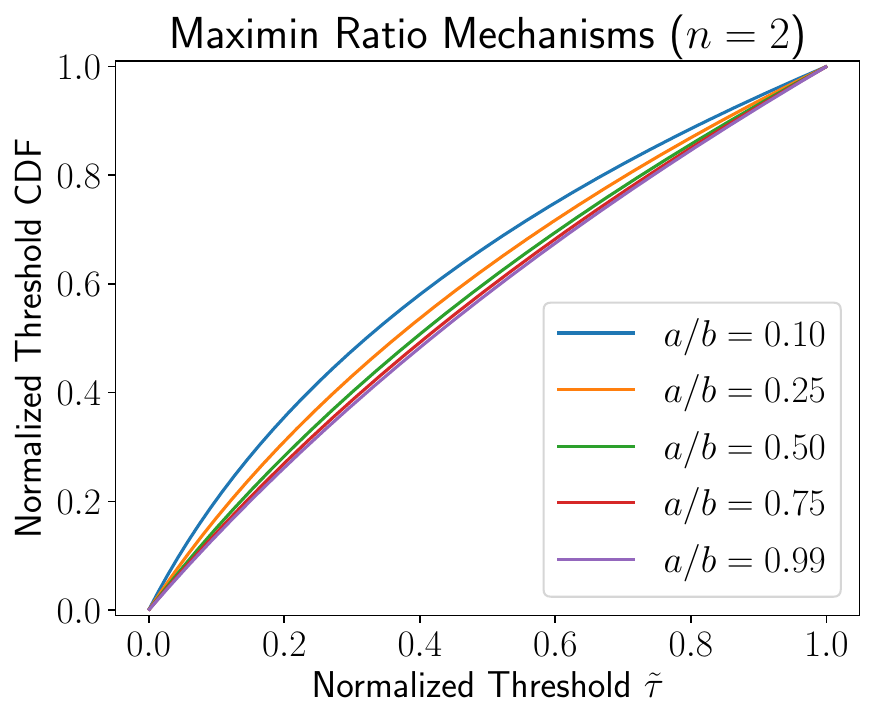}
	\end{subfigure}
	\begin{subfigure}{0.48\linewidth}
		\centering 
	\vspace*{1em}	\includegraphics[height=0.75\linewidth]{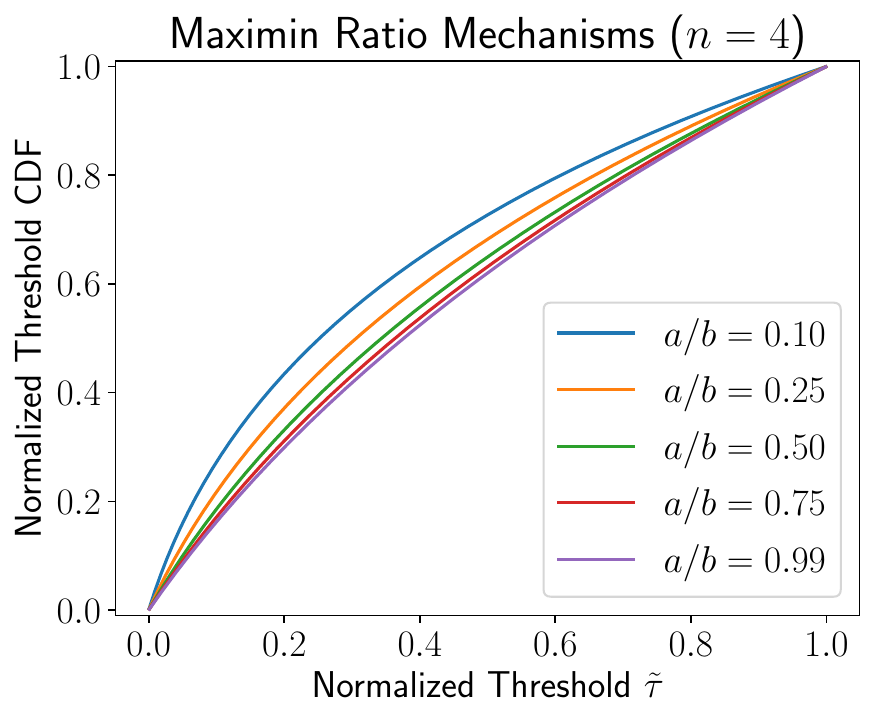}
	\end{subfigure}
	\caption{Normalized threshold distributions for $n \in \{2,4\}$ and $a/b \in \{0.10,0.25,0.50,0.75,0.99\}$}
\label{fig:ratio-structure}
\end{figure}

\subsection{Remark on the case $n=1$}\label{subsec:n-1-remark} 

An important corollary of Theorem~\ref{thm:char-all-mech-full-main} is the pricing case (one-bidder / no-competition). Applying the result with $n=1$ directly recovers the minimax regret result of \cite{BergemannSchlag08} and the maximin ratio result of \cite{ErenMaglaras10} as special cases.\footnote{More precisely, \cite{ErenMaglaras10} derives the maximin ratio against the second-best benchmark in the discrete price setting, whereas our result is against the first-best benchmark in the continuous setting. However, their numerical value for the ratio approaches ours as the grid resolution becomes finer. We can also show that for the $n = 1$ case, the maximin ratio for two benchmarks are the same. \cite{ErenMaglaras10} also does not explicitly derive the optimal mechanism, whereas we do.} We give the proof of this corollary in Appendix~\ref{app:subsec:n-1-remark}.

\begin{corollary}[Pricing]\label{cor:lambda-n-1}
Suppose $n=1$ \jaedit{and fix $\lambda$ in $[0,1]$}. For $a/b \leq e^{-1/\lambda}$, the minimax $\lambda$-regret is $\lambda e^{-1/\lambda} b$, achieved with the price distribution CDF $\Phi(v) = 1 + \lambda \log(v/b)$ for $v \geq e^{-1/\lambda} b$ and 0 otherwise. For $a/b \geq e^{-1/\lambda}$, the minimax $\lambda$-regret is $-a+\lambda a + \lambda \log(b/a)$ achieved with the price distribution CDF $\Phi(v) = 1 + \lambda \log(v/b)$ for $v \in [a,b]$.

In particular, the minimax regret is $b/e$ if $a/b \leq 1/e$ and $a \log(b/a)$ if $a/b \geq 1/e$.  For $a>0$, the maximin ratio is $1/(1+\log(b/a))$, achieved by the price distribution $\Phi(v) = 1 + \log(v/b)/(1+\log(b/a))$ for $v \in [a,b]$.
\end{corollary}

We remark that in the one-bidder case, $k_l = e^{-1/\lambda}$ and $k_h = 1$, so there are only two regimes (low and moderate support information), and this is reflected in the corollary statement. Moreover, with only one bidder, POOL becomes a degenerate mechanism that always allocates. This is why in the $a/b > e^{-1/\lambda}$ regime (moderate information), the optimal mechanism, which is a mixture of SPA and POOL, always allocates with positive probability. This can be seen in the pricing CDF $\Phi(v) = 1 + \lambda \log(v/b)$, which has a point mass of positive size $1 + \lambda \log(a/b) > 0$ at $v = a$.

\section{Minimax $\lambda$-Regret across Mechanism Classes}\label{sec:other-mech-classes}

Our main theorem (Theorem~\ref{thm:char-all-mech-full-main}) gives a complete characterization of the optimal robust performance when Nature's distribution is i.i.d. ($\mathbf{F} \in \mathcal{F}_{\textnormal{iid}}$) and the seller can choose any DSIC mechanism ($m \in \mathcal{M}_{\textnormal{all}}$). It turns out that the optimal mechanism is generally a randomization over $\textnormal{SPA}$ and $\textnormal{POOL}$ mechanisms. This optimal mechanism has interesting features, and we would like to quantify how much each feature contributes to the performance.  That is, without that feature, how much (robust) performance, if any, we will lose.  Equivalently, our results quantify the ``cost of simplicity'' or the performance loss if the seller is restricted to simpler classes of mechanisms. We formalize this problem by solving minimax $\lambda$-regret problems, $\lambda \in (0,1]$, when the mechanism classes $\mathcal{M}$ are successively smaller, omitting one feature at a time.  The subclasses under consideration are shown in Figure~\ref{fig:nested-mech-subclasses}.
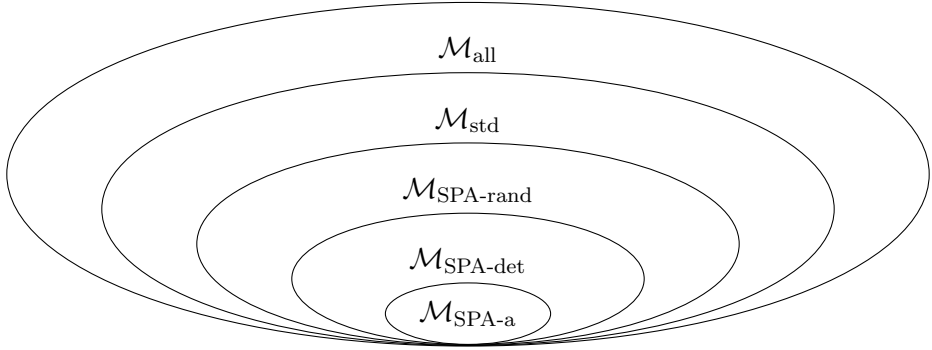
\begin{figure}[h!]
    	\centering
    	    \begin{tikzpicture}[font=\sffamily,breathe dist/.initial=2ex]
    \foreach \X [count=\Y,remember=\Y as \LastY] in 
    {$\mathcal{M}_{\textnormal{SPA-a}}$,$\mathcal{M}_{\textnormal{SPA-det}}$,$\mathcal{M}_{\textnormal{SPA-rand}}$,$\mathcal{M}_{\textnormal{std}}$,$\mathcal{M}_{\textnormal{all}}$}
     {\ifnum\Y=1
      \node[ellipse,draw,outer sep=0pt] (F-\Y) {\X};
     \else
      \node[anchor=south] (T-\Y) at (F-\LastY.north) {\X};
      \path let \p1=($([yshift=\pgfkeysvalueof{/tikz/breathe dist}]T-\Y.north)-(F-\LastY.south)$),
      \p2=($(F-1.east)-(F-1.west)$),\p3=($(F-1.north)-(F-1.south)$)
      in ($([yshift=\pgfkeysvalueof{/tikz/breathe dist}]T-\Y.north)!0.5!(F-\LastY.south)$) 
      node[minimum height=\y1,minimum width={\y1*\x2/\y3},
      draw,ellipse,inner sep=0pt] (F-\Y){};
     \fi}
    \end{tikzpicture}
\caption{Nested mechanism subclasses we consider, from biggest to smallest: DSIC mechanisms ($\mathcal{M}_{\textnormal{all}}$), standard mechanisms ($\mathcal{M}_{\textnormal{std}}$), SPA with random reserve ($\mathcal{M}_{\textnormal{SPA-rand}}$), SPA with deterministic reserve ($\mathcal{M}_{\textnormal{SPA-det}}$), SPA with no reserve ($\mathcal{M}_{\textnormal{SPA-a}}$)}
\label{fig:nested-mech-subclasses}
\end{figure}

First, our optimal mechanism is not standard because $\textnormal{POOL}$ might allocate to a bidder who is not the highest. To isolate the role of the pooling feature, we study the class of standard mechanisms that only allocate to the maximum bidder. Second, we study the need to deviate from SPAs in standard mechanisms, and hence study SPAs with randomized reserves. Lastly, we quantify the power of randomness and the power of using a reserve by computing minimax regret under the class $\mathcal{M}_{\textnormal{SPA-det}}$ of SPA with a deterministic reserve and the class $\mathcal{M}_{\textnormal{SPA-a}}$ of SPA with no reserve. Interestingly, we show that there are strict separations in terms of maximin ratio between $\mathcal{M}_{\textnormal{all}}$, $\mathcal{M}_{\textnormal{std}}$, $\mathcal{M}_{\textnormal{SPA-rand}}$, and $\mathcal{M}_{\textnormal{SPA-a}}$ (but not between $\mathcal{M}_{\textnormal{SPA-det}}$ and $\mathcal{M}_{\textnormal{SPA-a}}$). In other words, pooling and deviations from SPAs are critical for robust performance, and so is the randomization of reserve prices.

\subsection{Minimax $\lambda$-Regret Over Standard Mechanisms}\label{subsec:main-std-mech}

A mechanism is said to be \textit{standard} if it never allocates to an agent that does not have the highest value. Formally, it satisfies the following constraint:
\begin{align*}
x_i(v_i,\mathbf{v}_{-i}) &= 0 \quad \forall i, v_i, \mathbf{v}_{-i} \text{  such that  } v_i < \max(\mathbf{v})\,. \tag{STD}
\end{align*}
We can now define the class of all standard mechanisms.
\begin{definition}\label{def:std-mech}
The class of all standard mechanisms is given by 
\begin{align}\label{mech-ir-ic-ac-std}
    \mathcal{M}_{\textnormal{std}} = \left\{ (\mathbf{x},\mathbf{p}): \textnormal{(IR), (IC), (AC), (STD)} \right\}.
\end{align}
\end{definition}

It is clear that any second-price auction (SPA) with random reserve is standard, and intuitively, SPAs seem like ``natural'' and ``typical'' elements of this class, but as it turns out, other standard mechanisms lead to higher performance than SPAs when relative support information is high. \jaedit{Let $v^{(1)}$ and $v^{(2)}$ be the highest and second highest values in the vector $\mathbf{v}$.} We now introduce the following mechanism class.
\begin{definition}[Generous SPA]\label{def:genspa}
A \emph{generous SPA with reserve distribution $\Phi$},  denoted $\textnormal{GenSPA}(\Phi)$, is defined by the allocation rule $x$ given by, for each $i \in [n]$,
\begin{align*}
x_i(\mathbf{v}) = \begin{cases}
\Phi(v^{(1)}) &\text{ if $v_i$ is the highest and $v^{(2)} > a$}\,, \\
1 &\text{ if $v_i$ is the highest and $v^{(2)} = a$}\,, \\
\end{cases}
\end{align*}
and zero otherwise, breaking ties uniformly at random. The payment rule $p: [a,b]^n \to \mathbb{R}_{+}^n$ is determined uniquely from  Myerson's formula such that the resulting mechanism $(x,p)$ is dominant strategy incentive compatible. 
\end{definition}

We call this mechanism \textit{generous} SPA because it behaves like SPA, except in the case when all other non-highest agents have the lowest possible value $a$, then it always allocates (``generously''). We now state the main theorem of this section. 

\begin{theorem}[Optimal Standard Mechanism]\label{thm:char-std-mech}
Fix $n$ and $\lambda \in (0,1]$, and let $\tilde{a} = a/b \in [0,1)$. Define $k_l$ as in Theorem~\ref{thm:char-all-mech-full-main}. Then, the problem admits an optimal minimax $\lambda$-regret standard mechanism $m^*$, depending on $a/b$ as follows.
\begin{itemize}
    \item (Low Relative Support Information) For $a/b \leq k_l$, $m^* = \textnormal{SPA}(\Phi)$ is the same as in  Theorem~\ref{thm:char-all-mech-full-main}.
    \item (High Relative Support Information) For $a/b \geq k_l$, there is a probability distribution $\Phi$ such that $m^* = \textnormal{GenSPA}(\Phi)$.
\end{itemize}

\end{theorem}

Note that by Theorem~\ref{thm:char-all-mech-full-main}, if $a/b \leq k_l$, then SPA with random reserve is optimal in $\mathcal{M}_{\textnormal{all}}$, and it is also standard, so it is immediate that it is also optimal in the class $\mathcal{M}_{\textnormal{std}}$. Similar to Theorem~\ref{thm:char-all-mech-full-main}, Theorem~\ref{thm:char-std-mech} highlights the structural features of our optimal mechanism and is a corollary of Theorem~\ref{app:thm:char-std-mech} in the Appendix which fully characterizes the saddle point in closed form. 

The proof of Theorem~\ref{thm:char-std-mech}  follows a similar outline to that of Theorem~\ref{thm:char-all-mech-full-main}, although the calculations are nontrivial. In particular, we need to derive the expressions of \textit{conditional} distributions of order statistics for arbitrary $F$, taking into account potential ties, which complicate the calculations.\footnote{The existing results on conditional distributions of order statistics assume that $F$ has a density, see e.g.~\cite{DavidNagaraja03-order-statistics}. These results do not apply because we do not make any assumptions on $F$. In fact, the worst case $F$ has point masses.}
In contrast, the regret of  any $(g_u,g_d)$ mechanism (whose class contains all other mechanisms in this paper) depends only on the marginal distributions of the first- and second-order statistics, which are simpler (cf. Proposition~\ref{prop:reg-exp-g}). However, the hardest part is \textit{coming up with the right structural class} $\textnormal{GenSPA}$ that contains the optimal mechanism (within the subclass of standard mechanisms) and is tractable, because our techniques based on solving differential equations can pin down the candidate mechanism only once we fix the mechanism up to a one-dimensional functional parameter. We discuss key technical challenges and give the full proof in Appendix~\ref{app:subsec:main-std-mech}.

\subsection{Minimax $\lambda$-Regret over SPA with random  and deterministic reserve}\label{subsec:main-spa-rand-mech}

We can characterize the minimax $\lambda$-regret mechanism and its corresponding worst-case distribution and performance in the following theorem. 

\begin{theorem}[Optimal SPA with Random Reserve]\label{thm:char-spa-rand-mech}
Fix $n$ and $\lambda \in (0,1]$. Define $k_l$ as in Theorem~\ref{thm:char-all-mech-full-main} and $k_h' = \lambda n / ((1+\lambda)n-1)$. Then, the problem admits a minimax $\lambda$-regret $m^* = \textnormal{SPA}(\Phi^*)$, depending on $a/b$, as follows.
\begin{itemize}
    \item (Low Relative Support Information) For $a/b \leq k_l$, $m^* = \textnormal{SPA}(\Phi^*)$ is the same as in Theorem~\ref{thm:char-all-mech-full-main}.
    \item (High Relative Support Information) For $a/b \geq k_h'$, $\Phi^*$ is a point mass only at $a$, i.e., $m^* = \textnormal{SPA}(a)$ is a SPA with no reserve.
    \item (Moderate Relative Support Information) For $k_l \leq a/b \leq k_h'$, there is $r^* \in [a,b]$ such that $\Phi^*$ has a point mass at $a$ and a density on $[r^*,b]$.
\end{itemize}

\end{theorem}

The second bullet point of Theorem~\ref{thm:char-spa-rand-mech} formalizes the intuition highlighted in the introduction that in the high scale information regime ($a/b$ is close enough to 1), the optimal SPA with random reserve sets no reserve at all. Similar to Theorem~\ref{thm:char-all-mech-full-main}, Theorem~\ref{thm:char-spa-rand-mech} highlights the structural features of our optimal mechanism and is a corollary of Theorem~\ref{app:thm:char-spa-rand-mech} in Appendix~\ref{app:subsec:main-spa-rand-mech} which fully characterizes the saddle point in closed form.  The proof of the moderate information regime of Theorem~\ref{thm:char-spa-rand-mech} is the most challenging. It is different from previous saddle problems because in this case, the increasing condition on the reserve price distribution $\Phi$ is \textit{binding}; if we optimize pointwise, the resulting distribution is not increasing, which is infeasible. We characterize an optimal distribution of reserves using a Lagrangian approach that involves introducing a Lagrange multiplier for the monotonicity constraint and then designing a primal-dual pair that satisfies complementary slackness and Lagrangian optimality. We discuss key technical challenges and give the full proof in Appendix~\ref{app:subsec:main-spa-rand-mech}. 

Lastly, we characterize the optimal SPA with deterministic reserve $
\mathcal{M}_{\textnormal{SPA-det}}$ and SPA with no reserve $\mathcal{M}_{\textnormal{SPA-a}}$. Proposition~\ref{app:prop:sup-regret-spa-r-lmbd} in the Appendix gives the the minimax $\lambda$-regret for $\textnormal{SPA}(r)$ with a fixed deterministic reserve $r$. In particular, it subsumes the problem of choosing the regret-minimizing reserve $r$ as well as computing worst-case regret of SPA without reserve ($r=a$).

\subsection{Performance Separation Between Mechanism Classes}\label{subsec:other-mech-classes-discuss}

\jadelete{Table~\ref{table:maximin-ratio-mech-class-intro-n-4} in the introduction shows the maximin ratio as a function of $a/b$ of all mechanism classes we consider for $n = 4$.}

\jaedit{Figure~\ref{fig:ratio-mech-classes-vary-n} shows the maximin ratio as a function of $a/b$ of all mechanism classes for $n \in \{2,4\}$.} This metric captures the performance of the optimal mechanism. We can see that while $\mathcal{M}_{\textnormal{SPA-det}}$ and $\mathcal{M}_{\textnormal{SPA-a}}$ have the same maximin ratios (so a fixed reserve does not improve over no reserve), there are \textit{strict separations} between $\mathcal{M}_{\textnormal{all}}$, $\mathcal{M}_{\textnormal{std}}$, $\mathcal{M}_{\textnormal{SPA-rand}}$, and $\mathcal{M}_{\textnormal{SPA-a}}$.

\begin{figure}[h!]
	\begin{subfigure}{0.48\linewidth}
    	\centering
\includegraphics[height=0.75\linewidth]{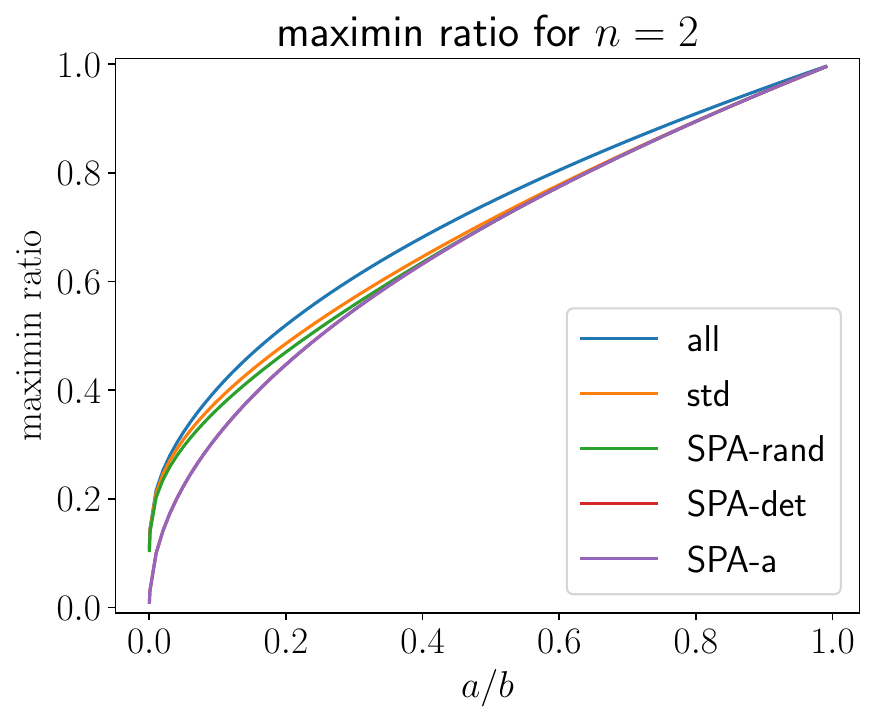}
    	\label{fig:ratio-mech-classes-n-2}
    \end{subfigure}%
	\begin{subfigure}{0.48\linewidth}
		\centering
		\includegraphics[height=0.75\linewidth]{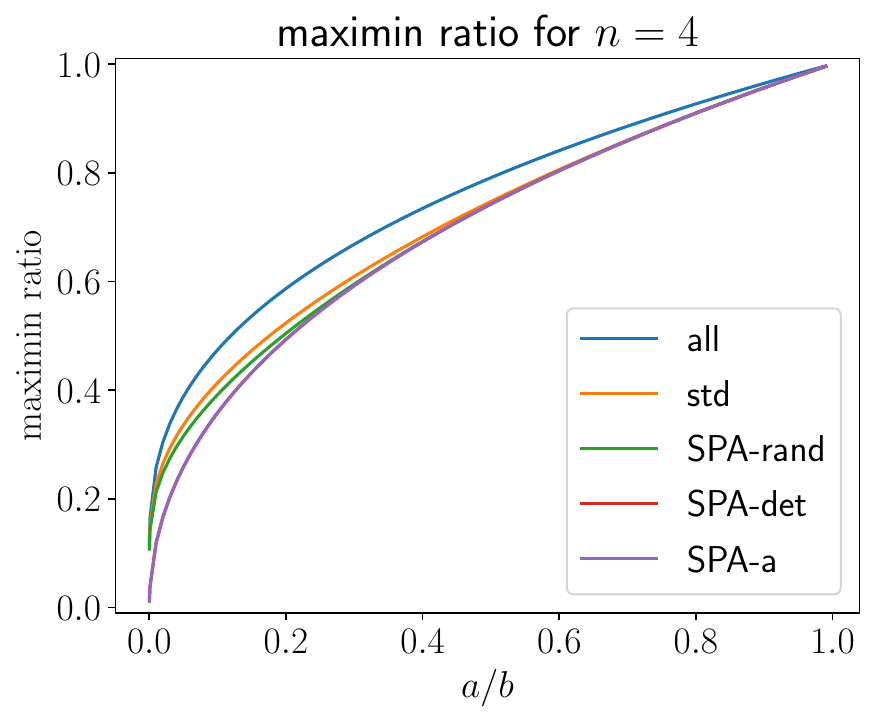}
		\label{fig:ratio-mech-classes-n-4}
	\end{subfigure}
	\caption{Maximin ratio as a function of $a/b$ for $n \in \{2,4\}$.}
\label{fig:ratio-mech-classes-vary-n}
\end{figure}

The gap between $\mathcal{M}_{\textnormal{std}}$ and $\mathcal{M}_{\textnormal{SPA-rand}}$ shows that no SPA is optimal within the class of standard mechanisms, even though the gap is quantitatively small. In contrast, the gap between $\mathcal{M}_{\textnormal{all}}$ and $\mathcal{M}_{\textnormal{std}}$ is significant. This means that in robust settings, \textit{it is important to sometimes allocate to non-highest bidders}. We can see from the plots with $n = 2$ and $n = 4$ that the non-standard gap becomes bigger and dominates all other gaps as $n$ gets large, so this becomes more important with more bidders.

\jadelete{\obcomment{The introduction feels far at this stage...} 

Echoing discussions in the introduction,}

\jaedit{These structural results} shows that there are interesting mechanism classes in DSIC mechanisms beyond SPA in the sense that they are robustly optimal in natural settings. \jaedit{In particular}, SPA is not optimal even within the class of standard mechanisms; $\textnormal{GenSPA}$ is. It is an open question whether $\textnormal{GenSPA}$ will also be useful in other settings as well.
 
\jadelete{
\jacomment{These are not really needed. Just the discussion of $\spa$ vs $\pool$ in the main text is enough.}
We can also use analytical results derived in this section to gain insights into the structure of optimal mechanisms within subclasses. We further discuss this in  Appendix~\ref{app:subsec:other-mech-classes-structure}.
}

\section{Extensions and Conclusion}\label{sec:conclude}

In this paper, we give an explicit characterization of a \textit{robustly} optimal mechanism to sell an item to $n$ buyers knowing only a lower bound and an upper bound of the support of values, where the seller's performance is evaluated in the worst case. Our general framework is broadly applicable to an arbitrary number $n$ of buyers and  several mechanism classes $\mathcal{M}$ and captures both regret and ratio objectives.

Furthermore, we note that it is possible to extend the framework to other classes of distributions. It is possible to show that the minimax $\lambda$-regret we have obtained for the case of i.i.d. distributions (and the corresponding optimal mechanism) does not change if Nature optimizes over broader classes of distributions capturing positive dependence: exchangeable and affiliated values, a common class considered with knowledge of the distributions \citep{MilgromWeber82};  and mixtures of i.i.d. distributions, another common class.  The results also do not change if Nature optimizes over the smaller class of i.i.d. regular distributions.

There are many avenues for future work. This present paper is a step in the more general agenda of \textit{robust mechanism design with partial information}, and it would be interesting to investigate how other forms of side information (such as moments, samples, and shapes of distributions) impact the structure and performance of optimal or near-optimal mechanisms, and the value of such information. Another direction is to consider other benchmarks, especially the second-best benchmark rather than the first-best benchmark considered in this paper.

\newpage

\bibliographystyle{plainnat}
\bibliography{references}

\newpage
\appendix

\pagenumbering{arabic}
\renewcommand{\thepage}{App-\arabic{page}}
\renewcommand{\theequation}{\thesection-\arabic{equation}}
\renewcommand{\thelemma}{\thesection-\arabic{lemma}}
\renewcommand{\theproposition}{\thesection-\arabic{proposition}}
\setcounter{page}{1}
\setcounter{section}{0}
\setcounter{proposition}{0}
\setcounter{lemma}{0}
\setcounter{equation}{0}

\setcounter{footnote}{0}

\begin{center}
 {\Large \textbf{Electronic Companion: 
\\ Robust Auction Design with Support Information \\}
\medskip
\ifx\blind\undefined
Jerry Anunrojwong\footnote{Columbia University, Graduate School of Business. Email: {\tt janunrojwong25@gsb.columbia.edu}}, ~
Santiago R. Balseiro\footnote{Columbia University, Graduate School of Business. Email: {\tt srb2155@columbia.edu}.}, ~ and Omar Besbes\footnote{Columbia University, Graduate School of Business. Email: {\tt ob2105@columbia.edu}.}.
\fi}
\end{center}

\vspace{-4em}

\addcontentsline{toc}{section}{Appendix} 
\part{Appendix} 
\setstretch{1.0}
\parttoc 
\newpage

\section{Proofs for Section~\ref{sec:problem-formulation}}\label{app:sec:problem-formulation}

\begin{proof}[Proof of Proposition~\ref{prop:minimax-regret-exp}]
The definition of $\textnormal{MaximinRatio}(\mathcal{M},\mathcal{F})$ says that it is a solution to
\begin{align*}
\sup_{(x,p) \in \mathcal{M}} \lambda \text{ s.t. } 
\frac{ \bE_{\mathbf{v} \sim \mathbf{F}} \left[ \sum_{i=1}^{n} p_i(\mathbf{v}) \right]   }{ \bE_{\mathbf{v} \sim \mathbf{F}} \left[ \max(\mathbf{v}) \right] }  \leq \lambda \quad \forall \mathbf{F} \in \mathcal{F} \, ,
\end{align*}
or 
\begin{align*}
\sup_{(x,p) \in \mathcal{M}} \lambda \text{ s.t. } 
\bE_{\mathbf{v} \sim \mathbf{F}} \left[ \lambda \max(\mathbf{v}) - \sum_{i=1}^{n} p_i(\mathbf{v}) \right] \leq 0 \quad \forall \mathbf{F} \in \mathcal{F} \, ,
\end{align*}
or 
\begin{align*}
\sup_{(x,p) \in \mathcal{M}} \lambda \text{ s.t. } 
\sup_{\mathbf{F} \in \mathcal{F}} \bE_{\mathbf{v} \sim \mathbf{F}} \left[ \lambda \max(\mathbf{v}) - \sum_{i=1}^{n} p_i(\mathbf{v}) \right] \leq 0 \, .
\end{align*}

That is, the maximin ratio is the highest value of $\lambda$ such that \textit{there exists} $(x,p) \in \mathcal{M}$ such that $\sup_{\mathbf{F} \in \mathcal{F}} \bE_{\mathbf{v} \sim \mathbf{F}} \left[ \lambda \max(\mathbf{v}) - \sum_{i=1}^{n} p_i(\mathbf{v}) \right] \leq 0$. 
Equivalently, it is the highest value of $\lambda$ such that $$R_{\lambda}(\mathcal{M},\mathcal{F}) = \inf_{(x,p) \in \mathcal{M}} \sup_{\mathbf{F} \in \mathcal{F}} \bE_{\mathbf{v} \sim \mathbf{F}} \left[ \lambda \max(\mathbf{v}) - \sum_{i=1}^{n} p_i(\mathbf{v}) \right] \leq 0 \,.\qedhere$$
\end{proof}

\jaedit{
\section{Proofs for Section~\ref{sec:main-all-mech}}\label{app:sec:main-all-mech} 

\subsection{Proofs for Section~\ref{subsec:proof-main-thm}}

This subsection contains the technical details deferred from \S\ref{subsec:proof-main-thm}, the proof of the main theorem.

We first state and derive technical lemmas in \S\ref{subsubsec:technical-lemmas}. We then give a reformulation of our main theorem in terms of the $(g_u,g_d)$ mechanism and the distribution $F^*$ in a saddle point in \S\ref{subsubsec:main-thm-reformulation}. We prove key results supporting the saddle calculation in \S\ref{subsubsec:saddle-key-results} and verify the saddle in \S\ref{subsubsec:saddle-verify}, thus proving the main theorem.

\subsubsection{Technical Lemmas}\label{subsubsec:technical-lemmas}

\begin{lemma}\label{lem:integral}
Let $\phi_0$ be a constant, then for any positive integer $n$ and $v \geq a$ we have the identity
\begin{align*}
\int_{t=a}^{t=v} \frac{(t-a)^{n-1}}{(t-\phi_0)^n} dt = \log\left( \frac{v-\phi_0}{a-\phi_0} \right) - \sum_{k=1}^{n-1} \frac{(v-a)^k}{k (v-\phi_0)^k} = \sum_{k=n}^{\infty} \frac{(v-a)^k}{k (v-\phi_0)^k} \, .
\end{align*}
\end{lemma}

\begin{proof}[Proof of Lemma~\ref{lem:integral}]
We first check the equality of the first and the second expressions. Note that both expressions are zero when $v = a$. It is then sufficient to check that the derivatives of the two expressions agree. The derivative of the second expression is
\begin{align*}
    \frac{1}{v-\phi_0} - \sum_{k=1}^{n-1} \frac{1}{k} k \left( \frac{v-a}{v-\phi_0} \right)^{k-1} \frac{(a-\phi_0)}{(v-\phi_0)^2} \\
    = \frac{1}{v-\phi_0} - \frac{(a-\phi_0)}{(v-\phi_0)^2} \frac{1 - \left(\frac{v-a}{v-\phi_0} \right)^{n-1}}{1-\frac{v-a}{v-\phi_0}} = \frac{1}{v-\phi_0} \left( \frac{v-a}{v-\phi_0} \right)^{n-1} = \frac{(v-a)^{n-1}}{(v-\phi_0)^n},
\end{align*}
which is the derivative of the first expression.

Now we check the third expression. We have the Taylor Series
\begin{align*}
    - \log(1-x) = \sum_{k=1}^{\infty} \frac{x^k}{k}.
\end{align*}
Substituting $x = \frac{v-a}{v-\phi_0}$ gives
\begin{align*}
    \log \left( \frac{v-\phi_0}{a-\phi_0} \right) = - \log \left( 1 - \frac{v-a}{v-\phi_0} \right) = \sum_{k=1}^{\infty} \frac{1}{k} \left( \frac{v-a}{v-\phi_0} \right)^{k}.
\end{align*}
We see that the first $n-1$ terms of $k$ cancel out, and we get the third expression.
\end{proof}

\begin{lemma}\label{lem:integral2} The following holds
    \begin{align*}
        & \int_{t=a}^{t=v} \left[ \frac{(t-a)^{n-1}}{(t-\phi_0)^{n}} - (1-\lambda) \frac{(t-a)^n}{(t-\phi_0)^{n+1}} \right] dt = \frac{(v-a)^{n}}{n(v-\phi_0)^n} + \lambda \sum_{k=n+1}^{\infty} \frac{(v-a)^k}{k(v-\phi_0)^k} \\
        &\qquad =\lambda \log \left( \frac{v-\phi_0}{a-\phi_0} \right) - \lambda \sum_{k=1}^{n-1} \frac{(v-a)^{k}}{k(v-\phi_0)^{k}} + (1-\lambda) \frac{(v-a)^{n}}{n(v-\phi_0)^{n}}\,.
    \end{align*}
\end{lemma}
\begin{proof}[Proof of Lemma~\ref{lem:integral2}]

We apply Lemma~\ref{lem:integral} to obtain
\begin{align*}
\int_{t=a}^{t=v} \frac{(t-a)^{n-1}}{(t-\phi_0)^n} dt = \log\left( \frac{v-\phi_0}{a-\phi_0} \right) - \sum_{k=1}^{n-1} \frac{(v-a)^k}{k (v-\phi_0)^k} = \sum_{k=n}^{\infty} \frac{(v-a)^k}{k (v-\phi_0)^k}  \\
\int_{t=a}^{t=v} \frac{(t-a)^{n}}{(t-\phi_0)^{n+1}} dt = \log\left( \frac{v-\phi_0}{a-\phi_0} \right) - \sum_{k=1}^{n} \frac{(v-a)^k}{k (v-\phi_0)^k} = \sum_{k=n+1}^{\infty} \frac{(v-a)^k}{k (v-\phi_0)^k} \, .
\end{align*}
Therefore,
\begin{align*}
&\int_{t=a}^{t=v} \left[ \frac{(t-a)^{n-1}}{(t-\phi_0)^{n}} - (1-\lambda) \frac{(t-a)^n}{(t-\phi_0)^{n+1}} \right] dt\\ &\qquad = \lambda \log \left( \frac{v-\phi_0}{a-\phi_0} \right) - \lambda \sum_{k=1}^{n-1} \frac{(v-a)^{k}}{k(v-\phi_0)^{k}} + (1-\lambda) \frac{(v-a)^{n}}{n(v-\phi_0)^{n}}   \\
&\qquad = \frac{(v-a)^{n}}{n(v-\phi_0)^n} + \lambda \sum_{k=n+1}^{\infty} \frac{(v-a)^k}{k(v-\phi_0)^k}\,,
\end{align*}
where the last equation follows from the Taylor series for the logarithm.
\end{proof}

\begin{lemma}\label{lem:int-by-part}
If $h$ is a differentiable function, then
\begin{align*}
\int_{w \in [a,b]} h(w) dG(w) = h(a) + \int_{w \in [a,b]} h'(w) (1-G(w)) dw  \, .
\end{align*}
\end{lemma}
\begin{proof}[Proof of Lemma~\ref{lem:int-by-part}]
\begin{align*}
\int_{w \in [a,b]} h(w) dG(w) &= \int_{w \in [a,b]} \left( h(a) + \int_{\tilde{w}=a}^{\tilde{w}=w} h'(\tilde{w}) d\tilde{w} \right) dG(w)\\
&= h(a) + \int_{\tilde{w}=a}^{\tilde{w}=b} h'(\tilde{w})\int_{w \in (\tilde{w},b]} dG(w) dw \\
&= h(a) + \int_{\tilde{w}=a}^{\tilde{w}=b} h'(\tilde{w}) (G(b)-G(\tilde{w})) d\tilde{w}\\
&= h(a) + \int_{\tilde{w}=a}^{\tilde{w}=b} h'(\tilde{w}) (1-G(\tilde{w})) d\tilde{w}\,.\qedhere
\end{align*}
\end{proof}

Now, we give a formal proposition that $(g_u,g_d)$ mechanisms and convex combinations of $\{\textnormal{SPA}(r),$ $\textnormal{POOL}(\tau)\}$ are \textit{almost} equivalent representations of the same mechanism class in the sense that one can be converted to another. 

\begin{proposition}\label{prop:gu-gd-rep} 
We have the following correspondence between the $(g_u,g_d)$ mechanisms arising in our main theorem (Theorem~\ref{app:thm:char-all-mech-gu-gd}) and convex combinations of SPAs and POOLs.

\begin{itemize}
\item[(1)] A mechanism is a $(g_u,g_d)$ mechanism with $g_u(v) \in [0,1]$ increasing in $v$, $g_u(a) = 0$, and $g_d(v) = 0$ for all $v$ if and only if it is $\textnormal{SPA}(\Phi)$, $\Phi$ has measure 1, and $\Phi(v) = g_u(v)$.
\item[(2)] A mechanism is a $(g_u,g_d)$ mechanism with $g_u(v) \in [0,1]$ increasing in $v$, $g_u(a)=g_d(a)=1/n$, and $g_u(v) + (n-1)g_d(v) = 1$ for all $v$ if and only if is $\textnormal{POOL}(\Psi)$, $\Psi$ has measure 1, and $\Psi(v) = 1-n g_d(v)$. 
\item[(3)] A mechanism is a $(g_u,g_d)$ mechanism with $g_u(v) \in [0,1]$ increasing in $v$ for $v \in [a,b]$, $g_u(a) = g_d(a) := \alpha$, $g_d(v) = \alpha$ for $v \in [a,v^*]$ for some constant $v^* \in [a,b]$, $g_u(v) + (n-1)g_d(v) = 1$ for $v \in [v^*,b]$ if and only if it is a randomization over $\textnormal{SPA}(\Phi)$, with $\Phi$ supported on $[a,v^*]$ and $\textnormal{POOL}(\Psi)$ with $\Psi$ supported on $[v^*,b]$. Furthermore, their cumulative probabilities are given by $\Phi(v) = g_u(v)-\alpha$ for $v \in [a,v^*]$ and $\Psi(v) = n(\alpha-g_d(v))$ for $v \in [v^*,b]$.
\end{itemize}

\end{proposition}

\begin{proof}[Proof of Proposition~\ref{prop:gu-gd-rep}]
    Note that the $(g_u,g_d)$ representation of $\spa(r)$ is $g_u(v) = \1(v \geq r)$ and $g_d(v) = 0$, and the $(g_u,g_d)$ representation of $\pool(r)$ is $g_u(v) = \frac{1}{n} + \frac{n-1}{n} \1(v \geq r)$ and $g_d(v) = \frac{1}{n} - \frac{1}{n} \1(v \geq r)$. All three cases follow from computing the convex combination of these.

    (1) is straightforward. For (2), the $(g_u,g_d)$ representation of $\pool(\Psi)$ is 
    \begin{align*}
        g_u(v) &= \int\left( \frac{1}{n} + \frac{n-1}{n} \1(v \geq r) \right) d\Psi(r) = \frac{1}{n} + \frac{n-1}{n} \Psi(v) \, , \\
        g_d(v) &= \int\left( \frac{1}{n} - \frac{1}{n} \1(v \geq r) \right) d\Psi(v) = \frac{1}{n} - \frac{1}{n} \Psi(v) \, .
    \end{align*}
    We can then see that $g_u(a) = g_d(a) = 1/n$, $\Psi(v) = 1 - n g_d(v)$ and $g_u(v) + (n-1) g_d(v) = 1$. Conversely, given this $(g_u,g_d)$, we can let $\Psi(v) = 1 - n g_d(v)$ giving a valid $\pool(\Psi)$.

    For (3), the $(g_u,g_d)$ mechanism is
\begin{align*}
    g_u(v) &= \int \1(v \geq r) d\Phi(r) + \int \left( \frac{1}{n} + \frac{n-1}{n} \1(v \geq \tau) \right) d\Psi(r) &&= \Phi(v) + \frac{1}{n} | \Psi | + \frac{n-1}{n} \Psi(v) \, , \\
    g_d(v) &= \int 0 d\Phi(r) + \int \left( \frac{1}{n} - \frac{1}{n} \1(v \geq \tau) \right) d\Psi(v) &&= \frac{1}{n} | \Psi | - \frac{1}{n} \Psi(v) \, .
\end{align*}

We therefore have a formula that transforms $(\Phi,\Psi)$ to $(g_u,g_d)$. 
From these formula, we immediately see that $g_u(a) = g_d(a)$; we let this be $\alpha$. We also see that $g_u(v) + (n-1) g_d(v) = \Phi(v) + |\Psi|$ is increasing in $v$, while $g_d(v) = \frac{1}{n}(|\Psi|-\Psi(v))$ is decreasing in $v$, because $\Phi$ and $\Psi$ are increasing functions. 

Conversely, assume that $(g_u,g_d)$ has these properties. We will show that we can invert these formulas and find the corresponding $(\Phi,\Psi)$. From $g_u(v) = \Phi(v) + \frac{1}{n} |\Psi| +\frac{n-1}{n} \Psi(v)$, setting $v = a$ gives $\alpha = g_u(a) = \frac{1}{n} |\Psi|$, so $|\Psi| = n\alpha$, and $|\Phi| = 1-|\Psi| = 1-n\alpha$. From  $g_d(v) = \frac{1}{n}(|\Psi|-\Psi(v))$, we get $\Psi(v) = |\Psi|-n g_d(v) = n(\alpha-g_d(v))$, and from $g_u(v) = \Phi(v) + \frac{1}{n}|\Psi| + \frac{n-1}{n} \Psi(v) = \Phi(v) + \frac{1}{n}(n\alpha) + \frac{n-1}{n} \cdot n(\alpha-g_d(v)) = \Phi(v) + n\alpha -(n-1)g_d(v)$, we get $\Phi(v) = g_u(v) + (n-1) g_d(v) - n\alpha$.
\end{proof}

\subsubsection{Reformulation of the Main Theorem}\label{subsubsec:main-thm-reformulation}

Armed with Proposition~\ref{prop:gu-gd-rep}, we can reformulate our main theorem (Theorem~\ref{thm:char-all-mech-full-main}) as follows.

\begin{theorem}[Main Theorem in $(g_u,g_d)$]\label{app:thm:char-all-mech-gu-gd}
Fix $n$ and $\lambda \in (0,1]$. Define $k_l \in (0,1)$ as a unique solution to 
\begin{align*}
\lambda \int_{t=k_l}^{t=1} \frac{(t-k_l)^{n-1}}{t^n} dt =  (1-k_l)^{n-1},
\end{align*}
and if $n = 1$, define $k_h = 1$ and if $n \geq 2$, define $k_h \in (0,1)$ to be a unique solution to
\begin{align*}
    \int_{t=k_h}^{t=1} \left[ \frac{(t-k_h)^{n-1}}{t^{n}} - (1-\lambda) \frac{(t-k_h)^{n}}{t^{n+1}}\right] dt = (1-k_h)^{n} .
\end{align*}

Then we have $R_{\lambda}(m,\mathbf{F}^*) \leq R_{\lambda}(m^*,\mathbf{F}^*) \leq R_{\lambda}(m^*,\mathbf{F})$ for any $m \in \mathcal{M}_{\textnormal{all}}$ and $\mathbf{F} \in \mathcal{F}_{\textnormal{iid}}$, where $m^*$ and $\mathbf{F}^*$ (which is $n$ i.i.d. with marginal $F^*$) is defined depending on the value of $a/b$ as follows.
\begin{itemize}
\item Suppose $a/b \leq k_l$, and let $r^* = k_l b$. We define $m^*$ as a $(g_u^*,g_d^*)$ mechanism with $g_u^*(v) = \Phi^*(v), g_d^*(v) = 0$,  where
\begin{align*}
\Phi^*(v) = g_u^*(v) = \lambda  \frac{v^{n-1}}{(v-r^*)^{n-1}} \int_{t=r^*}^{t=v} \frac{(t-r^*)^{n-1}}{t^n} dt \, , 
\end{align*}
and 
\begin{align*}
F^*(v) = \begin{cases}
0 &\text{ if } v \in [a,r^*] \\
1 - \frac{r^*}{v} &\text{ if } v \in [r^*, b) \\
1 &\text{ if } v = b \, .
\end{cases}
\end{align*}

\item Suppose $a/b \geq k_h$, and let $\phi_0 = (a-k_h b)/(1-k_h) \in [0,a]$. We define $m^*$ as a $(g_u^*,g_d^*)$ mechanism with
\begin{align*}
g_u^*(v) &= \frac{1}{n} + \lambda \left( \frac{v-\phi_0}{v-a} \right)^{n} \int_{t=a}^{t=v} \frac{(t-a)^{n}}{(t-\phi_0)^{n+1}} dt  \\
g_d^*(v) &= \frac{1-g_u^*(v)}{n-1} \, ,
\end{align*}
and
\begin{align*}
F^*(v) = \begin{cases}
1 - \frac{a-\phi_0}{v-\phi_0} &\text{ if } v \in [a,b) \\
1 &\text{ if } v = b \, .
\end{cases}
\end{align*}

\item For $k_l \leq a/b \leq k_h$, we define $m^*$ as a $(g_u^*,g_d^*)$ mechanism with 
\begin{align*}
g_u^*(v) = \begin{cases}
\alpha + \lambda \left( \frac{v}{v-a} \right)^{n-1} \int_{t=a}^{t=v} \frac{(t-a)^{n-1}}{t^{n}} dt
&\text{ for } v \in [a,v^*] \\
\frac{v^n}{(v-a)^n} \left[\frac{(b-a)^n}{b^n} - \int_{t=v}^{t=b} \left[ \frac{(t-a)^{n-1}}{t^n} - (1-\lambda) \frac{(t-a)^n}{t^{n+1}}  \right]dt   \right] &\text{ for } v \in [v^*,b] \, ,
\end{cases}
\end{align*}
and 
\begin{align*}
g_d^*(v) = \begin{cases}
\alpha &\text{ if } v \in [a,r^*] \\
\frac{1-g_u^*(v)}{n-1} &\text{ if } v \in [r^*,b] \, ,
\end{cases}
\end{align*}
and 
\begin{align*}
F^*(v) = \begin{cases}
1 - \frac{a}{v} &\text{ if } v \in [a,b) \\
1 &\text{ if } v = b \, ,
\end{cases}
\end{align*}
where $(v^*,\alpha)$ is the unique solution to
\begin{align*}
\frac{(r^*-a)^{n-1}}{(r^*)^{n-1}} ( 1 - n \alpha ) &= \lambda \int_{t=a}^{t=r^*} \frac{(t-a)^{n-1}}{t^n} dt \\
\frac{(b-a)^{n}}{b^{n}}  - \frac{(r^*-a)^{n}}{(r^*)^{n}} (1 - (n-1)\alpha) &= \int_{t=r^*}^{t=b} \left[ \frac{(t-a)^{n-1}}{t^n} - (1-\lambda) \frac{(t-a)^{n}}{t^{n+1}}  \right] dt .
\end{align*}
\end{itemize}
\end{theorem}

\subsubsection{Key Results Supporting the Saddle Calculation}\label{subsubsec:saddle-key-results}

\begin{proof}[Proof of Proposition~\ref{prop:reg-exp-g}]

We first derive the (Regret-$\mathbf{F}$) expression, i.e. the expected regret of a $(g_u,g_d)$ mechanism under an arbitrary joint distribution $\mathbf{F}$, assuming that $g_u$ and $g_d$ are continuous everywhere and differentiable everywhere except a finite number of points. Let $\mathbf{F}_n^{(1)}$ and $\mathbf{F}_{n}^{(2)}$ be the distributions of $v^{(1)}$ and $v^{(2)}$, the highest and second-highest entry of $\mathbf{v}$, respectively.

From Myerson's lemma,
\begin{align*}
p_i(\mathbf{v}) = v_i x_i(\mathbf{v}) - \int_{\tilde{v}_i=a}^{\tilde{v}_i=v_i} x_i(\tilde{v}_i,\mathbf{v}_{-i}) d\tilde{v}_i \, .
\end{align*}
the allocation rule $(g_u,g_d)$ gives 
\begin{align*}
p_i(\mathbf{v}) = \begin{cases}
v_i g_u(v_i) - (v^{(2)}-a) g_d(v^{(2)}) - \int_{t=v^{(2)}}^{t=v_i} g_u(t) dt &\text{ if $v_i$ is the highest and} \\ & \text{ $v^{(2)}$ is the second-highest}\,, \\
a g_d(v^{(1)}) &\text{ if $v_i$ is not the highest} \, ,
\end{cases}
\end{align*}
so the pointwise regret is
\begin{align*}
v^{(1)}(\lambda-g_u(v^{(1)})) - (n-1) a g_d(v^{(1)}) + (v^{(2)}-a) g_d(v^{(2)}) + \int_{t=v^{(2)}}^{t=v^{(1)}} g_u(t) dt \, .
\end{align*}

Therefore, by Lemma~\ref{lem:int-by-part},
\begin{align*}
&\bE[v^{(1)}(\lambda-g_u(v^{(1)}) ) - (n-1)a g_d(v^{(1)})] \\
&= a(\lambda-g_u(a)-(n-1)g_d(a)) + \int_{v \in [a,b]} (\lambda-g_u(v)-vg_u'(v)-(n-1)a g_d'(v)) (1-\mathbf{F}_n^{(1)}(v)) dv \, .
\end{align*}

Now we compute the second term.
\begin{align*}
\bE[(v^{(2)}-a) g_d(v^{(2)}) ] &= (a-a) g_d(a) + \int_{v \in [a,b]} (g_d(v)+(v-a)g_d'(v))(1-\mathbf{F}_n^{(2)}(v)) dv \\
&= \int_{v \in [a,b]} (g_d(v)+(v-a)g_d'(v))(1-\mathbf{F}_n^{(2)}(v)) dv \, .
\end{align*}

Lastly, we compute the third term
\begin{align*}
\bE\left[ \int_{t=v^{(2)}}^{t=v^{(1)}} g_u(t) dt \right] &= \bE\left[ \int_{v \in [a,b]} g_u(v) \1(v^{(2)} < v \leq v^{(1)}) dv \right] \\
&= \int_{v \in [a,b]} g_u(v) \Pr(v^{(2)} < v \leq v^{(1)}) dv \\
&= \int_{v \in [a,b]} g_u(v) (\mathbf{F}_n^{(2)}(v) - \mathbf{F}_n^{(1)}(v)) dv \, .
\end{align*}
 
Therefore, the regret is
\begin{align*}
a(\lambda-g_u(a)-(n-1)g_d(a)) + \int_{v \in [a,b]} (\lambda-g_u(v)-vg_u'(v)-(n-1)a g_d'(v)) (1-\mathbf{F}_n^{(1)}(v)) dv \\
+ \int_{v \in [a,b]} (g_d(v)+(v-a)g_d'(v))(1-\mathbf{F}_n^{(2)}(v)) dv  
+\int_{v  \in [a,b]} g_u(v) (\mathbf{F}_n^{(2)}(v)-\mathbf{F}_n^{(1)}(v)) dv \, .
\end{align*}

Rearranging this gives the (Regret-$\mathbf{F}$) expression.
\begin{align*}
R(g,\mathbf{F}) = a(\lambda-g_u(a)-(n-1)g_d(a)) + \int_{v \in [a,b]} (\lambda-g_u(v)+ g_d(v) -vg_u'(v)+ (v-na) g_d'(v))  \\
+  \int_{v \in [a,b]} \left( - \lambda -(n-1) (g_u(v)-g_d(v)) + v(g_u'(v)+(n-1) g_d'(v)) \right) \mathbf{F}_n^{(1)}(v)dv \\
+ \int_{v \in [a,b]} n(g_u(v)-g_d(v)-(v-a)g_d'(v))   \mathbf{F}_{n-1}^{(1)}(v)  dv \, .
\end{align*}

To derive (Regret-$F$), the expected regret expression with i.i.d. $F$, substitute $\mathbf{F}_n^{(1)}(v) = F(v)^{n}$ and $\mathbf{F}_n^{(2)}(v) = n F(v)^{n-1} - (n-1) F(v)^{n}$. \end{proof}

\jacomment{We don't really need the (Regret-$g$) expression because it is only used to motivate the guessing of $F^*$ which we omit (and in any case is already done in \S\ref{subsubsec:proof-high-info} in the $\pool$ case. It is also convenient for deriving expressions of the expected regret at the saddle point, but not strictly necessary. It is no longer in the theorem statement, so we omit the proof as well.}

\jadelete{
Now we derive the (Regret-$g$) expression, assuming that $g_u$ and $g_d$ are arbitrary but $\mathbf{F}$ is i.i.d. with marginal $F$ with a density $F'$ in the interior. We start with the expected pointwise regret expression
\begin{align*}
\bE\left[ v^{(1)}(\lambda-g_u(v^{(1)})) - (n-1) a g_d(v^{(1)}) + (v^{(2)}-a) g_d(v^{(2)}) + \int_{t=v^{(2)}}^{t=v^{(1)}} g_u(t) dt \right] \, .
\end{align*}
The third term is still the same
\begin{align*}
\bE\left[  \int_{t=v^{(2)}}^{t=v^{(1)}} g_u(t) dt \right] &= \int_{v \in [a,b]} g_u(v) (\mathbf{F}_n^{(2)}(v) - \mathbf{F}_n^{(1)}(v) ) dv \\
&= \int_{v \in [a,b]} g_u(v) (n F(v)^{n-1} - n F(v)^n ) dv \, .
\end{align*}

We will write the point mass of $F$ at $b$ as $f_b := F(\{b\}) = 1-F(b^-)$ for convenience. 
Now for the first term, we know that the probability that $v^{(1)} = b$ is
\begin{align*}
\Pr(v^{(1)} = b) = 1 - \Pr(v^{(1)} < b) = 1 - \prod_{i=1}^{n} \Pr(v_i < b) = 1 - F(b^-)^n = 1-(1-f_b)^n \, .
\end{align*}
The probability that $v^{(1)} = a$ is
\begin{align*}
\Pr(v^{(1)} = a) = \prod_{i=1}^{n} \Pr(v_i = a) = F(a)^n \, .
\end{align*}
For $v \in (a,b)$, $F^{(1)}$ has a density given by
\begin{align*}
f^{(1)}(v) = n F(v)^{n-1} F'(v) \, .
\end{align*}
Now,
\begin{align*}
\Pr(v^{(2)} = b) = \Pr(\text{ at least 2 of the $n$ $v$'s are $b$ }) \\
= 1- \Pr(\text{ exactly 0 of the $n$ $v$'s are $b$ }) - \Pr(\text{ exactly 1 of the $n$ $v$'s are $b$ }) \\
= 1 - (1-f_b)^n - n f_b(1-f_b)^{n-1} \\
= 1 - (1-f_b)^{n-1} (1+(n-1) f_b ) \, ,
\end{align*}
and
\begin{align*}
\Pr(v^{(2)} = a) = \Pr(\text{ all are $a$ }) + \Pr(\text{ $n-1$ are $a$, $1$ are $>a$}) \\
= F(a)^n + n F(a)^{n-1} (1-F(a)) = F(a)^{n-1} (1+(n-1) F(a)) \, .
\end{align*}
For $v \in (a,b)$, $F^{(2)}$ has a density given by
\begin{align*}
f^{(2)}(v) = n (n-1) F(v)^{n-2} (1-F(v)) F'(v) \, ,
\end{align*}
and CDF given by
\begin{align*}
F^{(2)}(v) = n F(v)^{n-1} - (n-1) F(v)^n \, .
\end{align*}

We then have
\begin{align*}
&\bE\left[ v^{(1)}(\lambda-g_u(v^{(1)})) - (n-1) a g_d(v^{(1)})  \right] \\
&= ( a(\lambda-g_u(a)) - (n-1)ag_d(a) ) \Pr(v^{(1)}=a) + ( b(\lambda-g_u(b)) - (n-1)a g_d(b) ) \Pr(v^{(1)} = b)  \\
&+ \int_{v \in [a,b]} ( v(\lambda-g_u(v)) - (n-1)a g_d(v) ) f^{(1)}(v) dv \\
&= ( a(\lambda-g_u(a)) - (n-1)ag_d(a) ) F(a)^n + ( b(\lambda-g_u(b)) - (n-1)a g_d(b) ) (  1 - ( 1 - f_b )^n )  \\
&+ \int_{v \in [a,b]} ( v(\lambda-g_u(v)) - (n-1)a g_d(v) ) n F(v)^{n-1} F'(v) dv \, ,
\end{align*}
and
\begin{align*}
&\bE\left[  (v^{(2)}-a) g_d(v^{(2)})  \right] \\
&= (a-a) g_d(a) \Pr(v^{(2)} = a) + (b-a) g_d(b) \Pr(v^{(2)} = b) + \int_{v=a}^{v=b}  (v-a) g_d(v) f^{(2)}(v) dv \\
&= (b-a) g_d(b) ( 1 - (1-f_b)^{n-1} (1+(n-1) f_b )) ) + \int_{v=a}^{v=b} (v-a) g_d(v) n (n-1) F(v)^{n-2} (1-F(v)) F'(v) dv \, .
\end{align*}

Therefore, 
\begin{align*}
R(g,F) &= ( a(\lambda-g_u(a)) - (n-1)ag_d(a) ) F(a)^n + (  b(\lambda-g_u(b)) - (n-1)a g_d(b) ) (  1 - ( 1 - f_b )^n )  \\
&+ (b-a) g_d(b) ( 1 - (1-f_b)^{n-1} (1+(n-1) f_b )) ) \\
&+ \int_{v=a}^{v=b}   \left[ (n-1)(v-a)g_d(v) (1-F(v)) + (v (\lambda-g_u(v) ) - (n-1)a g_d(v)) F(v) \right]n F(v)^{n-2} F'(v) dv \\
&+ \int_{v=a}^{v=b}  g_u(v) (n F(v)^{n-1} - n F(v)^n ) dv \, ,
\end{align*}
or
\begin{align*}
R(g,F) &= ( a(\lambda-g_u(a)) - (n-1)ag_d(a) ) F(a)^n + (  b(\lambda-g_u(b)) - (n-1)a g_d(b) ) (  1 - ( 1 - f_b )^n )  \\
&+ (b-a) g_d(b) ( 1 - (1-f_b)^{n-1} (1+(n-1) f_b )) ) \\
&+ \int_{v=a}^{v=b}  \lambda v n F(v)^{n-1} F'(v)  - g_u(v) n v F(v)^{n-1} F'(v)  \\
&+ \int_{v=a}^{v=b} \left\{ (v-a)(1-F(v)) - a F(v) \right\} g_d(v) n(n-1) F(v)^{n-2} F'(v) + g_u(v) (n F(v)^{n-1} - n F(v)^n ) dv \, .
\end{align*}

By integration by parts, the first term of the third line is $\lambda$ times
\begin{align*}
\int_{v=a}^{v=b} v n F(v)^{n-1} F'(v) dv = b F(b^-)^n - aF(a)^n - \int_{v=a}^{v=b} F(v)^n dv \, ,
\end{align*}
where $F(b^-) = (1-f_b)$. Substituting this in gives
\begin{align*}
R(g,F) &= \lambda b - a \left( g_u(a) + (n-1) g_d(a) \right) F(a)^n - \left( b g_u(b) + (n-1)a g_d(b) \right) \left( 1 -(1-f_b)^n \right) \\
&+ (b-a) g_d(b) \left( 1-(1-f_b)^{n-1} (1+(n-1)f_b) \right) \\
&+ \int_{v=a}^{v=b} -\lambda F(v)^n + g_u(v) n F(v)^{n-1} (1-F(v) - vF'(v) ) \\
&+ \int_{v=a}^{v=b} g_d(v) n(n-1) F(v)^{n-2} F'(v) \left\{ (v-a)(1-F(v)) - a F(v) \right\} dv\,.\qedhere
\end{align*}
}

\begin{proof}[Proof of Proposition~\ref{prop:nature-saddle-foc-soc}]

We use the following (Regret-$F$) expression for $\lambda$-regret
\begin{align*}
R(g,F) = a(\lambda-g_u(a)-(n-1)g_d(a)) + \int_{v \in [a,b]} (\lambda-g_u(v)+ g_d(v) -vg_u'(v)+ (v-na) g_d'(v))  \\
+  \int_{v \in [a,b]} \left( -\lambda -(n-1) (g_u(v)-g_d(v)) + v(g_u'(v)+(n-1) g_d'(v)) \right) F(v)^n dv \\
+ \int_{v \in [a,b]} n(g_u(v)-g_d(v)-(v-a)g_d'(v))   F(v)^{n-1}  dv
\end{align*}

In Nature's saddle, we fix the mechanism $(g_u,g_d)$ and optimize over $F$.  The integral expression is separable over $F(v)$ for $v \in (a,b)$.  Here we will assume that the optimization is done pointwise. 

The first order condition on $F$ on the regret pointwise is
\begin{align*}
 \left( -\lambda -(n-1) (g_u^*(v)-g_d^*(v)) + v(g_u'^*(v)+(n-1) g_d'^*(v)) \right) \cdot n F(v)^{n-1}   \\
+ n(g_u(v)-g_d(v)-(v-a)g_d'(v))  \cdot (n-1) F(v)^{n-2}  = 0
\end{align*}

Nature's saddle states that over all $F$,  $F^*$ maximizes the $\lambda$-regret.  If pointwise optimization is valid, then $F^*$ must satisfy the above FOC equation. Dividing both sides by $n F^*(v)^{n-2}$ gives (\ref{eqn:foc}) as required.

For $F^*$ to be maximizing, we also need the second-order conditions to hold, namely,  that the second derivative with respect to $F(v)$ evaluated at $F^*(v)$ is negative\footnote{Note that if $n = 2$ the last term disappear, so we can write $F^*(v)^{n-3}$ there with the understanding that the entire term becomes zero for $n =2$.}:
\begin{align*}
 \left( -\lambda -(n-1) (g_u^*(v)-g_d^*(v)) + v(g_u'^*(v)+(n-1) g_d'^*(v)) \right) \cdot n (n-1) F^*(v)^{n-2}   \\
+ n(g_u(v)-g_d(v)-(v-a)g_d'(v))  \cdot (n-1) (n-2) F^*(v)^{n-3}  < 0
\end{align*}
or
\begin{align*}
 &\left( -\lambda -(n-1) (g_u^*(v)-g_d^*(v)) + v(g_u'^*(v)+(n-1) g_d'^*(v)) \right)   F^*(v) \\
 &\qquad + (n-2)(g_u(v)-g_d(v)-(v-a)g_d'(v))    < 0
\end{align*}
but from the (\ref{eqn:foc}) equality that we have just derived,
\begin{align*}
 &\left( -\lambda -(n-1) (g_u^*(v)-g_d^*(v)) + v(g_u'^*(v)+(n-1) g_d'^*(v)) \right)  F^*(v) \\
 &\qquad + (n-2)(g_u(v)-g_d(v)-(v-a)g_d'(v))   \\
&=   \left( -\lambda -(n-1) (g_u^*(v)-g_d^*(v)) + v(g_u'^*(v)+(n-1) g_d'^*(v)) \right)   F^*(v) \\
& \qquad + (n-1)(g_u^*(v)-g_d^*(v)-(v-a)g_d'^*(v))  
- (g_u^*(v)-g_d^*(v)-(v-a)g_d'^*(v))  \\
&= - (g_u^*(v)-g_d^*(v)-(v-a)g_d'^*(v)) \,.
\end{align*}
Therefore,  our condition reduces to (\ref{eqn:soc}), as required. \end{proof}

\subsubsection{Verification of the Saddle}\label{subsubsec:saddle-verify}

The last step of the proof is to verify the (\ref{eqn:foc}) and (\ref{eqn:soc}) conditions for the $((g_u^*,g_d^*),F^*)$ pair given in Theorem~\ref{app:thm:char-all-mech-gu-gd}. To do this, we first show that the $(g_u^*,g_d^*)$ satisfies a certain ordinary differential equation (ODE).

\begin{proposition}\label{prop:opt-g-ode}
Define $(v^*, \alpha) = (b, 0)$ in the $a/b \leq k_l$ regime, $(v^*, \alpha) = (a, 1/n)$ in the $a/b \geq k_h$ regime, and $(v^*,\alpha)$ be defined as stated in Theorem~\ref{app:thm:char-all-mech-gu-gd} in the the $k_l \leq a/b \leq k_h$ regime. Also define $\phi_0$ as in Theorem~\ref{app:thm:char-all-mech-gu-gd} in the $a/b \geq k_h$ regime, and $\phi_0 = 0$ in other regimes. Let $(g_u^*,g_d^*)$ be given as in Theorem~\ref{app:thm:char-all-mech-gu-gd}. Then $(g_u^*,g_d^*)$ satisfies $g_u^*(v) = g^*(v)$, $g_d^*(v) = \alpha$ for $v \in [a,v^*]$ and $g_d^*(v) = (1-g^*(v))/(n-1)$ for $v \in [v^*,b]$, and $g^*$ is an increasing continuous function that satisfies the ODE
\begin{align*}
(g^*)'(v) + \frac{(n-1)r^*}{v(v-r^*)}(g^*(v)-\alpha) &= \frac{\lambda}{v} &\text{ for } v \in (r^*,v^*) \label{eqn:ode-g-1} \tag{ODE-$g$-1} \\
(g^*)'(v) + \frac{n(a-\phi_0)}{(v-\phi_0)(v-a)} g^*(v) &= \frac{1}{v-a} - \frac{1-\lambda}{v-\phi_0} &\text{ for } v \in (v^*,b) \, . \label{eqn:ode-g-2}  \tag{ODE-$g$-2}
\end{align*}

Furthermore, in the $k_l \leq a/b \leq k_h$ regime, the system of equations defining $(v^*,\alpha)$ actually has a unique solution, and if we view $r^*$ and $\alpha$ as a function of $a/b$, then we have $r^* \uparrow b$ and $\alpha \downarrow 0$ as $a/b \downarrow k_l$, while $r^* \downarrow a$ and $\alpha \uparrow 1/n$ as $a/b  \uparrow k_h$.
\end{proposition}

\begin{proof}[Proof of Proposition~\ref{prop:opt-g-ode}]

We first consider the regime $v \leq v^*$ where the ODE is (\ref{eqn:ode-g-1}). By multiplying both sides of (\ref{eqn:ode-g-1}) by $(v-r^*)^{n-1}/v^{n-1}$, we observe that (\ref{eqn:ode-g-1}) is equivalent to
\begin{align*}
    \frac{d}{dv} \left[ \frac{(v-r^*)^{n-1}}{v^{n-1}} (g^*(v) - \alpha) \right] = \lambda \frac{(v-r^*)^{n-1}}{v^{n}} \, .
\end{align*}

For the case $a/b \leq k_l$, we set $\alpha = 0, g^*(r^*) = 0$ and integrate the above equation from $v = r^*$ to arbitrary $v$ to get the $g_u^* \equiv g^*$ as stated in the theorem statement.

For the case $k_l \leq a/b \leq k_h$, we set $g^*(a) = \alpha$ and integrate the above equation from $v = a$ to arbitrary $v$  to get the $g_u^* \equiv g^*$ as stated in the theorem statement. 

In both cases, we can use Lemma~\ref{lem:integral} to write $g^*(v)$ in the valid region as
\begin{align*}
    g^*(v) = \alpha + \lambda \sum_{k=n}^{\infty} \frac{1}{k} \left( \frac{v-r^*}{v} \right)^{k-(n-1)} \, ,
\end{align*}
which immediately implies that $g^*$ is an increasing continuous function in $v$. 

In the $a/b \leq k_l$ case, the valid region starts at $v = r^*$ and the expression immediately implies that $g^*(r^*) = 0 = \alpha$, and
\begin{align*}
    g^*(b) = \lambda \sum_{k=n}^{\infty} \frac{1}{k} \left( \frac{b-r^*}{b} \right)^{k-(n-1)} \, , 
\end{align*}
which equals 1 because $r^*/b$ satisfies the same defining equation as $k_l$, so we can set $r^* = k_l b$, and the expression is decreasing in $r^*$, so the equation setting the above to 1 has a unique solution in $r^*$ (equivalently, in $k_l$) if and only if as $r^* \downarrow a$, the expression is $\geq 1$, which is equivalent to $a/b \leq k_l$ that we had just assumed.

In the $k_l \leq a/b \leq k_h$ case, the valid region is from $v = a$ to $v = v^*$. The expression implies $g^*(a) = \alpha$. The defining equation for $(v^*,\alpha)$ in this regime implies that 
\begin{align*}
    \frac{(v^*-a)^{n-1}}{(v^*)^{n-1}} (g^*(v^*) - \alpha) = \lambda \int_{t=a}^{t=v^*} \frac{(t-a)^{n-1}}{t^{n}} dt = \frac{(v^*-a)^{n-1}}{(v^*)^{n-1}} ( 1 - n \alpha) \, ,
\end{align*}
so $g^*(v^*) = 1-(n-1) \alpha$. (We still need to prove that the two equations defining $(v^*,\alpha)$ has a unique solution; we will defer this to the end of the proof.)

Now we consider the regime $v \geq v^*$ where the ODE is (\ref{eqn:ode-g-2}). By multiplying both sides by $(v-a)^{n}/(v-\phi_0)^{n}$, we observe that (\ref{eqn:ode-g-2}) is equivalent to
\begin{align*}
    \frac{d}{dv} \left[ \frac{(v-a)^{n}}{(v-\phi_0)^{n}} g^*(v) \right] = \frac{(v-a)^{n-1}}{(v-\phi_0)^{n}} - (1-\lambda) \frac{(v-a)^{n}}{(v-\phi_0)^{n+1}} \, .
\end{align*}

In the case $a/b \geq k_h$, this equation applies for all $v \in [a,b]$, so we integrate this equation from $v = a$ to arbitrary $v$ and note that $\frac{(v-a)^{n}}{(v-\phi_0)^{n}} g^*(v)$ is 0 when $v = a$ (because of the $(v-a)$ factor), so we get
\begin{align*}
    \frac{(v-a)^{n}}{(v-\phi_0)^{n}} g^*(v) - 0 = \int_{t=a}^{t=v} \left[ \frac{ (t-a)^{n-1}}{(t-\phi_0)^{n}}  - (1-\lambda)  \frac{(t-a)^{n}}{(t-\phi_0)^{n+1}} \right] dt \, ,
\end{align*}
which is equivalent to the $g^*$ as stated in the theorem statement. By Lemma~\ref{lem:integral2}, we can write $g^*(v)$ as
\begin{align*}
    g^*(v) = \frac{1}{n} + \lambda \sum_{k=n+1}^{\infty} \frac{(v-a)^{k-n}}{k(v-\phi_0)^{k-n}} \, .
\end{align*}
The expression immediately implies that $g^*(v)$ is an increasing continuous function of $v$ and $g^*(a) = 1/n$. We also have
\begin{align*}
    g^*(b) = \frac{1}{n} + \lambda \sum_{k=n+1}^{\infty} \frac{(b-a)^{k-n}}{k(b-\phi_0)^{k-n}} = 1 \, ,
\end{align*}
by the defining equation of $\phi_0$, and by inspecting the defining equations for $\phi_0$ and $k_h$ we see that $\phi_0 = (a-k_h b)/(1-k_h)$ as claimed. The defining equation of $k_h$ is
\begin{align*}
    \frac{1}{n} + \lambda \sum_{k=n+1}^{\infty} \frac{(1-k_h)^{k-n}}{k}  = 1 \, .
\end{align*}
The expression is decreasing in $k_h$ and it is $1/n < 1$ as $k_h \uparrow 1$, and $\frac{1}{n} + \lambda \sum_{k=n+1}^{\infty} \frac{1}{k} = \infty$ as $k_h \downarrow 0$ because the harmonic series is divergent, so the equation has a unique solution $k_h$. 

In the case $k_l \leq a/b \leq k_h$, this equation applies for $v \in [v^*,b]$. By requiring that $g^*(b) = 1$, integrating the equation from arbitrary $v$ to $v = b$ gives
\begin{align*}
    \frac{(b-a)^{n}}{b^{n}} - \frac{(v-a)^{n}}{v^{n}} g^*(v) = \int_{t=v}^{t=b} \left[ \frac{(t-a)^{n-1}}{t^{n}} - (1-\lambda) \frac{(t-a)^{n}}{t^{n+1}} \right] dt \, .
\end{align*}
Just as before, Lemma~\ref{lem:integral2} implies that the right hand side is decreasing and continuous in $v$, so $g^*(v)$ is increasing and continuous in $v$. We also have
\begin{align*}
    \frac{(b-a)^{n}}{b^{n}} - \frac{(v-a)^{n}}{v^{n}} g^*(v^*) &= \int_{t=v^*}^{t=b} \left[ \frac{(t-a)^{n-1}}{t^{n}} - (1-\lambda) \frac{(t-a)^{n}}{t^{n+1}} \right] dt\\
    &= \frac{(b-a)^{n}}{b^{n}} - \frac{(v-a)^{n}}{v^{n}} (1-(n-1)\alpha) \, ,
\end{align*}
by the defining equations for $(v^*,\alpha)$, so $g^*(v^*) = 1-(n-1) \alpha$. We therefore see that the value of $g^*$ at $v^*$ from both the $v \leq v^*$ and the $v \geq v^*$ regions are equal, so $g^*$ is continuous at $v^*$ as well.

Finally, we will prove that the equations defining $(v^*,\alpha)$ in the $k_l \leq a/b \leq k_h$ regime have a unique solution.

Eliminating $\alpha$ from the two equations gives
\begin{align*}
n -  \frac{n (v^*)^{n}}{(v^*-a)^{n}} \left( \frac{(b-a)^n}{b^n} - \int_{t=v^*}^{t=b} \left[ \frac{(t-a)^{n-1}}{t^n} - (1-\lambda) \frac{(t-a)^{n}}{t^{n+1}}  \right] dt \right) \\
 = (n-1) - \frac{(n-1)(v^*)^{n-1}}{(v^*-a)^{n-1}} \lambda\int_{t=a}^{t=v^*} \frac{(t-a)^{n-1}}{t^n} dt \, ,
\end{align*}
or
\begin{align*}
\frac{n(b-a)^n}{b^n}  &= \frac{(r^*-a)^{n}}{(r^*)^{n}} + \frac{(n-1)(r^*-a)}{(r^*)} \lambda\int_{t=a}^{t=r^*} \frac{(t-a)^{n-1}}{t^n} dt\\
&\quad+ n \int_{t=r^*}^{t=b} \left[ \frac{(t-a)^{n-1}}{t^n} - (1-\lambda) \frac{(t-a)^{n}}{t^{n+1}}  \right] dt \, .
\end{align*}

Let $\text{fn}(r^*)$ denote the right hand side viewed as a function of $r^*$, namely,
\begin{align*}
\text{fn}(v) &:=  \frac{(v-a)^{n}}{v^{n}} + \frac{(n-1)(v-a)}{v} \lambda\int_{t=a}^{t=v} \frac{(t-a)^{n-1}}{t^n} dt + n \int_{t=v}^{t=b} \left[ \frac{(t-a)^{n-1}}{t^n} - (1-\lambda) \frac{(t-a)^{n}}{t^{n+1}}  \right] dt   \, .
\end{align*}

We claim that $\text{fn}(v)$ is a decreasing function. That is, we want to show that $d\text{fn}(v)/dv = \text{fn}'(v) \leq 0$. 
We compute
\begin{align*}
\text{fn}'(v) &= n \left( \frac{v-a}{v} \right)^{n-1}  \frac{a}{v^2} + (n-1) \lambda \frac{d}{dv} \left[ \frac{(v-a)}{v} \int_{t=a}^{t=v} \frac{(t-a)^{n-1}}{t^n} dt \right]\\
&\qquad- n \left( \frac{(v-a)^{n-1}}{v^n} - (1-\lambda) \frac{(v-a)^{n}}{v^{n+1}} \right) \, .
\end{align*}
Note that
\begin{align*}
n \left( \frac{v-a}{r} \right)^{n-1}  \frac{a}{v^2}  = \frac{na (v-a)^{n-1}}{v^{n+1}} = \frac{n(v-(v-a))(v-a)^{n-1}}{v^{n+1}} = \frac{n(v-a)^{n-1}}{v^{n}} - \frac{n(v-a)^{n}}{v^{n+1}} \, .
\end{align*}
We then have
\begin{align*}
\text{fn}'(v) &= - n \lambda \frac{ (v-a)^{n}}{v^{n+1}} + (n-1) \lambda \frac{d}{dv} \left[ \frac{(v-a)}{v} \int_{t=a}^{t=v} \frac{(t-a)^{n-1}}{t^n} dt \right] \\
&= - n \lambda \frac{ (v-a)^{n}}{v^{n+1}} + (n-1) \lambda  \left[\frac{(v-a)}{v} \frac{(v-a)^{n-1}}{v^n} + \frac{a}{v^2} \int_{t=a}^{t=v} \frac{(t-a)^{n-1}}{t^n} dt  \right] \\
&= \lambda \left[  - \frac{(v-a)^{n}}{v^{n+1}} + \frac{(n-1)a}{v^2} \int_{t=a}^{t=v} \frac{(t-a)^{n-1}}{t^n} dt \right] \, .
\end{align*}
Therefore, we have $\text{fn}'(v) \leq 0$ if and only if
\begin{align*}
\int_{a}^{v} \frac{(t-a)^{n-1}}{t^n} dt \leq \frac{1}{(n-1)a} \frac{(v-a)^{n}}{v^{n-1}} \, .
\end{align*}

We can prove this inequality as follows.  Both sides are zero for $v = a$, so it is sufficient to show that the derivative of the LHS is $\leq$ the derivative of the RHS. This is true because the derivative of the LHS is $(v-a)^{n-1}/v^n$ and the derivative of the RHS is
\begin{align*}
\frac{1}{(n-1)a} \frac{v^{n-1} n(v-a)^{n-1} - (v-a)^{n} (n-1)v^{n-2}}{v^{2n-2}} = \frac{1}{(n-1)a} \frac{(v-a)^{n-1}}{v^{n}} (nv-(n-1)(v-a)) \\
= \frac{(v-a)^{n-1}}{v^n} + \frac{(v-a)^{n-1}}{(n-1)av^{n-1}} \geq  \frac{(v-a)^{n-1}}{v^n} \, .
\end{align*}

Therefore, we have proved that $\text{fn}(v)$ is decreasing in $v$.

To show that the equation $\text{fn}(v^*) = n \left( \frac{b-a}{b} \right)^{n}$ has a unique solution $v^* \in [a,b]$, it is sufficient to show that $\text{fn}(a) \geq n \left( \frac{b-a}{b} \right)^{n} \geq \text{fn}(b)$, or
\begin{align*}
n \int_{a}^{b} \left[ \frac{(t-a)^{n-1}}{t^n} - (1-\lambda) \frac{(t-a)^{n}}{t^{n+1}} \right] dt &\geq  n \left( \frac{b-a}{b} \right)^{n}\\
&\geq \left( \frac{b-a}{b} \right)^{n} + \frac{(n-1)(b-a)}{b} \lambda \int_{a}^{b} \frac{(t-a)^{n-1}}{t^n} dt \, .
\end{align*}

The first inequality
\begin{align*}
\int_{a}^{b} \left[ \frac{(t-a)^{n-1}}{t^n} - (1-\lambda) \frac{(t-a)^{n}}{t^{n+1}} \right] dt \geq  \left( \frac{b-a}{b} \right)^{n}
\end{align*}
is true by the definition of $k_h$ and $a/b \leq  k_h$. The second inequality is equivalent to
\begin{align*}
\lambda \int_{t=a}^{t=b} \frac{(t-a)^{n-1}}{t^n} dt \leq \frac{(b-a)^{n-1}}{b^{n-1}} \, ,
\end{align*}
which is true by the definition of $k_l$ and $a/b \geq k_l$. 

We also conclude from the above that as $a/b \uparrow k_h$ we have $r^* \downarrow a$, while as $a/b \downarrow k_l$, we have $r^* \uparrow b$.

Now we will show that $\alpha \in [0,1/n]$ and  as $a/b \uparrow k_h$ we have $\alpha \uparrow 1/n$, while as $a/b \downarrow k_l$, we have $\alpha \downarrow 0$.

We will first show that for any $v \in [a,b]$, we have $\int_{a}^{v} \frac{(t-a)^{n-1}}{t^n} dt \leq \frac{1}{\lambda} \frac{(v-a)^{n-1}}{v^{n-1}}$. Let $\text{fn}(v) := \int_{a}^{v} \frac{(t-a)^{n-1}}{t^n} dt - \frac{1}{\lambda} \frac{(v-a)^{n-1}}{v^{n-1}}$. (We overload the $\text{fn}$ notation here --- it has nothing to do with the earlier $\text{fn}$; it is just a shorthand that we discard after we finish proving the technical statement.) Note that $\text{fn}(a) = 0$ and $\text{fn}(b) \leq 0$ by the definition of $a/b \geq k_l$. We have
\begin{align*}
\text{fn}'(v)  = \frac{(v-a)^{n-1}}{v^{n}} - \frac{1}{\lambda} (n-1) \frac{(v-a)^{n-2}}{v^{n-2}} \frac{a}{v^2} = \frac{(v-a)^{n-2}}{v^n} \left(v- \left( 1 + \frac{n-1}{\lambda} \right) a \right) \, .
\end{align*}
So $\text{fn}(v)$ is decreasing for $v \leq \left( 1 + \frac{n-1}{\lambda} \right) a$ and increasing for $v \geq \left( 1 + \frac{n-1}{\lambda} \right) a$. Regardless of whether $b \leq \left( 1 + \frac{n-1}{\lambda} \right) a$ or not, we have $\text{fn}(v) \leq \max(\text{fn}(a), \text{fn}(b)) \leq 0$, and we are done. Note also that as $a/b \downarrow k_l$, we have $r^* \uparrow b$ so every inequality here approaches equality, and $\alpha \downarrow 0$.

From the defining equation we have
\begin{align*}
1 - n \alpha = \lambda \frac{(r^*)^{n-1}}{(r^*-a)^{n-1}} \int_{a}^{r^*} \frac{(t-a)^{n-1}}{t^n} dt \, ,
\end{align*}
but by the lemma we have just proved, 
\begin{align*}
\int_{a}^{r^*} \frac{(t-a)^{n-1}}{t^n} dt \leq \frac{1}{\lambda} \frac{(r^*-a)^{n-1}}{(r^*)^{n-1}} \, ,
\end{align*}
so $1-n\alpha \leq 1$, which implies $\alpha \geq 0$. 

From the same equation, it is clear that $\frac{(r^*)^{n-1}}{(r^*-a)^{n-1}} \int_{a}^{r^*} \frac{(t-a)^{n-1}}{t^n} dt \geq 0$, so $1-n\alpha \geq 0$, so $\alpha \leq 1/n$. Furthermore, as $a/b \uparrow k_h$, we have $r^* \downarrow a$ so
\begin{align*}
1-n\alpha = \lambda \frac{(r^*)^{n-1}}{(r^*-a)^{n-1}} \int_{a}^{r^*} \frac{(t-a)^{n-1}}{t^n} dt \leq  \lambda \frac{(r^*)^{n-1}}{(r^*-a)^{n-1}} \int_{a}^{r^*} \frac{(t-a)^{n-1}}{a^n} dt = \lambda \frac{(r^*)^{n-1}(r^*-a)}{na^n} \downarrow 0 \, ,
\end{align*} 
so $\alpha \uparrow 1/n$.

We conclude that for $k_l \leq a/b \leq k_h$, there is a unique valid solution $(r^*,\alpha)$. Furthermore, as $a/b \uparrow k_h$ we have $r^* \downarrow a$ and $\alpha \uparrow 1/n$, while as $a/b \downarrow k_l$, we have $r^* \uparrow b$ and $\alpha \downarrow 0$, as desired.   
\end{proof}

Proposition~\ref{prop:opt-g-ode} gives a unifying description of the mechanism across three support information regimes and makes it clear that the intermediate regime interpolates between the SPA regime and the POOL regime. 

We will now use (\ref{eqn:ode-g-1}) and (\ref{eqn:ode-g-2}) given in Proposition~\ref{prop:opt-g-ode} to check that (\ref{eqn:foc}) and (\ref{eqn:soc}) hold.

\paragraph{Checking (\ref{eqn:foc})} We first consider the case $v \leq v^*$, where (\ref{eqn:ode-g-1}) applies. We only need to check this case in the low information ($a/b \leq k_l$) and moderate information ($k_l \leq a/b \leq k_h$) regimes, because this case becomes vacuous ($v^* = a$) in the high information ($a/b \geq k_h$) regime. In both of these regimes, $F^*(v) = 1-r^*/v$. Therefore, the (\ref{eqn:foc}) equation is
\begin{align*}
    \left[ - \lambda - (n-1) (g^*(v) -\alpha) + v (g^*)'(v) \right]\left( 1 - \frac{r^*}{v} \right) + (n-1) (g^*(v) - \alpha) \\
    = \left( 1 - \frac{r^*}{v} \right) \left[ - \lambda + \frac{(n-1)r^*}{v(v-r^*)} (g^*(v) - \alpha) + v (g^*)'(v) \right]  = 0 \, ,
\end{align*}
where the last equality holds by (\ref{eqn:ode-g-1}).

Now we consider the case $v \geq v^*$ where (\ref{eqn:ode-g-1}) applies. We only need to check this case in the high information ($a/b \geq k_h$) and moderate information ($k_l \leq a/b \leq k_h$) regimes, because this case becomes vacuous ($v^* = b$) in the low information ($a/b \geq k_h$) regime. In both regimes, $F^*(v) = 1 - (a-\phi_0)/(v-\phi_0) = (v-a)/(v-\phi_0)$. ($\phi_0 = 0$ for the moderate information regime, and $\phi_0 > 0$ for the high information regime.) Therefore, the (\ref{eqn:foc}) is
\begin{align*}
    &\left[ - \lambda - (n-1) \left( g^*(v) - \frac{1-g^*(v)}{n-1} \right)  + v \cdot 0 \right] \left( \frac{v-a}{v-\phi_0} \right) \\
    &\qquad+ (n-1) \left[ g^*(v) - \frac{1-g^*(v)}{n-1} - (v-a) \left( \frac{-(g^*)'(v)}{n-1} \right) \right] \\
    &= \left( - \lambda - (n-1) g^*(v) - 1 + g^*(v) \right) \left( \frac{v-a}{v-\phi_0} \right) + (n-1) g^*(v) -1+g^*(v) + (v-a) (g^*)'(v) \\
    &= (v-a) (g^*)'(v) - \left( \frac{v-a}{v-\phi_0} \right) \lambda + (n g^*(v)-1 ) \left( \frac{a-\phi_0}{v-\phi_0} \right) \\
    &= (v-a) \left( \frac{1}{v-a} - \frac{1-\lambda}{v-\phi_0} - \frac{n(a-\phi_0)}{(v-\phi_0)(v-a)} g^*(v) \right) - \left( \frac{v-a}{v-\phi_0} \right) \lambda + (n g^*(v)-1 ) \left( \frac{a-\phi_0}{v-\phi_0} \right) = 0 \, ,
\end{align*}
where the second-to-last equality holds by (\ref{eqn:ode-g-1}).

\paragraph{Checking (\ref{eqn:soc}).}

The following applies whether we are in the regime $a/b \leq k_l$, $a/b \geq k_h$, or $k_l \leq a/b \leq k_h$.

In the case $v \leq v^*$, we have $g_u^*(v) = g^*(v)$ and $g_d^*(v) = \alpha$, so (\ref{eqn:soc}) reduces to $g^*(v) - \alpha > 0$, which is true by definition of $g^*$ in the interior.

In the case $v \geq v^*$, we have $g_u^*(v) = g^*(v)$ and $g_d^*(v) = (1-g^*(v))/(n-1)$, so (\ref{eqn:soc}) reduces to $n g^*(v) - 1 + (v-a) (g^*)'(v) > 0$. 

Here, $g^*$ satisfies (\ref{eqn:ode-g-2}) with $\phi_0 = 0$ in the case $k_l \leq a/b \leq k_h$ and $\phi_0 \geq 0$ in the case $a/b \geq k_h$. Substituting $(g^*)'(v) = \frac{1}{v-a} - \frac{1-\lambda}{v-\phi_0} - \frac{n(a-\phi_0)}{(v-\phi_0)(v-a)} g^*(v)$, we get that (\ref{eqn:soc}) reduces to  $n g^*(v) - 1 + (v-a) (g^*)'(v) = \frac{(v-a)}{(v-\phi_0)} (n g^*(v) - 1+\lambda) > 0$, so we have to prove that $n g^*(v) -1 + \lambda > 0$. From $v \geq v^*$ we have $g^*(v) \geq 1-(n-1)\alpha$, we have $n g^*(v) -1+\lambda \geq n(1-(n-1)\alpha) -1+\lambda = (n-1)(1-\alpha) + \lambda \geq 0$, because $\alpha \leq 1, \lambda \geq 0$, with strict inequality everywhere but the boundary, as desired.

\subsection{Proofs for Section~\ref{subsec:structure-mech-all}}\label{app:subsec:structure-mech-all}

\begin{proof}[Proof of Proposition~\ref{prop:maximin-ratio-regime}]

Suppose for the sake of contradiction that the $\spa$ regime is possible. By Theorem~\ref{thm:char-all-mech-full-main}, the $\lambda$-regret
\begin{align*}
    -1 + (1-k_l)^n + \lambda \left(1 - \int_{t=k_l}^{t=1} \left( 1 - \frac{k_l}{t} \right)^n dt \right)
\end{align*}
is zero, while the corresponding $\lambda$ satisfies
\begin{align*}
    \lambda = \frac{(1-k_l)^{n-1}}{\int_{t=k_l}^{t=1} \frac{(t-k)^{n-1}}{t^{n}} dt }\,.
\end{align*}
Substituting the expression of $\lambda$ gives
\begin{align}
    -1 + (1-k_l)^n + (1-k_l)^{n-1} \frac{1 - \int_{t=k_l}^{t=1} \left( 1 - \frac{k_l}{t} \right)^{n} dt }{ \int_{t=k_l}^{t=1} \frac{(t-k_l)^{n-1}}{t^{n}} dt } = 0\,. \label{eqn:lmbd-reg-kl}
\end{align}
Let $f(k_l)$ be the left hand side of (\ref{eqn:lmbd-reg-kl}).
We will derive a contradiction by showing that $f(k_l) > 0$ for all $0 < k_l \leq 1$. We will prove this by viewing $k_l \in (0,1]$ as a free variable.
Let 
\begin{align*}
    I(k_l) = \int_{t=k_l}^{t=1} \frac{(t-k_l)^{n-1}}{t^{n}} dt = \sum_{i=n}^{\infty} \frac{1}{i} (1-k_l)^{i}\,.
\end{align*}
Note that
\begin{align*}
    I'(k_l) = \sum_{i=n}^{\infty} (1-k_l)^{i-1} (-1) = - \frac{(1-k_l)^{n-1}}{1-(1-k_l)} = - \frac{(1-k_l)^{n-1}}{k_l} \, .
\end{align*}
Also, by integration by parts,
\begin{align*}
    \int_{t=k_l}^{t=1} \left( 1 - \frac{k_l}{t} \right)^{n} dt = \left[ \left( 1 - \frac{k_l}{t} \right)^{n} t \right]_{t=k_l}^{t=1} - \int_{t=k_l}^{t=1} t n \left( 1 - \frac{k_l}{t} \right)^{n-1} \frac{k}{t^2} dt = (1-k_l)^{n} - n k I(k_l)\,.
\end{align*}
Therefore, both integrals in $f(k_l)$ can be written in terms of $I(k_l)$. We want to show that
\begin{align*}
    -1 + (1-k_l)^n + (1-k_l)^{n-1} \frac{1 - (1-k_l)^{n} + n k I(k_l) }{ I(k_l) } > 0\,.
\end{align*}
This is equivalent to
\begin{align*}
    I(k_l) < \frac{1-(1-k_l)^{n}}{(1-k_l)^{-(n-1)} - (1+(n-1)k_l)}\,.
\end{align*}
Note that $(1-k_l)^{-(n-1)} > (1+(n-1)k_l)$, so the above manipulation is valid, and both sides of the inequality are positive. Let
\begin{align*}
    g(k_l) = I(k_l) - \frac{1-(1-k_l)^{n}}{(1-k_l)^{-(n-1)} - (1+(n-1)k_l)}\,.
\end{align*}
It is clear from the integral definition of $I(k_l)$ that $\lim_{k_l \uparrow 1} I(k_l) = 1$. We now compute, by L'Hopital's rule,
\begin{align*}
    \lim_{k_l \uparrow 1} \frac{1-(1-k_l)^{n}}{(1-k_l)^{-(n-1)} - (1+(n-1)k_l)} = \lim_{k_l \uparrow 1} \frac{-n(1-k_l)^{n-1}(-1)}{(-n+1)(1-k_l)^{-n}(-1)-(n-1)} = 0\,.
\end{align*}
Therefore, $\lim_{k_l \uparrow 1} g(k_l) = 0$ To prove that $g(k_l) < 0$ it is sufficient to prove that $g(k_l)$ is strictly increasing in $k_l$, i.e., $g'(k_l) > 0$. This is very convenient because $I'(k_l)$ does not involve an integral. 
We compute
\begin{align*}
    g'(k_l) &= I'(k_l) - \frac{d}{dk_l} \left[ \frac{1-(1-k_l)^{n}}{(1-k_l)^{-(n-1)} - (1+(n-1)k_l)} \right] \\
    &= - \frac{(1-k_l)^{n-1}}{k_l} - \frac{ 1 }{ \left[ (1-k_l)^{-n+1} - (1+(n-1)k_l) \right]^2 } \times \Bigg\{ \\ &\left[ (1-k_l)^{-n+1} - (1+(n-1)k_l) \right] \left[ -n(1-k_l)^{n-1}(-1) \right] \\ &- \left[ 1 - (1-k_l)^n \right] \left[ (-n+1)(1-k_l)^{-n} (-1) - (n-1) \right] \Bigg\} \\
    &= - \frac{(1-k_l)^{n-1}}{k_l} + \frac{ \frac{(n-1)( 1-(1-k_l)^{n} )^2}{ (1-k_l)^{n} } - n \left[ 1 - (1+(n-1)k_l)(1-k_l)^{n-1} \right]  }{ \frac{1}{(1-k_l)^{2n-2}} \left[ 1 - (1+(n-1)k_l)(1-k_l)^{n-1} \right]^2 }\,.
\end{align*}
Then $g'(k_l) > 0$ is equivalent to
\begin{align*}
    (n-1) k_l (1-(1-k_l)^n)^2 - n (1-k_l)^n k_l \left[ 1 - (1+(n-1)k_l)(1-k_l)^{n-1} \right] \\ - (1-k_l) \left[ 1 - (1+(n-1)k_l)(1-k_l)^{n-1} \right]^2 > 0\,.
\end{align*}
Let $x = 1-k_l \in [0,1)$, algebraic simplification gives that the left hand side is
\begin{align*}
    &(n-1)(1-x)(1-x^n)^2 - nx^n (1-x) \left[ 1-(1+(n-1)(1-x))x^{n-1} \right]\\
    &\quad- \left[ 1-(1+(n-1)(1-x))x^{n-1} \right]^2 \\
    &= (1-x^n)(n-1-nx+x^n)\,.
\end{align*}
It is clear that $1-x^n > 0$. We also have
\begin{align*}
    n-1-nx+x^n = n(1-x) - (1-x^n) = (1-x) \left(n - \sum_{i=0}^{n-1} x^{i} \right) > 0\,.
\end{align*}

We conclude that the optimal maximin ratio mechanism is never in the $\spa$ regime.\end{proof}

\subsection{Proofs for Section~\ref{subsec:n-1-remark} }\label{app:subsec:n-1-remark}

\begin{proof}[Proof of Corollary~\ref{cor:lambda-n-1}]
We will use the equivalent formulation of Theorem~\ref{app:thm:char-all-mech-gu-gd}. $k_l$ is a solution to $\lambda \log(1/k_l) = 1$, so $k_l = \exp(-1/\lambda)$. $k_h$ is a solution to $(1-\lambda)(1-k_h) + \lambda \log(1/k_h) = (1-k_h)$ or $\log(1/k_h) = (1-k_h)$, so $k_h = 1$.

Therefore, we only need to consider the regime $a/b \leq k_l$ and $k_l \leq a/b \leq k_h$. 

For $a/b \leq k_l$, the regret is $b/e$. For $a/b \geq k_l$, we know that $(v^*,\alpha)$ is a solution to the following system of equations
\begin{align*}
    1 - \alpha &= \lambda \log\left( \frac{v^*}{a} \right) \\
    \frac{b-a}{b} - \frac{v^* - a}{v^*} &= \lambda \log\left( \frac{b}{v^*} \right) - \frac{(1-\lambda)a}{b} + \frac{(1-\lambda)a}{v^*} \, .
\end{align*}
The solution to this is $v^* = b$ and $\alpha = 1- \lambda \log(b/a)$.

We now substitute $n = 1$ to the minimax $\lambda$-regret and $g_u^*(v)$ expressions. (We have $\Phi(v) = g_u^*(v)$, the price distribution CDF.) For $a/b \leq k_l$, the minimax $\lambda$-regret is
\begin{align*}
    -(1-\lambda) b + \left[ \left( 1 - k_l \right) - \lambda \int_{\tilde{r}=k_l}^{\tilde{r}=1} \left( 1 - \frac{k_l}{\tilde{r}} \right) d\tilde{r} \right] b \\
    = - (1-\lambda) b + \left[ (1-k_l) - \lambda (1-k_l + k_l \log(k_l)) \right] b \\
    = (-1+\lambda + 1 - e^{-1/\lambda} - \lambda + \lambda e^{-1/\lambda} + e^{-1/\lambda}) b \\
    = \lambda e^{-1/\lambda} b \, ,
\end{align*}
and the price CDF is
\begin{align*}
    g_u^*(v) = \lambda \frac{v^{n-1}}{(v-r^*)^{n-1}} \int_{t=r^*}^{t=v} \frac{(t-r^*)^{n-1}}{t^{n}} dt = \lambda \int_{t=r^*}^{t=v} \frac{1}{t} dt =  \lambda \log \left( \frac{v}{r^*} \right) = \lambda \log \left( \frac{v}{e^{-1/\lambda} b} \right) \\ = 1 + \lambda \log \left( \frac{v}{b} \right) \text{ for } v \in [r^*,b]  \, .
\end{align*}
For $a/b \geq k_l$, the minimax $\lambda$-regret is
\begin{align*}
    -(1-\lambda) b + b \left( 1 - \frac{a}{b} \right) + (b-a) \alpha (1-1) - \lambda \int_{v=a}^{v=b} \left( 1 - \frac{a}{v} \right) dv \\
    = - (1-\lambda) b + (b-a) - \lambda (b-a-a \log(b/a)) = -a + \lambda a + \lambda \log(b/a) \, ,
\end{align*}
and the price CDF is
\begin{align*}
    g_u^*(v) &= \alpha + \lambda \left( \frac{v}{v-a} \right)^{n-1} \int_{t=a}^{t=v} \frac{(t-a)^{n-1}}{t^{n}} dt = 1 - \lambda \log \left( \frac{b}{a} \right) + \lambda \int_{t=a}^{t=v} \frac{1}{t} dt\\
    &= 1 - \lambda \log\left( \frac{b}{a} \right) + \lambda \log \left( \frac{v}{a} \right) 
    = 1 + \lambda \log \left( \frac{v}{b} \right) \text{ for } v \in [a,b] \, .
\end{align*}

We get the minimax regret by substituting $\lambda = 1$, which is $b/e$ for $a/b \leq 1/e$ and $a \log (b/a)$ for $a/b \geq 1/e$. The maximin ratio is $\lambda$ such that the minimax $\lambda$-regret is zero. Because $\lambda e^{-1/\lambda} b > 0$ always, minimax $\lambda$-regret cannot be zero in this regime. In the other regime, $k_l \leq a/b \leq k_h$, The $\lambda$ such that $-a + \lambda a + \lambda \log(b/a) = 0$ is $\lambda^* = 1/(1+\log(b/a))$, so this is the maximin ratio value, and the corresponding price CDF is
\begin{align*}
    \Phi^*(v) = 1 + \lambda^* \log \left( \frac{v}{b} \right) = 1 + \frac{\log(v/b)}{1+\log(b/a)} \text{ for } v \in [a,b] \,  . \quad \qedhere
\end{align*}
\end{proof}
}

\jadelete{

\section{Full Proof of Theorem \ref{thm:char-all-mech-full-main}}\label{app:sec:proof-main-thm-outline}

\subsection{$(g_u,g_d$) Mechanisms As A Unified Direct Mechanism Representation}\label{app:subsec:spa-pool-gu-gd}

Our main theorem shows that the optimal mechanism is a convex combination over the base mechanisms $\{\textnormal{SPA}(r),\textnormal{POOL}(\tau)\}$. To prove the necessary saddle inequalities, we need to be able to compute the expected regret of any such mechanism against some distribution. It is more mathematically convenient to ``flatten'' the randomness over thresholds and give an almost equivalent ``direct'' representation that gives the allocation rule $x(\mathbf{v})$ and payment rule $p(\mathbf{v})$ for any valuation vector $\mathbf{v}$ as follows.

\begin{definition}[\emph{$(g_u,g_d)$} mechanisms]\label{def:gu-gd} 
Let $g_u,g_d: [a,b] \to [0,1]$ be given functions.  A mechanism  \emph{$(g_u,g_d)$} is defined by the allocation rule $x: [a,b]^n \to [0,1]^n$ given by, for each $i \in [n]$,
\begin{align*}
x_i(\mathbf{v}) = \begin{cases}
\frac{1}{k} g_u(v_{\max}) + \frac{k-1}{k} g_d(v_{\max}) &\text{ if } v_i = \max(\mathbf{v}) := v_{\max} \text{  and there are $k$ entries in $\mathbf{v}$ equal to $v_{\max}$} \\
g_d(v_{\max}) &\text{ if } v_i < \max(\mathbf{v}) := v_{\max}
\end{cases}
\end{align*} 
and the payment rule $p: [a,b]^n \to \mathbb{R}_{++}^n$ is determined uniquely from  Myerson's formula such that the resulting mechanism $(x,p)$ is dominant strategy incentive compatible. 
\end{definition}

In other words, the mechanism allocates $g_u(v_{\max})$ to the highest bidder(s) and $g_d(v_{\max})$ to the non-highest bidder(s). If there are $k$ highest bidders, select one of them to be the ``winner'' with $g_u$ uniformly at random. 

The $(g_u,g_d)$ mechanism representation has the advantage that all mechanisms in the three different support information regimes (SPA, POOL, and a convex combination between SPA and POOL) can all be represented in this form, so we can have a unified treatment for all of them.

We prove Proposition~\ref{prop:gu-gd-rep} in  Appendix~\ref{app:subsec:spa-pool-gu-gd-detail}.

We choose to highlight $\{\textnormal{SPA}(r), \textnormal{POOL}(\tau)\}$ in the main text because the base mechanisms SPA and POOL are intuitive and interpretable. However, the $(g_u,g_d)$ representation is more mathematically convenient because it directly gives the allocation probabilities and payments for each valuation vector $\mathbf{v}$ which are ingredients for the expected regret calculation, and also because it allows the three cases to be treated in a unified manner.

The rest of this Appendix will be devoted to proving the above reformulated main theorem (Theorem~\ref{app:thm:char-all-mech-gu-gd}). We need to verify Nature's saddle (fixing the mechanism $(g_u^*,g_d^*)$, the best distribution $F$ is $F^*$) and Seller's saddle (fixing the distribution $F^*$, the best mechanism is the $(g_u^*,g_d^*)$ mechanism). (The above Theorem~\ref{app:thm:char-all-mech-gu-gd} only states the optimal mechanism and worst-case distribution that together form a saddle pair. We will separately verify the expressions of the minimax $\lambda$-regret that we also state in Theorem~\ref{thm:char-all-mech-full-main} later in Proposition~\ref{prop:minimax-lmbd-regret-expressions}.)

The Seller's Saddle is a Bayesian mechanism design problem, and we can check optimality through the conditions of \cite{monteiro-svaiter-optimal-auction-general-dist}. (In particular, we see that the mechanism and the distribution have the same support, and the distribution has constant nonnegative virtual value on the support.) Therefore, for the rest of this proof we will focus on Nature's saddle.

\subsection{Conditions for Nature's Saddle}\label{app:subsec-conditions-nature-saddle}

To prove the necessary saddle inequalities, we need an expression for the expected regret of a $(g_u,g_d)$ mechanism under a given distribution $F$, given in the following proposition.

\begin{proposition}[Expected Regret of a $(g_u,g_d)$ mechanism]\label{prop:reg-exp-g-full}
Let $R_{\lambda}(\mathbf{g},\mathbf{F}) := R_{\lambda}((g_u,g_d), \mathbf{F})$ be the expected regret of a $(g_u,g_d)$ mechanism and valuation distribution $\mathbf{F}$.

If we assume that $g_u$ and $g_d$ are continuous everywhere and differentiable everywhere except a finite number of points, and $\mathbf{F}$ is arbitrary, then the expected regret is
\begin{align*}
R_{\lambda}(\mathbf{g},\mathbf{F}) = a(\lambda-g_u(a)-(n-1)g_d(a)) + \int_{v=a}^{v=b} (\lambda-g_u(v)+ g_d(v) -vg_u'(v)+ (v-na) g_d'(v)) dv  \\
+  \int_{v=a}^{v=b} (-\lambda+vg_u'(v)+(n-1)a g_d'(v))  \mathbf{F}_n^{(1)}(v)dv + (g_u(v)-g_d(v)-(v-a)g_d'(v))  \mathbf{F}_n^{(2)}(v)dv  \, , \tag{Regret-$\mathbf{F}$} \label{eqn:regret-Fbold}
\end{align*}
where $\mathbf{F}_n^{(1)}$ and $\mathbf{F}_n^{(2)}$ are the first and second order statistics (highest and second-highest) among $n$ agents whose valuations are drawn from $\mathbf{F}$. 

If we further assume that $\mathbf{F}$ is i.i.d. with marginal $F$, then
\begin{align*}
R_{\lambda}(\mathbf{g},\mathbf{F}) &= a(\lambda-g_u(a)-(n-1)g_d(a)) + \int_{v=a}^{v=b} (\lambda-g_u(v)+ g_d(v) -vg_u'(v)+ (v-na) g_d'(v))  \\
&+  \int_{v=a}^{v=b} \left( -\lambda -(n-1) (g_u(v)-g_d(v)) + v(g_u'(v)+(n-1) g_d'(v)) \right) F(v)^n dv \\
&+ \int_{v=a}^{v=b} n(g_u(v)-g_d(v)-(v-a)g_d'(v))   F(v)^{n-1}  dv \, . \tag{Regret-$F$} \label{eqn:regret-F}
\end{align*}

If, instead, we let $g_u$ and $g_d$ be arbitrary but we assume that $\mathbf{F}$ is i.i.d. with marginal $F$ such that $F$ has a density $f = F'$ on $[a,b)$ (so it potentially has point masses only at $a$ and $b$ of size $F(a)$ and $F(\{b\}) := f_b$ respectively), then we have the (Regret-$\mathbf{g}$) expression

\begin{align*}
R_{\lambda}(\mathbf{g},\mathbf{F}) &= \lambda b - a \left( g_u(a) + (n-1) g_d(a) \right) F(a)^n - \left( b g_u(b) + (n-1)a g_d(b) \right) \left( 1 -(1-f_b)^n \right) \\
&+ (b-a) g_d(b) \left( 1-(1-f_b)^{n-1} (1+(n-1)f_b) \right) \\
&+ \int_{v=a}^{v=b} -\lambda F(v)^n + g_u(v) n F(v)^{n-1} (1-F(v) - vF'(v) ) dv \\
&+ \int_{v=a}^{v=b} g_d(v) n(n-1) F(v)^{n-2} F'(v) \left\{ (v-a)(1-F(v)) - a F(v) \right\} dv \, . \tag{Regret-$\mathbf{g}$}  \label{eqn:regret-gbold}
\end{align*}

The above expression is valid for $n \geq 1$ if we take the expression $n(n-1)F(v)^{n-2}$ to be zero for $n = 1$.

\end{proposition}

We prove Proposition~\ref{prop:reg-exp-g} in Appendix~\ref{app:subsec:spa-pool-gu-gd-detail}. The precise expression for the regret is not qualitatively as important as the crucial fact that the regret can be written as an integral over a polynomial function of $F(v)$ only, as seen in (\ref{eqn:regret-F}). Therefore, in Nature's saddle, if we fix $(g_u^*,g_d^*)$ and optimize over $F$, we can optimize over each $F(v)$ pointwise, subject only to monotonicity constraints. If we can find a global minimum of the polynomial that also satisfies the monotonicity constraints, then we will automatically achieve pointwise optimality. As argued in the main text, the expression $\alpha(v) F(v)^{n-1} - \beta(v) F(v)^{n}$ as a function of $F(v)$ has a global maximum at $F^*(v)$ if the following two conditions hold: (i) first-order condition $(n-1) \alpha(v) - n \beta(v) F^*(v) = 0$, (ii) $\alpha(v) \geq 0$. We therefore have the following proposition. 

\begin{proposition}\label{prop:nature-saddle-foc-soc}
Suppose that $m^*$ is a $(g_u^*,g_d^*)$ mechanism and $F^*(v)$ is an increasing function that satisfies the following conditions:
\begin{align*}
 \left( -\lambda -(n-1) (g_u^*(v)-g_d^*(v)) + v(g_u'(v)+(n-1) g_d'^*(v)) \right)  F^*(v)  \\ + (n-1) (g_u^*(v)-g_d^*(v)-(v-a)g_d'^*(v))   = 0 \tag{FOC} \label{eqn:foc} \\
g_u^*(v)-g_d^*(v)-(v-a)g_d'^*(v)  > 0 \, . \tag{SOC} \label{eqn:soc}
\end{align*}
Then $R(m^*,F^*) \leq R(m^*,F)$ for any $F$.
\end{proposition}

We label the condition corresponding to $\alpha(v) \geq 0$ as (\ref{eqn:soc}) because it is equivalent to the second-order condition on the $I(F) \equiv \alpha F^{n-1} - \beta F^{n}$. To see this, we compute $\partial^2 I(F)/\partial F^2 |_{F = F^*} = (n-1) (F^*)^{n-3} ( (n-2) \alpha + n \beta F^* ) = (n-1) F^{n-3} (-\alpha)$, where the last equality holds because of FOC $(n-1) \alpha + n \beta F^* = 0$. Therefore, SOC $\partial^2 I(F)/\partial F^2 |_{F = F^*}  \leq 0$ is equivalent to $\alpha(v) \geq 0$. Nevertheless, we want to emphasize that we do \textit{not} say that first-order and second-order conditions together imply global optimality in general. This is not true; it is only true in this case due to the special structure of the integrand, which we analyze directly.

\subsection{Verification of First and Second Order Conditions for Nature's Saddle}\label{subsec:verify-foc-soc-nature-saddle}

}

\jadelete{
\section{Technical Results Supporting the Proof of the Main Theorem (Theorem~\ref{thm:char-all-mech-full-main} and~\ref{app:thm:char-all-mech-gu-gd}) From Appendix~\ref{app:sec:proof-main-thm-outline}}\label{app:sec:proof-main-thm-detail}

\subsection{Technical lemmas}\label{app:subsec:technical-lemmas-proof-main-thm}

Throughout this section, we will use the following technical lemmas.

\subsection{$(g_u,g_d)$ Mechanisms As A Unified Direct Mechanism Representation (Appendix~\ref{app:subsec:spa-pool-gu-gd})}\label{app:subsec:spa-pool-gu-gd-detail}

\subsection{Conditions for Nature's Saddle (Appendix~\ref{app:subsec-conditions-nature-saddle})}\label{app:subsec-conditions-nature-saddle-detail}

\subsection{Verification of Conditions for Nature's Saddle (Appendix~\ref{subsec:verify-foc-soc-nature-saddle})}\label{subsec:verify-foc-soc-nature-saddle-detail}

\subsection{Minimax $\lambda$-Regret Expressions}\label{app:subsec:proof-main-thm-additional}

The expected regret expression of a $(g_u,g_d)$ mechanism, as shown in Proposition~\ref{prop:reg-exp-g} naturally depends on both $g_u$ and $g_d$. However, we know that in the actual saddle point, there is a definite relationship between $g_u$ and $g_d$, so we can let $g_u \equiv g$ and write both $g_u$ and $g_d$ in terms of $g$. In Lemma~\ref{lem:regret-gbold-to-regret-g} below, we will derive such an expression (\ref{eqn:regret-g}). We then use (\ref{eqn:regret-g}) to derive the expressions of the minimax $\lambda$-regret stated in the main theorem (Theorem~\ref{thm:char-all-mech-full-main}) in Proposition~\ref{prop:minimax-lmbd-regret-expressions}.

\begin{lemma}\label{lem:regret-gbold-to-regret-g}
    If we set $g_u(v) = g(v)$, $g_d(v) = \alpha$, for $v \in [a,v^*]$ and $g_u(v) = g(v), g_d(v) = (1-g(v))/(n-1)$ for $v \in [v^*,b]$, $g(a) = \alpha, g(b) = 1$ in (\ref{eqn:regret-gbold}), then we get the following expression
    \begin{align*}
&\lambda b - a n \alpha F(a)^n - b + b(1-f_b)^n \\
&+\int_{v=a}^{v=v^*} \Big[ -  \lambda F(v)^n  + \alpha n(n-1) F(v)^{n-2} F'(v) \left\{ (v-a)(1-F(v))-aF(v) \right\} \\ &\quad\quad\quad\quad\quad + n F(v)^{n-1} (1-F(v)-vF'(v)) g(v)  \Big] dv \\
&+ \int_{v=v^*}^{v=b} \Big[ - \lambda F(v)^n + n F(v)^{n-2} F'(v) \left\{ (v-a)(1-F(v)) - a F(v) \right\}  \\ &\quad\quad\quad\quad\quad + n F(v)^{n-2} \left\{ F(v) - F(v)^2 - (v-a) F'(v) \right\} g(v) \Big] dv \, .
\tag{Regret-$g$}  \label{eqn:regret-g}
\end{align*}
    
\end{lemma}

\begin{proof}[Proof of Lemma~\ref{lem:regret-gbold-to-regret-g}]
We start from
    \begin{align*}
R_{\lambda}(\mathbf{g},\mathbf{F}) &= \lambda b - a \left( g_u(a) + (n-1) g_d(a) \right) F(a)^n - \left( b g_u(b) + (n-1)a g_d(b) \right) \left( 1 -(1-f_b)^n \right) \\
&+ (b-a) g_d(b) \left( 1-(1-f_b)^{n-1} (1+(n-1)f_b) \right) \\
&+ \int_{v=a}^{v=b} -\lambda F(v)^n + g_u(v) n F(v)^{n-1} (1-F(v) - vF'(v) ) dv \\
&+ \int_{v=a}^{v=b} g_d(v) n(n-1) F(v)^{n-2} F'(v) \left\{ (v-a)(1-F(v)) - a F(v) \right\} dv \, . \tag{Regret-$\bold{g}$}  
\end{align*}
The first line is 
\begin{align*}
    \lambda b - a (n\alpha) F(a)^n - (b \cdot 1 + (n-1) a \cdot 0) (1-(1-f_b)^n) \\
    = \lambda b - a n \alpha F(a)^n - b + b(1-f_b)^n \, .
\end{align*}
The second line is zero because $g_d(b) = 0$.
The third and fourth line (with integrals) are
\begin{align*}
    &\int_{v=a}^{v=v^*} - \lambda F(v)^n + g(v) n F(v)^{n-1} (1-F(v)-vF'(v))\\
    &\quad + \alpha n(n-1) F(v)^{n-2} F'(v) \left\{ (v-a)(1-F(v))-aF(v) \right\} dv \\
    &+ \int_{v=v^*}^{v=b} - \lambda F(v)^n + g(v) n F(v)^{n-1} (1-F(v)-vF'(v))\\
    &\quad+  (1-g(v)) n F(v)^{n-2} F'(v) \left\{ (v-a)(1-F(v))-aF(v) \right\} dv \, ,
\end{align*}
which simplifies to the integrals in the lemma statement.
\end{proof}

\begin{proposition}[Minimax $\lambda$-Regret Expressions]\label{prop:minimax-lmbd-regret-expressions}
The minimax $\lambda$-regret expressions given in Theorem~\ref{thm:char-all-mech-full-main} are correct.
\end{proposition}

\begin{proof}[Proof of Proposition~\ref{prop:minimax-lmbd-regret-expressions}]
We first consider the case $a/b \leq k_l$, then $v^* = b$, $\alpha = 0$. We note that
$f_b = r^*/b = k_l$ and $1-F^*(v) - v(F^*)'(v) = 0$ so only the first integral is nonzero, and the $g^*(v)$ term disappears. Therefore, the minimax $\lambda$-regret is 
\begin{align*}
    \lambda b - b + b (1-k_l)^{n} + \int_{v=k_l b}^{v=b} - \lambda \left( 1 - \frac{k_l b}{v} \right)^{n} = - (1-\lambda) b + \left[ (1-k_l)^{n} - \lambda \int_{t=k_l}^{t=1} \left( 1 - \frac{k_l}{t} \right)^{n} dt \right] b \, .
\end{align*}
Now we consider the case $a/b \geq k_h$, then $v^* = a$, $\alpha = 1/n$. We note that
\begin{align*}
    a - \phi_0 = a - \frac{a - k_h b}{1 - k_h} = \frac{k_h}{1-k_h}(b-a) = k_h' (b-a) \, ,
\end{align*}
and $b - \phi_0 = (b-a) + (a-\phi_0) = (1+k_h')(b-a)$, so $f_b = \frac{a-\phi_0}{b-\phi_0} = \frac{k_h'}{1+k_h'} = k_h$. We let $v = a + (b-a) t$, so as $v$ varies from $a$ to $b$, we have $t$ varies from 0 to 1. We also have
\begin{align*}
    F^*(v) = \frac{v-a}{v-\phi_0} = \frac{(b-a)t}{a-\phi_0 + (b-a) t} = \frac{t}{k_h' + t} \, ,
\end{align*}
and
\begin{align*}
    (F^*)'(v) = \frac{(a-\phi_0)}{(v-\phi_0)^2} = \frac{1}{b-a} \frac{k_h'}{(k_h'+t)^2} \, .
\end{align*}
Therefore, the minimax regret is
\begin{align*}
    \lambda b - b  + b(1-k_h)^n + \int_{v=a}^{v=b} \left[ - \lambda F^*(v)^{n} + n F^*(v)^{n-2} (F^*)'(v) (v-a) - n F^*(v)^{n-1} v F'(v) \right] dv \\
    = -(1-\lambda) b + (1-k_h)^{n} b + \int_{t=0}^{t=1} \Big[ - \lambda \left( \frac{t}{k_h'+t} \right)^{n} + n \left( \frac{t}{k_h'+t} \right)^{n-2} \frac{1}{(b-a)}  \frac{k_h'}{(k_h'+t)^2} \cdot (b-a) t \\ - n \left( \frac{t}{k_h' + t} \right)^{n-1} (a + (b-a) t) \frac{1}{b-a} \frac{k_h'}{(k_h'+t)^2}  \Big] (b-a) dt \, .
\end{align*}

We have
\begin{align*}
    f_b = \frac{a-\phi_0}{b - \phi_0} = \frac{(1-k_h) a - (a-k_h b)}{(1-k_h) b - (a-k_h b)} = k_h,
\end{align*}
and $F^*(v) - F^*(v)^2 - (v-a) (F^*)'(v) = 0$ so only the second integral is nonzero, and the $g^*(v)$ term disappears. We have $F^*(v) = (v-a)/(v-\phi_0)$ and $(F^*)'(v) = (a-\phi_0)/(v-\phi_0)^2$ Therefore, the minimax $\lambda$-regret is
\begin{align*}
    \lambda b - b + b (1-k_h)^{n}  - a \int_{v=a}^{v=b} n F^*(v)^{n-1} (F^*)'(v) dv \\+ \int_{v=a}^{v=b} \Big[ - \lambda F^*(v)^n + n F^*(v)^{n-2} (v-a) (F^*)'(v) (1-F^*(v))    \Big] dv \, .
\end{align*}
We separate the integral into two because we can directly evaluate the first one, and the second one can be written purely in terms of $t$ when we write $v = a + (b-a) t$. The first integral is
\begin{align*}
    a \left[ F^*(v)^{n} \right]_{v=a}^{v=b}  = a ( F^*(b^-)^{n} - F^*(a)^{n} ) = a ( (1-k_h)^n -0^{n} )b = a(1-k_h)^{n} \, .
\end{align*}

For the second integral, we write $v = a + (b-a) t$ for $t \in [0,1]$ and use the expressions for $F^*(v)$ and $(F^*)'(v)$ in terms of $t$ to get
\begin{align*}
    (b-a) \int_{t=0}^{t=1} \Big[ - \lambda \left( \frac{t}{k_h'+t} \right)^{n} + n \left( \frac{t}{k_h'+t} \right)^{n-2}  \cdot t \cdot \frac{k_h'}{(k_h'+t)^2} \cdot \frac{k_h'}{k_h'+t} \Big] dt \\
    = (b-a) \int_{t=0}^{t=1} \frac{t^{n-1}  ( n k_h'^2 - \lambda t (t+k_h')}{(t+k_h')^{n+1}} dt \, .
\end{align*}
Substituting these in gives the expression for minimax $\lambda$-regret for $a/b \geq k_h$.

Lastly, we consider the case $k_l \leq a/b \leq k_h$. Here, $F^*(v) = 1 - a/v$, so we have that $(v-a)(1-F^*(v)) - a F^*(v) = 0 $ and $F^*(v) - F^*(v)^2 - (v-a) (F^*)'(v) = 0$. We also have $F^*(a) = 0$ and $f_b = a/b$. Therefore, all terms in the integrals both from $a$ to $v^*$ and from $v^*$ to $b$ vanish except the $-\lambda F(v)^{n}$ term. Therefore, the minimax $\lambda$-regret is
\begin{align*}
    \lambda b - b + b (1-f_b)^{n} - \lambda \int_{v=a}^{v=b} \left( 1 - \frac{a}{v} \right)^{n} dv \, ,
\end{align*}
which is the desired expression.
\end{proof}

\subsection{Optimal Pricing}\label{app:subsec:optimal-pricing}

}

\section{Proofs and Discussions from Section~\ref{sec:other-mech-classes}}\label{app:sec:other-mech-classes}

\subsection{Proofs and Discussions from Section~\ref{subsec:main-std-mech}}\label{app:subsec:main-std-mech}

Before we prove the main theorem characterizing the minimax $\lambda$-regret standard mechanism (Theorem~\ref{thm:char-std-mech}), we first derive the regret expression of generous SPA $\textnormal{GenSPA}(\Phi)$.

\begin{proposition}\label{prop:reg-std}
The regret of the mechanism $\textnormal{GenSPA}(\Phi)$ under distribution $\mathbf{F}$ that is $n$ i.i.d. with marginal $F$ is
\begin{align*} 
    & a(\lambda-\Phi(a)) + \int_{v=a}^{v=b} (\lambda-\Phi(v)-v \Phi'(v)) (1-F(v)^n) + \Phi(v) (n F(v)^{n-1} - n F(v)^n) dv \\
    &\quad + F(a)^{n-1} (b-a)(n-(n-1)F(a)) \\
    &\quad - F(a)^{n-1} \int_{v=a}^{v=b} \left( (n-(n-1)F(a)) \Phi(v) + v(nF(v)-(n-1)F(a)) \Phi'(v)  \right) dv \, .
\end{align*}
\end{proposition}

\begin{proof}[Proof of Proposition~\ref{prop:reg-std}]

The $\lambda$-regret expression, pointwise at $\mathbf{v}$, is 
\begin{align*}
    \lambda v^{(1)} - \left( v^{(1)} \Phi(v^{(1)}) - \int_{t=v^{(2)}}^{t=v^{(1)}} \Phi(t) dt \right) \1(v^{(2)} > a) - \left( v^{(1)} \cdot 1 - \int_{t=v^{(2)}}^{t=v^{(1)}} 1 dt \right) \1(v^{(2)} = a) \\
    = \lambda v^{(1)} - \left( v^{(1)} \Phi(v^{(1)}) - \int_{t=v^{(2)}}^{t=v^{(1)}} \Phi(t) dt \right) \1(v^{(2)} > a) - \left( v^{(2)}  \right) \1(v^{(2)} = a) \, .
\end{align*}

Writing $\1(v^{(2)} > a) = 1 - \1(v^{(2)} = a)$, the regret expression becomes
\begin{align*}
\left(  \lambda v^{(1)} - v^{(1)} \Phi(v^{(1)}) + \int_{t=v^{(2)}}^{t=v^{(1)}} \Phi(t) dt \right) + \left(  v^{(1)} \Phi(v^{(1)}) - \int_{t=v^{(2)}}^{t=v^{(1)}} \Phi(t) dt - v^{(2)} \right) \1(v^{(2)} = a) \\
= \left(  \lambda v^{(1)} - v^{(1)} \Phi(v^{(1)}) + \int_{t=v^{(2)}}^{t=v^{(1)}} \Phi(t) dt \right) + \left(  v^{(1)} \Phi(v^{(1)}) - \int_{t=a}^{t=v^{(1)}} \Phi(t) dt - a \right) \1(v^{(2)} = a) \, .
\end{align*}

We will now calculate the distribution of $v^{(1)} | v^{(2)} = a$. 

We note that, for any $x \in [a,b]$,
\begin{align*}
    \Pr(v^{(1)} > x,v^{(2)} = a) = \Pr(\text{ $n-1$ $v$'s are $a$, one is $> x$}) = n F(a)^{n-1} F((x,b])) = n F(a)^{n-1} (1-F(x)) \, ,
\end{align*}
and
\begin{align*}
    \Pr(v^{(2)} = a) = \Pr(\text{ exactly $n-1$ are $a$ }) + \Pr(\text{ exactly $n$ are $a$ }) \\
    = n F(a)^{n-1} (1-F(a)) + F(a)^n = F(a)^{n-1} (n-(n-1)F(a)) \, .
\end{align*}

If $\Pr(v^{(2)} = a) > 0$ (that is, if $F$ has an atom at $a$), then dividing the two equations gives
\begin{align*}
    \Pr(v^{(1)} > x| v^{(2)} = a) = \frac{n(1-F(x))}{n-(n-1)F(a)} \quad \text{ for } x \in [a,b] \, .
\end{align*}
This gives us the CDF of $v^{(1)}|v^{(2)} = a$ as
\begin{align*}
\tilde{F}(x) := \Pr(v^{(1)} \leq x | v^{(2)} = a) = 1 -  \frac{n(1-F(x))}{n-(n-1)F(a)} = \frac{n F(x) - (n-1) F(a)}{n-(n-1)F(a)} \quad \text{ for } x \in [a,b] \, .
\end{align*}
(If $\Pr(v^{(2)} = a) = 0$, then the expression we want to evaluate is zero and we won't need this conditional distribution anyway.)

We can now calculate the expected regret. The first term is exactly the expected regret expression from SPA:
\begin{align*}
    &\bE_{\mathbf{v} \sim F} \left[ \lambda v^{(1)} - \left( v^{(1)} \Phi(v^{(1)}) - \int_{t=v^{(2)}}^{t=v^{(1)}} \Phi(t) dt \right) \right] \\
    &= a(\lambda-\Phi(a)) + \int_{v=a}^{v=b} (\lambda-\Phi(v)-v \Phi'(v)) (1-F(v)^n) + \Phi(v) (n F(v)^{n-1} - n F(v)^n) dv \, .
\end{align*}

The second term can be written as
\begin{align*}
    &\bE_{\mathbf{v} \sim F}  \left[     \left(  v^{(1)} \Phi(v^{(1)}) - \int_{t=v^{(2)}}^{t=v^{(1)}} \Phi(t) dt - a \right) \1(v^{(2)} = a) \right] \\
    &= \Pr(v^{(2)} = a) \bE \left[ \left(  v^{(1)} \Phi(v^{(1)}) - \int_{t=a}^{t=v^{(1)}} \Phi(t) dt - a \right) \Bigg| v^{(2)} = a \right] \\
    &= F(a)^{n-1} (n-(n-1)F(a))  \left[ -a + \int_{v \in [a,b]} \left( v \Phi(v) - \int_{t=a}^{t=v} \Phi(t) dt  \right) d\tilde{F}(v) \right] \, .
\end{align*}

We now use the following integration-by-part-like statements.
\begin{align*}
    \int_{w \in [a,b]} h(w) dG(w) &= h(b)  - \int_{v=a}^{v=b} h'(v)G(v) dv \, .
\end{align*}

Therefore,
\begin{align*}
&\int_{v \in [a,b]} v \Phi(v) d\tilde{F}(v) \\
&= b \Phi(b) - \int_{v=a}^{v=b} (\Phi(v) + v \Phi'(v))\tilde{F}(v) dv \\
&= b - \int_{v=a}^{v=b} (\Phi(v) + v\Phi'(v)) \frac{n F(v) - (n-1) F(a)}{n-(n-1)F(a)} dv \, .
\end{align*}

We also have
\begin{align*}
&\int_{v \in [a,b]} \int_{t=a}^{t=v} \Phi(t) dt d\tilde{F}(v) \\
&= \int_{t=a}^{t=b} \int_{v \in (t,b]} \Phi(t) d\tilde{F}(v) dt \\
&= \int_{t=a}^{t=b} \Phi(t) (1-\tilde{F}(t)) dt \\
&= \int_{v=a}^{v=b} \Phi(v) \frac{n(1-F(v))}{n-(n-1)F(a)} dv \, .
\end{align*}

Therefore, the second term is
\begin{align*}
& F(a)^{n-1} (n-(n-1)F(a)) \left[ -a + b - \int_{v=a}^{v=b} (\Phi(v) + v\Phi'(v)) \frac{nF(v)-(n-1)F(a)}{n-(n-1)F(a)} - \int_{v=a}^{v=b} \Phi(v) \frac{n(1-F(v))}{n-(n-1)F(a)} dv \right] \\
&= F(a)^{n-1} (n-(n-1)F(a)) \left[ b-a - \int_{v=a}^{v=b} \left( \Phi(v) + v \Phi'(v) \frac{nF(v)-(n-1)F(a)}{n-(n-1)F(a)} \right) dv \right] \\
&= F(a)^{n-1} \left[ (b-a)(n-(n-1)F(a)) - \int_{v=a}^{v=b} \left( (n-(n-1)F(a)) \Phi(v) + v(nF(v)-(n-1)F(a)) \Phi'(v)  \right) dv \right] \, .
\end{align*}

We therefore have the regret expression
\begin{align*} 
    &= a(\lambda-\Phi(a)) + \int_{v=a}^{v=b} (\lambda-\Phi(v)-v \Phi'(v)) (1-F(v)^n) + \Phi(v) (n F(v)^{n-1} - n F(v)^n) dv \\
    &\quad + F(a)^{n-1} (b-a)(n-(n-1)F(a))\\
    &\quad - F(a)^{n-1} \int_{v=a}^{v=b} \left( (n-(n-1)F(a)) \Phi(v) + v(nF(v)-(n-1)F(a)) \Phi'(v)  \right) dv \, .
\end{align*}
Note that the extra term (in the second line) is linear in $F$. 
\end{proof}

We are now ready to state and prove the main theorem.

\begin{theorem}\label{app:thm:char-std-mech}
Fix $n$ and $\lambda \in (0,1]$ and let $k = a/b \in [0,1)$. Define $k_l$ as in Theorem~\ref{thm:char-std-mech}. The minimax $\lambda$-regret problem $R_{\lambda}(\mathcal{M}_{\textnormal{std}},\mathcal{F}_{\textnormal{iid}})$ admits the following saddle point $(m^*,F^*)$, depending on $a/b$, as follows. 
\begin{itemize}
\item For $a/b \leq k_l$, the optimal mechanism $m^*$ and worst-case distribution $F^*$ are the same as those of Theorem~\ref{thm:char-all-mech-full-main}. 
\item For $a/b \geq k_l$, the optimal mechanism $m^*$ is $\textnormal{GenSPA}(\Phi^*)$ with
\begin{align*}
 \Phi^*(v) =  \frac{ \int_{t=a}^{t=v} \frac{\lambda \left( 1 - c/t \right)^{n-1}}{t} dt	   }{  \left( 1 - c/v \right)^{n-1} - \left( 1 - c/a \right)^{n-1}  } \quad \text{ for } v \in [a,b] \, ,
\end{align*}
where $c \in [0,a]$ is a unique constant such that $\Phi^*(c) = 1$. The worst case distribution $F^*$ is an isorevenue distribution defined by $F^*(v) = 1-c/v$ for $v \in [a,b)$ and $F^*(b) = 1$.
\end{itemize}
\end{theorem}

\paragraph{Remark.} Note that the worst case distribution has two point masses at $v = a$ and $v = b$ (of size $1-c/a$ and $c/b$ respectively), whereas the worst case distributions of $\mathcal{M}_{\textnormal{all}}$ in Theorem~\ref{thm:char-all-mech-full-main} each only has one point mass at $v = b$. Also, the distributions $\Phi$ and $\Psi$ of the reserve of $\textnormal{SPA}$ and threshold of $\textnormal{POOL}$ in Theorem~\ref{thm:char-all-mech-full-main} do not have any point masses, whereas the distribution $\Phi$ of $\textnormal{GenSPA}$ in Theorem~\ref{app:thm:char-std-mech} has a point mass at $v = a$ of size $\Phi^*(a) = \lim_{v \downarrow a} \Phi^*(v) = \lambda(a/c-1)/(n-1)$.

\begin{proof}[Proof of Theorem~\ref{app:thm:char-std-mech}]
Seller's saddle is that, fixing $F^*$, the \jaedit{given} mechanism $m^*$ gives the lowest regret over all standard mechanisms. Because $F^*$ is fixed, this is equivalent to that $m^*$ maximizes expected revenue under $F^*$. This is a standard Bayesian mechanism design problem, and we check with \cite{MonteiroSvaiter10} that $m^*$ is indeed optimal, even over all DSIC mechanisms. Henceforth, we will focus on Nature's Saddle.

Note that if $a/b \leq k_l$, then Theorem~\ref{thm:char-all-mech-full-main} immediately tells us that the same $\textnormal{SPA}(\Phi^*)$ is optimal over $\mathcal{M}_{\textnormal{all}}$, and thus over $\mathcal{M}_{\textnormal{std}}$ also. Henceforth we assume $a/b \geq k_l$.

We have the regret expression from Proposition~\ref{prop:reg-std}:

\begin{align*} 
    & a(\lambda-\Phi(a)) + \int_{v=a}^{v=b} (\lambda-\Phi(v)-v \Phi'(v)) (1-F(v)^n) + \Phi(v) (n F(v)^{n-1} - n F(v)^n) dv \\
    &\quad + F(a)^{n-1} (b-a)(n-(n-1)F(a)) \\
    &\quad- F(a)^{n-1} \int_{v=a}^{v=b} \left( (n-(n-1)F(a)) \Phi(v) + v(nF(v)-(n-1)F(a)) \Phi'(v)  \right) dv \, .
\end{align*}

The additional term is linear in $F$, so if pointwise optimization gives a global maximum then, it will also give a global maximum now. The first-order condition gives
\begin{align*}
(\lambda - \Phi(v) - v \Phi'(v)) (-nF(v)^{n-1}) + \Phi(v) ( n(n-1)F(v)^{n-2} - n^2 F(v)^{n-1}) - nv F(a)^{n-1} \Phi'(v) = 0 \, ,
\end{align*}
or
\begin{align*}
v(F(v)^{n-1} - F(a)^{n-1}) \Phi'(v) + (n-1) F(v)^{n-2} (1-F(v)) \Phi(v) = \lambda F(v)^{n-1} \, .
\end{align*}

We want $\Phi$ such that $F(v) = 1-c/v$ is a solution to that equation. With $F(v) = 1-c/v$ we have $F'(v) = \frac{c}{v^2} = \frac{1-F(v)}{v}$. So
\begin{align*}
(F(v)^{n-1} - F(a)^{n-1}) \Phi'(v) + (n-1) F(v)^{n-2} F'(v) \Phi(v) = \frac{\lambda F(v)^{n-1}}{v} \, ,
\end{align*}
or
\begin{align*}
 \frac{d}{dv} \left[(F(v)^{n-1} - F(a)^{n-1}) \Phi(v)  \right]= \frac{\lambda F(v)^{n-1}}{v} \, .
\end{align*}

Because $(F(v)^{n-1} - F(a)^{n-1}) \Phi(v)$ is zero when $v = a$, we have
\begin{align*}
 \Phi(v) = \frac{ \int_{t=a}^{t=v} \frac{\lambda F(t)^{n-1}}{t} dt	   }{  F(v)^{n-1} - F(a)^{n-1}  } =  \frac{ \int_{t=a}^{t=v} \frac{\lambda \left( 1 - \frac{c}{t}\right)^{n-1}}{t} dt	   }{  \left( 1 - \frac{c}{v} \right)^{n-1} - \left( 1 - \frac{c}{a} \right)^{n-1}  } \, .
\end{align*}

We will now show that, for any $c \in (0,a]$, $\Phi(v)$ is increasing in $v$.

Let $\tilde{v} = 1-c/v$, $\tilde{a} = 1-c/a$, and use the substitution $u = 1-c/t$, $t = \frac{c}{1-u}$, $dt = \frac{c}{(1-u)^2} du$ to get
\begin{align*}
    \frac{1}{\lambda} \Phi(v) = \frac{ \int_{u=\tilde{a}}^{u=\tilde{v}} \frac{u^{n-1}}{1-u} du }{\tilde{v}^{n-1} - \tilde{a}^{n-1}} \, .
\end{align*}

Showing that $\Phi(v)$ is increasing in $v$ is equivalent to showing that the right hand side is increasing in $\tilde{v} = 1-c/v$. From $v \geq a$ we have $\tilde{v} \geq \tilde{a}$. The derivative with respect to $\tilde{v}$ of the right hand side is
\begin{align*}
    \frac{(\tilde{v}^{n-1}-\tilde{a}^{n-1}) \frac{\tilde{v}^{n-1}  }{1-\tilde{v}} - \left( \int_{u=\tilde{a}}^{u=\tilde{v}} \frac{u^{n-1}}{1-u} du \right) (n-1)\tilde{v}^{n-2}  }{  (\tilde{v}^{n-1} - \tilde{a}^{n-1})^2  } \, .
\end{align*}

This expression is nonnegative if and only if
\begin{align*}
    \int_{u=\tilde{a}}^{u=\tilde{v}} \frac{u^{n-1}}{1-u} du \leq \frac{\tilde{v}^n - \tilde{a}^{n-1} \tilde{v}}{(n-1)(1-\tilde{v})} \, .
\end{align*}

It is clear that this is true if it holds for $\tilde{v} = a$ and the derivative of LHS is $\leq$ the derivative of RHS. For $\tilde{v} = a$, both sides are zero, so the inequality holds. The condition that the derivative of LHS is $\leq$ the derivative of RHS is
\begin{align*}
    \frac{\tilde{v}^{n-1}}{1-\tilde{v}} \leq \frac{1}{n-1} \frac{ (1-\tilde{v})(n\tilde{v}^{n-1}-\tilde{a}^{n-1}) - (\tilde{v}^{n} - \tilde{a}^{n-1} \tilde{v})(-1) }{(1-\tilde{v})^2} \, .
\end{align*}
This is equivalent to 
\begin{align*}
    (n-1) \tilde{v}^{n-1} - (n-1) \tilde{v}^n \leq n \tilde{v}^{n-1} - \tilde{a}^{n-1} - n \tilde{v}^n + \tilde{a}^{n-1} \tilde{v} + \tilde{v}^n - \tilde{a}^{n-1} \tilde{v} \, ,
\end{align*}
or
\begin{align*}
    \tilde{a}^{n-1} \leq \tilde{v}^{n-1} \, ,
\end{align*}
which is true because $\tilde{v} \geq \tilde{a}$.

Lastly, we will show that if $a/b \geq k_l$, then there is a $c \in (0,a]$ such that $\Phi(b) = 1$, making this a valid solution. 
From
\begin{align*}
 \Phi(b) =   \frac{ \int_{t=a}^{t=b} \frac{\lambda \left( 1 - \frac{c}{t}\right)^{n-1}}{t} dt	   }{  \left( 1 - \frac{c}{b} \right)^{n-1} - \left( 1 - \frac{c}{a} \right)^{n-1}  } \, ,
\end{align*}
consider the right hand side as a function of $c$. The numerator is clearly a decreasing function of $c$, because $\frac{\lambda \left( 1 - \frac{c}{t}\right)^{n-1}}{t} $, for each fixed $t$, is a decreasing function of $c$. The denominator is an increasing function of $c$ because 
\begin{align*}
 \left( 1 - \frac{c}{b} \right)^{n-1} - \left( 1 - \frac{c}{a} \right)^{n-1} = c \left( \frac{1}{a} - \frac{1}{b} \right) \left( \sum_{k=0}^{n-1} \left( 1 - \frac{c}{b} \right)^{k} \left( 1 - \frac{c}{a} \right)^{n-1-k} \right) \, ,
\end{align*}
and each of $c$, $1-\frac{c}{b}$, $1-\frac{c}{a}$ are positive and increasing in $c$, and $\frac{1}{a} -\frac{1}{b} > 0$.

Therefore, the right hand side is decreasing in $c$. As $c \downarrow 0$, the numerator converges to $\int_{t=a}^{t=b} \frac{\lambda }{t} dt = \lambda \log(b/a) > 0$, and the denominator converges to 0, so the expression converges to $+\infty$. At $c = a$, the expression is
\begin{align*}
\frac{ \int_{t=a}^{t=b} \frac{\lambda \left( 1 - \frac{a}{t}\right)^{n-1}}{t} dt	   }{  \left( 1 - \frac{a}{b} \right)^{n-1}  } \leq 1 \, ,
\end{align*}
where the $\leq 1$ holds because $a/b \geq k_l$. Therefore, there is a $c \in (0,a]$ such that $\Phi(b) = 1$. 

Lastly, we note that
\begin{align*}
    \Phi(a) = \lim_{v \downarrow a} \Phi(v) = \lambda \lim_{v \downarrow a} \frac{ \frac{\left( 1 - \frac{c}{v} \right)^{n-1}}{v} }{  (n-1) \left( 1 - \frac{c}{v} \right)^{n-2} \frac{c}{v^2} } = \lambda \frac{(a-c)}{(n-1)c} \, ,
\end{align*}
where the second equality holds by L'Hopital's rule. 
\end{proof}

\subsection{Proofs and Discussions from Section~\ref{subsec:main-spa-rand-mech}}\label{app:subsec:main-spa-rand-mech}

The main theorem in the main text (Theorem~\ref{thm:char-spa-rand-mech}) is an immediate corollary of the main theorem in this Appendix (Theorem~\ref{app:thm:char-spa-rand-mech}). We first outline key challenges of the proof before diving into the full proof of the main theorem.

\begin{proof}[Key Challenges of the Proof of The Main Theorem]
The proof of the moderate information regime is the hardest, and we outline key technical ideas here. Our candidate $\Phi^*$ has a point mass $\Phi_0$ at $a$ and a density on $[r^*,b]$, while the candidate worst-case distribution $F^*$ has a point mass $F_0$, an isorevenue density on $[r^*,b)$, and a point mass at $b$. 
The seller's saddle is to find $\Phi$ that minimizes the regret $R(\Phi) := R(\Phi,F^*)$ such that $\Phi(v) \in [0,1]$ is an increasing function. This is different from previous saddle problems because in this case, the increasing condition is \textit{binding}; if we optimize pointwise, the result is nonincreasing, which is infeasible. We write the Lagrangian $\mathcal{L}(\Phi,\mu) = R(\Phi) - \int_{a}^{b} \mu(v) d\Phi(v)$. Here, $\mu: [a,b] \to \mathbb{R}_{+}$ is the dual variable associated with the increasing constraint. Complementary slackness requires that $\mu^*$ is zero wherever $\Phi^*$ is strictly increasing; this immediately suggests the correct form of $\Phi^*$, and that $\mu^*(v) = 0$ on $[r^*,b]$. Lagrangian optimality requires that  $\Phi^*$ also minimizes $\mathcal{L}(\Phi,\mu^*)$, which is linear in $\Phi$ (as can be made explicit by integration by parts on the $d\Phi(v)$ term). We satisfy this by requiring that the coefficients of every $\Phi(v)$ term to be zero. These, together with complementary slackness, pin down $\mu^*$. Nature's saddle is significantly simpler because the condition that $F$ is increasing does not bind here, so pointwise optimization (as before) works.  The full proof is given in Appendix~\ref{app:subsec:main-spa-rand-mech}.
\end{proof}

Before we proceed to the proof, we derive the regret expression for $\textnormal{SPA}(\Phi)$, which is just $(g_u,g_d)$ with $g_u = \Phi$ and $g_d \equiv 0$.

\begin{proposition}\label{prop:reg-spa-rand}
Let the mechanism be a second-price auction with random reserve CDF $\Phi$ and distribution $\mathbf{F}$ with regret $R(\Phi,\mathbf{F})$.

Suppose $\Phi: [a,b] \to [0,1]$ is absolutely continuous, while $\mathbf{F}$ is arbitrary then we have the (Regret-$\mathbf{F}$) expression
\begin{align*}
R(\Phi,\mathbf{F}) &= a(\lambda-\Phi(a)) + \int_{v \in [a,b]} (\lambda-\Phi(v)-v \Phi'(v)) (1-F^{(1)}(v)) dv \\&\quad
+\int_{v  \in [a,b]} \Phi(v) (F^{(2)}(v) - F^{(1)}(v)) dv \, .
\end{align*}
If we further assume that $\mathbf{F}$ is $n$ i.i.d. then
\begin{align*}
R(\Phi,\mathbf{F}) &= a(\lambda-\Phi(a)) + \int_{v \in [a,b]} (\lambda-\Phi(v)-v \Phi'(v)) (1-F(v)^n) dv \\&\quad
+\int_{v  \in [a,b]} \Phi(v) (n F(v)^{n-1} - n F(v)^n) dv \, .
\end{align*}

Suppose instead that $\mathbf{F}$ is i.i.d. with marginal $F$ that has a density in $(a,b)$, and we denote by $F(\{b\})$ and $\mathbf{F}^{(1)}(b)$ the mass at $b$. Then we have the (Regret-$\Phi$) expression
\begin{align*}
R(\Phi,F) &= \lambda b - a \Phi(a) F(a)^n - (1-(1-f_b)^n) b \Phi(b) + \\&\quad \int_{v=a}^{v=b} - \lambda F(v)^n + \Phi(v) n F(v)^{n-1} (1-F(v) - vF'(v)) dv \, .
\end{align*}

\end{proposition}

\begin{proof}[Proof of Proposition~\ref{prop:reg-spa-rand}]
These expressions follow immediately by taking $g_u(v) = \Phi(v)$ and $g_d(v) = 0$.
\end{proof}

We are now ready to state and prove the main theorem.

\begin{theorem}[Optimal SPA with Random Reserve]\label{app:thm:char-spa-rand-mech}
Fix $n$ and $\lambda \in (0,1]$ and let $k = a/b \in [0,1)$. Define $k_l$ and $k_h'$ as in Theorem~\ref{thm:char-spa-rand-mech}. The minimax $\lambda$-regret problem $R_{\lambda}(\mathcal{M}_{\textnormal{SPA-rand}},\mathcal{F}_{\textnormal{iid}})$ admits a  saddle point $(m^*=\textnormal{SPA}(\Phi^*),F^*)$ which is characterized as follows.  
\begin{itemize}
\item For $a/b \leq k_l$, an optimal mechanism $m^*$ and worst-case distribution $F^*$ are the same as those identified in Theorem~\ref{thm:char-all-mech-full-main} and its proof.  
\item For $a/b \geq k_h'$, an optimal mechanism is a SPA with no reserve and the worst-case distribution $F^*$ is a two-point distribution with point masses at $v = a$ and $v=b$ with weights $f_a := \frac{n-1}{n-1+\lambda}$ and $f_b := \frac{\lambda}{n-1+\lambda}$. 
\item For $k_l \leq a/b \leq k_h'$, let $F_0 \in [0,1)$ be a unique solution to
\begin{align*}
&\frac{ \lambda F_0^n}{(n-1)(1-F_0)} - \left( 1 - \frac{n-(n-1)F_0}{n} \tilde{a} \right)^{n-1} \\&\quad  = \lambda \left[ \log\left( \frac{n-(n-1)F_0}{n(1-F_0)} \tilde{a} \right) - \sum_{k=1}^{n-1} \frac{1}{k} F_0^k + \sum_{k=1}^{n-1} \frac{1}{k} \left( 1 - \frac{n-(n-1)F_0}{n} \tilde{a} \right)^{k} \right] ,
\end{align*}
and
\begin{align*}
r^* &= \frac{n-(n-1)F_0}{n(1-F_0)}a,\\
\Phi_0 &= \frac{\lambda F_0}{(n-1)(1-F_0)},\\
c &= \frac{n-(n-1)F_0}{n} a,\\
d &= \frac{F_0^{n}}{(n-1)(1-F_0)} + \log(1-F_0) + \sum_{k=1}^{n-1} \frac{1}{k} F_0^k \, .
\end{align*}
Then, an optimal mechanism (optimal reserve distribution) is 
\begin{align*}
\Phi^*(v) = \begin{cases}
\Phi_0 &\text{ for } v \in [a,r^*] \\
\lambda \left( \frac{v}{v-c} \right)^{n-1} \left[ d + \log\left( \frac{v}{c} \right) - \sum_{k=1}^{n-1} \frac{1}{k} \left( \frac{v-c}{v} \right)^{k} \right] &\text{ for } v \in [r^*,b] \, ,
\end{cases} 
\end{align*}
and a worst-case distribution is given by $F^*(v) = F_0$ for $v \in [a,r^*]$, $F^*(v) = 1-c/v$ for $v \in [r^*,b)$ and $F^*(b) = 1$.
\end{itemize}
\end{theorem}

\begin{proof}[Proof of Theorem~\ref{app:thm:char-spa-rand-mech}]

We will prove the 3 cases separately.

\paragraph{Low Information ($a/b \leq k_l$) Regime.} We know from Theorem~\ref{thm:char-all-mech-full-main} that $\textnormal{SPA}(\Phi^*)$ is minimax optimal over $\mathcal{M}_{\textnormal{all}}$ and so it is also minimax optimal over $\mathcal{M}_{\textnormal{SPA-rand}}$.

\paragraph{High Information ($a/b \geq \frac{\lambda n}{(1+\lambda)n-1}$) Regime.}

We claim that the minimax optimal mechanism $\Phi^*$ in this regime is SPA without reserve and $F^*$ is a two-point distribution with mass $f_a := \frac{n-1}{n-1+\lambda}$ at $a$ and $f_b := \frac{\lambda}{n-1+\lambda}$ at $b$. Here, $R(\Phi^*,F^*) = -(1-\lambda)b + \left(\frac{n-1}{n-1+\lambda}\right)^{n-1} (b-a)$.

We want to show that $R(\Phi^*,F) \leq R(\Phi^*,F^*) \leq R(\Phi,F^*)$.

\underline{Part 1: $R(\Phi,F^*) \geq R(\Phi^*,F^*)$}

We have the regret expression $R(\Phi,F)$ from Proposition~\ref{prop:reg-spa-rand}:
\begin{align*}
R(\Phi,F) &= \lambda b - a \Phi(a) F(a)^n - (1-(1-f_b)^n) b \Phi(b) \\&\quad+ \int_{v=a}^{v=b} - \lambda F(v)^n + \Phi(v) n F(v)^{n-1} (1-F(v) - vF'(v)) dv \, .
\end{align*}

Under $F^*$, we have $F^*(v) = F(a) = f_a \frac{n-1}{n-1+\lambda}$ for all $v \in (a,b)$, which means every appearance of $F(v)$ becomes a constant:
\begin{align*}
R(\Phi,F^*) = \lambda b - a f_a^n \Phi(a) - (1-f_a^n) b \Phi(b) + \int_{v=a}^{v=b} -\lambda f_a^n + n f_a^{n-1} (1-f_a) \Phi(v) dv \, .
\end{align*}

We have $\Phi(b) \leq 1$ and $-\lambda f_a^n + n f_a^{n-1} (1-f_a) \Phi(v) \geq -\lambda f_a^n + n f_a^{n-1} (1-f_a) \Phi(a)$, so
\begin{align*}
R(\Phi,F^*) \geq \lambda b - a \Phi(a) f_a^n -(1-f_a^n) b + (b-a) (-\lambda f_a^n + nf_a^{n-1} (1-f_a) \Phi(a)) \\
= - (1-\lambda) b + (n f_a^{n-1} - (n-1+\lambda) f_a^n) (b-a) + (1-\Phi(a)) f_a^{n-1} ( a f_a - (b-a)(n-nf_a)) \, .
\end{align*}

The last term is $\geq 0$ because we require that $a f_a - (b-a)(n-nf_a) \geq 0 \Leftrightarrow \frac{a}{b} \geq \frac{n-n f_a}{n-(n-1) f_a} = \frac{\lambda n}{(1+\lambda)n-1}$. We now see that the choice $f_a = \frac{n-1}{n-1+\lambda}$ is chosen so the second term is maximized, and the bound on $\frac{a}{b}$ that is required to make the third term work follows accordingly.  Therefore,
\begin{align*}
R(\Phi,F^*) \geq  - (1-\lambda) b + (n f_a^{n-1} - (n-1+\lambda) f_a^n) (b-a)  = -(1-\lambda) b +  \left( \frac{n-1}{n-1+\lambda} \right)^{n-1} (b-a) \, .
\end{align*}

\underline{Part 2: $R(\Phi^*,F) \leq R(\Phi^*,F^*) $}

Because $\Phi^*$ is an SPA without reserve, we have $\Phi^*(v) = 1$ for every $v$. 
This gives
\begin{align*}
R(\Phi^*,F) = a(\lambda-1) + \int_{v=a}^{v=b} (\lambda - 1)(1-F(v)^n) + (nF(v)^{n-1} - n F(v)^n) dv \\
= -(1-\lambda) b + \int_{v=a}^{v=b} n F(v)^{n-1} -(n-1+\lambda) F(v)^n dv \\
\leq -(1-\lambda) b + (b-a) \sup_{z \in [0,1]} n z^{n-1} - (n-1+\lambda) z^n \\
= -(1-\lambda) b + \left( \frac{n-1}{n-1+\lambda} \right)^{n-1}(b-a) = R(\Phi^*,F^*) \, .
\end{align*}

\paragraph{Moderate Information ($k_l \leq a/b \leq \frac{\lambda n}{(1+\lambda)n-1}$) Regime.}

We will exhibit a saddle point with the following structure: $\Phi^*$ has a point mass $\Phi^*(a)$ at $a$, then it is flat on $[a,r^*]$ (that is, $\Phi^*(v) = \Phi^*(a)$ on $[a,r^*]$), then it has a density on $[r^*,b]$, but no point mass at $b$ (so $\Phi^*(b) = 1$). $F^*$ has a point mass $F^*(a)$ at $a$, then it is flat on $[a,r^*]$, then it is $F^*(v) = 1-c/v$ on $[r^*,b)$. That is, $F^*(v) = \max(1-c/v,F^*(a))$ on $v \in [a,b)$. Importantly, we assume that $F^*$ is continuous at $r^*$, so $1-c/r^* = F^*(a)$, and this $r^*$ is the same as $r^*$ of $\Phi^*$.

We use the regret expression (Regret-$\Phi$)  under $F^*$:
\begin{align*}
    &R_n(\mathbf{\Phi},\mathbf{F}) = \lambda b - F(a)^n a \Phi(a) - (1-(1-f_b)^n) b \Phi(b)\\&\quad + \int_{v=a}^{v=b} -\lambda F(v)^n + \Phi(v) n
    F(v)^{n-1} (1-F(v)-vF'(v)) dv \, .
\end{align*}

Note that $F^*(b^-) =  1-f_b = 1-c/b$. We will write $F_0 := F^*(a)$ for convenience. 

We can then write
\begin{align*}
    R_n(\Phi, \mathbf{F}^*) = \lambda b - F_0^n a \Phi(a) - (1 - (1-c/b)^n) b \Phi(b) + \int_{v=a}^{v=r^*} - \lambda F_0^n +  \Phi(v) n F_0^{n-1} (1-F_0) dv \\ + \int_{v=r^*}^{v=b} - \lambda \left( 1 - \frac{c}{v} \right)^n dv \, .
\end{align*}

Consider the problem
\begin{align*}
    R^* := \min_{\Phi} R_n(\Phi, \mathbf{F}^*) = \text{Regret}(\Phi) \text{ s.t. } \Phi(v) \in [0,1] \text{ non-decreasing } \, .
\end{align*}

We dualize the non-decreasing constraint. 

Let
\begin{align*}
    \mathcal{L}(\Phi,\mu) = \text{Regret}(\Phi) - \int_{a}^{b} \mu(v) d\Phi(v) \, ,
\end{align*}
and
\begin{align*}
    q(\mu) = \min_{\Phi(v) \in [0,1]} \mathcal{L}(\Phi,\mu) \text{ with } \mu(v) \geq 0 \, .
\end{align*}

Weak duality says that $R^* \geq q(\mu)$ for all $\mu: [a,b] \to \mathbb{R}_+$.

To get to strong duality, we want to choose a specific $\mu^*$ such that:
\begin{itemize}
    \item Complementary Slackness (CS): $\int_{a}^{b} \mu^*(v) d\Phi^*(v) = 0$. That is, wherever $\Phi^*$ is strictly increasing, $\mu^*$ is zero.
    \item Lagrangian Optimality (LO): $\Phi^* \in \arg\min_{\Phi \in \text{CDF} } \mathcal{L}(\Phi,\mu^*)$.
\end{itemize}

Because the condition that $\Phi(v)$ is non-decreasing doesn't bind (it is strictly increasing) on $[r^*,b]$, by complementary slackness $\mu^*(v) = 0$ on $[r^*,b]$. 

We can then use the integration by partss to get (using $\mu^*(r^*) = 0$)
\begin{align*}
    \int_{v=a}^{v=b} \mu^*(v) d\Phi(v) = \int_{v=a}^{v=r^*} \mu^*(v) d\Phi(v) = \mu^*(r^*) \Phi(r^*) - \mu^*(a) \Phi(a) - \int_{v=a}^{v=r^*} (\mu^*)'(v) \Phi(v) dv \\ 
    = - \mu^*(a) \Phi(a) - \int_{v=a}^{v=r^*} (\mu^*)'(v) \Phi(v)  dv \, .
\end{align*}

We substitute this into the expression for $\mathcal{L}(\Phi,\mu^*)$ to get

\begin{align*}
    \mathcal{L}(\Phi,\mu^*) = \lambda b   - \lambda (r^*-a) F_0^n + (\mu^*(a) - F_0^n a) \Phi(a) - \left(1 - \left(1-\frac{c}{b} \right)^n \right) b \Phi(b) - \lambda \int_{v=r^*}^{v=b} \left( 1 - \frac{c}{v} \right)^n dv \\
    + \int_{v=a}^{v=r^*} \left[ (\mu^*)'(v) + nF_0^{n-1} (1-F_0) \right] \Phi(v) dv \, .
\end{align*}

We want to choose $\mu^*$ such that
\begin{align*}
    \mu^*(a) - F_0^n a &= 0\\
    (\mu^*)'(v) + n F_0^{n-1} (1-F_0) &= 0 \, ,
\end{align*}
so that the above expression for $\mathcal{L}(\Phi,\mu^*)$ becomes independent of $\Phi$ (zero out the coefficient of $\Phi(a)$ and $\Phi(v)$ between $a$ and $r^*$; there still is $\Phi(b)$ but we will let this be 1). From
\begin{align*}
    \mu^*(r^*) = \mu^*(a) + \int_{v=a}^{v=r^*} (\mu^*)'(v) dv \\
    0 = a F_0^n - (r^*-a) n F_0^{n-1} (1-F_0) \\
    a F_0 = n (r^*-a)  (1-F_0) \, .
\end{align*}

With that $\mu^*$, we have
\begin{align*}
    \mathcal{L}(\Phi,\mu^*) = \lambda b   - \lambda (r^*-a) F_0^n - \left(1 - \left(1-\frac{c}{b} \right)^n \right) b \Phi(b) - \lambda \int_{v=r^*}^{v=b} \left( 1 - \frac{c}{v} \right)^n dv \, .
\end{align*}

 If we further assume that $\Phi(b) = 1$, then
\begin{align*}
    \mathcal{L}(\Phi,\mu^*) = -(1-\lambda) b   - \lambda (r^*-a) F_0^n + b  \left(1-\frac{c}{b} \right)^n  - \lambda \int_{v=r^*}^{v=b} \left( 1 - \frac{c}{v} \right)^n dv \, .
\end{align*} 

Now we derive conditions from the fact that $F$ maximizes regret given fixed $\Phi^*$, that is, the saddle $R(\Phi^*,F) \leq R(\Phi^*,F^*)$. We use the Regret-$F$ equation
\begin{align*}
    R(\Phi,F) = a(\lambda-\Phi(a)) +  \int_{v=a}^{v=b} \left[ (\lambda-\Phi(v)-v\Phi'(v)) (1-F(v)^n) + \Phi(v) n F(v)^{n-1}(1-F(v))\right] dv \, .
\end{align*}

We do pointwise optimization for each $v$. For $v \in (a,r^*)$. $\Phi^*(v) = \Phi^*(a)$ is a constant, so $F(v)$ that maximizes that is a constant, the same for every $v$, given by 
\begin{align*}
    F^*(v) \in \arg\max_{z} (\lambda- \Phi^*(a)) (1-z^n) + \Phi^*(a) n z^{n-1} (1-z) \, .
\end{align*}
Taking the derivative of $z$ gives
\begin{align*}
    -(\lambda-\Phi^*(a)) nz^{n-1} + n \Phi^*(a) ((n-1)z^{n-2}-nz^{n-1}) = 0 \\
    -(\lambda-\Phi^*(a)) z + \Phi^*(a) (n-1-nz) = 0 \, .
\end{align*}

Write $\Phi^*(a) = \Phi_0$ for convenience. By the first-order condition, $z = F^*(a) = F_0$ satisfies this equation, so
\begin{align*}
    -(\lambda-\Phi_0) F_0 + \Phi_0 (n-1-n F_0) = 0 \, .
\end{align*}

For $v \in (r^*,b)$, $\Phi^*(v)$ is no longer a constant (but this is the regime that we have dealt with before). We have
\begin{align*}
    F^*(v) \in \arg\max_{z}   (\lambda-\Phi(v)-v\Phi'(v)) (1-z^n) + \Phi(v) n z^{n-1}(1-z) \, .
\end{align*}
The first order condition gives
\begin{align*}
    -(\lambda-\Phi(v)-v\Phi'(v)) nz^{n-1} + \Phi(v) n ((n-1)z^{n-2} - nz^{n-1}) = 0 \\
    -(\lambda-\Phi(v)-v\Phi'(v)) z + \Phi(v)  ((n-1) - nz) = 0 \, .
\end{align*}

By the first order condition, $z = F^*(v) = 1-c/v$ satisfies this equation, so
\begin{align*}
    -(\lambda-(\Phi^*)(v)-v(\Phi^*)'(v)) \left( 1 - \frac{c}{v} \right) + \Phi(v)  \left((n-1) - n\left( 1 - \frac{c}{v} \right) \right) = 0 \, ,
\end{align*}
which simplifies to
\begin{align*}
    (\Phi^*)'(v) = \frac{\lambda}{v} - \frac{(n-1)c}{v(v-c)} \Phi^*(v) \, .
\end{align*}

We have seen this ODE before. The solution is
\begin{align*}
    \left( \frac{v-c}{v} \right)^{n-1} \Phi^*(v) = \lambda \left[ d + \log\left( \frac{v}{c} \right) - \sum_{k=1}^{n-1} \frac{1}{k} \left( \frac{v-c}{v} \right)^{k} \right] \, ,
\end{align*}
for some constant $d$.

$\Phi^*(r^*) = \Phi^*(a) = \Phi_0$ is the point mass of $\Phi^*$ at $a$, which is unknown. We get
\begin{align*}
    \left( \frac{r^*-c}{r^*} \right)^{n-1} \Phi_0 = \lambda \left[ d + \log \left( \frac{r^*}{c} \right) - \sum_{k=1}^{n-1} \frac{1}{k} \left( \frac{r^*-c}{r^*} \right)^{k} \right] \, .
\end{align*}

From $1-c/r^* = F_0 \geq 0$ we have $c \leq r^*$ with strict inequality if $F_0 > 0$. So the summation makes sense (and also tells us it doesn't necessarily go away as zero like before). With $\Phi^*(b) = 1$ we get
\begin{align*}
    \left( \frac{b-c}{b} \right)^{n-1} = \lambda \left[ d + \log \left( \frac{b}{c} \right) - \sum_{k=1}^{n-1} \frac{1}{k} \left( \frac{b-c}{b} \right)^{k} \right] 
    , .
\end{align*}

Therefore, we have 5 equations for 5 unknowns $\Phi_0, F_0, c, d, r^*$:
\begin{align}
    a F_0 &= n (r^*-a)  (1-F_0) \label{eq:spa-sys-1-lmbd}  \\
    -(\lambda-\Phi_0) F_0 + \Phi_0 (n-1-n F_0) &= 0 \label{eq:spa-sys-2-lmbd}\\
    \left( \frac{r^*-c}{r^*} \right)^{n-1} \Phi_0 &= \lambda \left[ d + \log \left( \frac{r^*}{c} \right) - \sum_{k=1}^{n-1} \frac{1}{k} \left( \frac{r^*-c}{r^*} \right)^{k} \right] \label{eq:spa-sys-3-lmbd} \\
    \left( \frac{b-c}{b} \right)^{n-1} &= \lambda \left[ d + \log \left( \frac{b}{c} \right) - \sum_{k=1}^{n-1} \frac{1}{k} \left( \frac{b-c}{b} \right)^{k} \right] \label{eq:spa-sys-4-lmbd} \\
    1 - \frac{c}{r^*} &= F_0 \label{eq:spa-sys-5-lmbd} \, .
\end{align}

For notational convenience, we will write $F^*(a)$ as $F_0$. 

We will write every variable in terms of $F_0$, so we have a single-variable equation we can solve. 

From (\ref{eq:spa-sys-1-lmbd}), we get
\begin{align}
r^* = a \left( \frac{F_0}{n(1-F_0)} + 1 \right) =  \frac{n-(n-1)F_0}{n(1-F_0)} a \label{eq:spa-sys-1p-lmbd} \, .
\end{align}

From (\ref{eq:spa-sys-2-lmbd}) we get
\begin{align}
\Phi_0 = \frac{ \lambda F_0}{(n-1)(1-F_0)} \label{eq:spa-sys-2p-lmbd} \, .
\end{align}

From (\ref{eq:spa-sys-5-lmbd}) we get
\begin{align}
c = r^*(1-F_0) = \frac{n-(n-1)F_0}{n} a \label{eq:spa-sys-5p-lmbd} \, .
\end{align}

Subtracting (\ref{eq:spa-sys-3-lmbd}) and (\ref{eq:spa-sys-4-lmbd}) gives
\begin{align*}
\left( \frac{r^*-c}{r^*} \right)^{n-1} \Phi_0 -  \left( \frac{b-c}{b} \right)^{n-1} = \lambda \left[ \log \left( \frac{r^*}{b} \right) - \sum_{k=1}^{n-1} \frac{1}{k} \left( \frac{r^*-c}{r^*} \right)^{k} + \sum_{k=1}^{n-1} \frac{1}{k} \left( \frac{b-c}{b} \right)^{k} \right] \, .
\end{align*}

Substituting the expressions of $r^*, \Phi_0, c$ in terms of $F_0$ from (\ref{eq:spa-sys-1p-lmbd}), (\ref{eq:spa-sys-2p-lmbd}), (\ref{eq:spa-sys-5p-lmbd}) gives (writing $\tilde{a} = a/b$)
\begin{align*}
\frac{ \lambda F_0^n}{(n-1)(1-F_0)} - \left( 1 - \frac{n-(n-1)F_0}{n} \tilde{a} \right)^{n-1}  \\ = \lambda \left[ \log\left( \frac{n-(n-1)F_0}{n(1-F_0)} \tilde{a} \right) - \sum_{k=1}^{n-1} \frac{1}{k} F_0^k + \sum_{k=1}^{n-1} \frac{1}{k} \left( 1 - \frac{n-(n-1)F_0}{n} \tilde{a} \right)^{k} \right] \, .
\end{align*}

Let $\text{fn}(F_0)$ be the left hand side minus the right hand side (taking $\tilde{a} = a/b$ as fixed):
\begin{align*}
\text{fn}(F_0) = \frac{ \lambda F_0^n}{(n-1)(1-F_0)} -  \left( 1 - \frac{n-(n-1)F_0}{n} \tilde{a} \right)^{n-1} - \lambda \log\left( \frac{n-(n-1)F_0}{n(1-F_0)} \tilde{a} \right) \\ + \lambda \sum_{k=1}^{n-1} \frac{1}{k} F_0^k - \lambda \sum_{k=1}^{n-1} \frac{1}{k} \left( 1 - \frac{n-(n-1)F_0}{n} \tilde{a} \right)^{k} \, .
\end{align*}

(When $\tilde{a}$ is not fixed, we will write the above expression instead as $L(F_0,\tilde{a})$, and we will use this notation later in the proof.)

We first note that, because $\tilde{a} \geq k_l$, by definition of $k_l$ we have
\begin{align*}
\text{fn}(0) = -(1-\tilde{a})^{n-1} - \lambda \log(\tilde{a}) - \lambda \sum_{k=1}^{n-1} \frac{1}{k} (1-\tilde{a})^k \leq 0 \, ,
\end{align*}
with equality only when $\tilde{a} = k_l$.

We also note that as $F_0 \uparrow 1$, LHS grows as $\frac{1}{1-F_0}$ whereas RHS grows as $\log\left( \frac{1}{1-F_0}\right)$, so $\lim_{F_0 \uparrow 1} \text{fn}(F_0) = +\infty$.

Lastly, we compute the derivative of $\text{fn}$ as 

\begin{align*}
\text{fn}'(F_0) &= \frac{\lambda F_0^{n-1} (n-(n-1)F_0)}{(n-1)(1-F_0)^2} - (n-1) \left( 1 - \frac{n-(n-1)F_0}{n} \tilde{a} \right)^{n-2} \cdot \frac{(n-1)}{n} \tilde{a}  \\ 
&+\lambda \left[  \frac{(n-1)}{n-(n-1)F_0} - \frac{1}{1-F_0} + \sum_{k=1}^{n-1} F_0^{k-1} - \sum_{k=1}^{n-1} \left( 1 - \frac{n-(n-1)F_0}{n} \tilde{a} \right)^{k-1} \cdot \frac{(n-1)}{n} \tilde{a}  \right] \\
&= \frac{\lambda F_0^{n-1} (n-(n-1)F_0)}{(n-1)(1-F_0)^2} - (n-1) \left( 1 - \frac{n-(n-1)F_0}{n} \tilde{a} \right)^{n-2} \cdot \frac{(n-1)}{n} \tilde{a} \\
&+ \lambda \left[ \frac{(n-1)}{n-(n-1)F_0} - \frac{1}{1-F_0}  + \frac{1-F_0^{n-1}}{1-F_0} - \frac{(n-1)}{n-(n-1)F_0} \left[ 1 - \left( 1 - \frac{n-(n-1)F_0}{n} \tilde{a} \right)^{n-1} \right] \right]  \\
&= \frac{ \lambda F_0^{n-1} (n-(n-1)F_0)}{(n-1)(1-F_0)^2} - \frac{(n-1)^2}{n} \tilde{a} \left( 1 - \frac{n-(n-1)F_0}{n} \tilde{a} \right)^{n-2} \\ 
&+ \lambda \left[  - \frac{F_0^{n-1}}{1-F_0} + \frac{(n-1)}{n-(n-1)F_0} \left( 1 - \frac{n-(n-1)F_0}{n} \tilde{a} \right)^{n-1} \right] \\
&= F_0^{n-1} \lambda \left( \frac{n-(n-1)F_0}{(n-1)(1-F_0)^2} - \frac{1}{1-F_0} \right)  \\
&+ \left( 1- \frac{n-(n-1)F_0}{n} \tilde{a} \right)^{n-2} \left[ - \frac{(n-1)^2}{n} \tilde{a} + \frac{\lambda(n-1)}{n-(n-1)F_0} \left( 1-  \frac{n-(n-1)F_0}{n} \tilde{a} \right) \right] \\
&=  \frac{\lambda F_0^{n-1}}{(n-1)(1-F_0)^2}   + \left( 1- \frac{n-(n-1)F_0}{n} \tilde{a} \right)^{n-2} (n-1) \left( \frac{\lambda}{n-(n-1)F_0} - \frac{n-1+\lambda}{n} \tilde{a} \right) \, .
\end{align*}

We can see from the expression that $\text{fn}'(F_0)$ is increasing in $F_0$, so either $\text{fn}$ is increasing for all $F_0$ in range (if $\text{fn}'(0) \geq 0$, or it is decreasing for $F_0 \leq F_0^*$ for some constant $F_0^*$ and increasing for $F_0 \geq F_0^*$. Because, for $a/b > k_l$, $\text{fn}(0) < 0$ and $\lim_{F_0 \uparrow 1} \text{fn}(F_0) = +\infty$, we conclude that the equation $\text{fn}(F_0) = 0$ has a unique solution $F_0 \in (0,1)$. 

We will also need to show that this unique solution $F_0$ leads to feasible values of other parameters as well. The ones that concern us are $\Phi^*(a)$ and $r^*$.

We must have
\begin{align*}
\Phi^*(a) = \frac{ \lambda F_0}{(n-1)(1-F_0)} \leq 1 \Leftrightarrow F_0 \leq \frac{n-1}{n-1+\lambda} \, ,
\end{align*}
and
\begin{align*}
r^* = \frac{n-(n-1)F_0}{n(1-F_0)} a \leq b \Leftrightarrow F_0 \leq \frac{n(1-\tilde{a})}{n-(n-1)\tilde{a}} \, .
\end{align*}

We will show that $F_0 \leq \frac{n-1}{n}$.

We note that $F_0(\tilde{a}=\frac{\lambda n}{(1+\lambda)n-1}) = \frac{n-1}{n-1+\lambda}$ because when we plug in $\tilde{a} = \frac{\lambda n}{(1+\lambda)n-1}, F_0 = \frac{n-1}{n-1+\lambda}$ in the $L(F_0,\tilde{a})$ expression we get zero. We claim that $L(F_0, \tilde{a}=\frac{\lambda n}{(1+\lambda)n-1})$ is increasing in $F_0$ for $\frac{n-1}{n-1+\lambda} \leq F_0 \leq 1$. If this is true then we are done, because if for some $\tilde{a} < \frac{\lambda n}{(1+\lambda)n-1}$ we have $F_0(\tilde{a}) > \frac{n-1}{n-1+\lambda}$, then $0 = L(F_0(\tilde{a}), \tilde{a}) > L(F_0(\tilde{a}), \frac{\lambda n}{(1+\lambda)n-1}) > L(\frac{n-1}{n-1+\lambda},\frac{\lambda n}{(1+\lambda)n-1}) = 0$, a contradiction.  

To show that $L(F_0, \frac{\lambda n}{(1+\lambda)n-1})$ is increasing in $F_0$ for $1 \geq F_0 \geq \frac{n-1}{n-1+\lambda}$, we calculate (taking the expression from the $\text{fn}'(F_0)$ earlier)
\begin{align*}
\frac{\partial L(F_0, \tilde{a} = \frac{\lambda n}{(1+\lambda)n-1})}{F_0} = \\
= \frac{ \lambda F_0^{n-1}}{(n-1)(1-F_0)^2} +  \left( 1 - \frac{\lambda (n-(n-1)F_0)}{(1+\lambda)n-1} \right)^{n-2}   \cdot (n-1) \lambda \left( \frac{1}{n-(n-1)F_0} - \frac{n-1+\lambda}{(1+\lambda)n-1} \right) \geq 0 \, .
\end{align*}
This is true because
\begin{align*}
\frac{1}{n-(n-1)F_0} \geq \frac{1}{n-(n-1) \frac{\lambda n}{(1+\lambda)n-1}} = \frac{n-1+\lambda}{(1+\lambda)n-1} \, ,
\end{align*}
and
\begin{align*}
1 - \frac{\lambda(n-(n-1)F_0)}{(1+\lambda)n-1}  > 0 \, ,
\end{align*}
so the derivative is positive, as desired. 
(Note also that the expression $\frac{\partial L(F_0, \tilde{a} = \frac{ \lambda n}{(1+\lambda) n-1})}{F_0}$ is increasing in $F_0$ so we can also just plug in $F_0 = \frac{n-1}{n-1+\lambda}$ in that expression and check that the resulting expression is $> 0$.

Lastly, we want to show that $F_0 \leq \frac{n(1-\tilde{a})}{n-(n-1)\tilde{a}}$. Note that $\tilde{a} < \frac{ \lambda n}{(1+\lambda)n-1}$ implies that $\frac{n-1}{n-1+\lambda} < \frac{n(1-\tilde{a})}{n-(n-1)\tilde{a}}$, so this inequality is immediately implied by $F_0 \leq \frac{n-1}{n-1+\lambda}$ which we just proved.

We conclude that for $k_l \leq \tilde{a} \leq \frac{ \lambda n}{(1+\lambda) n-1}$, these parameters give rise to a feasible mechanism that is minimax optimal, as desired.
\end{proof}

\begin{proposition}\label{app:prop:sup-regret-spa-r-lmbd}
The worst-case $\lambda$-regret of $\emph{SPA}(r)$, the second-price auction with fixed reserve $r$, is 
\begin{align*}
R_{\lambda}(\textnormal{SPA}(r),\mathcal{F}_{\textnormal{iid}}) =\begin{cases}
-(1-\lambda) b + (b-a) \left( \frac{n-1}{n-1+\lambda} \right)^{n-1}  &\text{ if } r = a \\
-(1-\lambda) b + \frac{(b-r)^n (n-1)^{n-1}}{((n-1+\lambda)(b-r)-r)^{n-1}} &\text{ if } a < r \leq \frac{\lambda}{1+\lambda} b \\
\lambda r &\text{ if } \frac{\lambda}{1+\lambda} b \leq r \leq b \, .
\end{cases}
\end{align*}

%

Therefore, 
\begin{align*}
\inf_{r \in [a,b]} R_{\lambda}(\textnormal{SPA}(r),\mathcal{F}_{\textnormal{iid}}) = \begin{cases}  - (1-\lambda) b  +  \left( \frac{n}{n+\lambda} \right)^{n} b  &\text{ if } \frac{a}{b} \leq 1 - \left( \frac{n}{n+\lambda}\right)^{n} \left( \frac{n-1+\lambda}{n-1} \right)^{n-1}   \\
- (1-\lambda) b  + 
\left( \frac{n-1}{n-1+\lambda} \right)^{n-1} (b-a) &\text{ if }   \frac{a}{b} \geq 1 - \left( \frac{n}{n+\lambda}\right)^{n} \left( \frac{n-1+\lambda}{n-1} \right)^{n-1} \, .
\end{cases}
\end{align*}


The optimal $r$  in the first case (``low'' $a/b$) is $r = \frac{\lambda}{n+\lambda} b$ and the optimal $r$ in the second case (``high'' $a/b$) is $r = a$.

\end{proposition}

This result is valid for any $n \geq 1$, if for $n = 1$ we interpret any term with the $n-1$ exponent as $1$. In agreement with results from $\mathcal{M}_{\textnormal{SPA-rand}}$ and our intuition, when scale information is important ($a/b$ is high), the regret-minimizing reserve is no reserve $r = a$. 

Before we prove the main result (Proposition~\ref{app:prop:sup-regret-spa-r-lmbd}) characterizing the minimax $\lambda$-regret SPA with fixed reserve (including no reserve), we first derive the regret expression of $\textnormal{SPA}(r)$ for a fixed $r$.

\begin{proposition}\label{prop:reg-spa-det-a}
Fix $r \in [a,b]$. The regret of $\textnormal{SPA}(r)$ against a joint distribution $\mathbf{F}$ is
\begin{align*}
-(1-\lambda) b + r \mathbf{F}_n^{(1)}(r^-) - \lambda  \int_{v=a}^{v=r} \mathbf{F}_n^{(1)}(v) dv + \int_{v=r}^{v=b} ( \mathbf{F}_n^{(2)}(v) - \lambda \mathbf{F}_n^{(1)}(v) ) dv \, ,
\end{align*}
where $\mathbf{F}_n^{(1)}(r^-) = \Pr(v^{(1)} < v)$.
If we further assume that $\mathbf{F}$ is $n$ i.i.d. with marginal $F$, then the regret is
\begin{align*}
-(1-\lambda) b + r F(r^-)^{n} - \lambda  \int_{v=a}^{v=r} F(v)^{n} dv + \int_{v=r}^{v=b} ( n F(v)^{n-1} - (n-1+\lambda) F(v)^{n}  ) dv \, .
\end{align*}
\end{proposition}

\begin{proof}
The regret is
\begin{align*}
    &  \lambda \bE[v^{(1)}] - \bE[\max(v^{(2)},r) \1(v^{(1)} \geq r) ] \\
    &= \lambda \bE[v^{(1)}] - \bE[v^{(2)} \1(v^{(2)} > r)] - \bE[r \1(v^{(2)} \leq r \leq v^{(1)})] \, .
\end{align*}

The first term is 
\begin{align*}
    \bE[v^{(1)}] = \int_{v \geq 0} \Pr(v^{(1)} > v) dv = b - \int_{v=a}^{v=b} \mathbf{F}_n^{(1)}(v) dv \, .
\end{align*}

The second term's calculation is analogous to that of Lemma~\ref{lem:int-by-part}. We have
\begin{align*}
    \bE[v^{(2)} \1(v^{(2)} > r) ] 
    &= \int_{v' \in (r,b]} v' d\mathbf{F}_n^{(2)}(v') 
    = \int_{v' \in (r,b]} \left( a + \int_{v \in [a,v')} dv \right) d\mathbf{F}_n^{(2)}(v')\\
    &= a (\mathbf{F}_n^{(2)}(b) - \mathbf{F}_n^{(2)}(r)) + \int_{v \in [a,r]} \int_{v' \in (r,b]} d\mathbf{F}_n^{(2)}(v')dv +  \int_{v \in (r,b]} \int_{v' \in (v,b]} d\mathbf{F}_n^{(2)}(v') dv \\
    &=a (1 - \mathbf{F}_n^{(2)}(r)) +  \int_{v \in [a,r]} (1 - \mathbf{F}_n^{(2)}(r) ) dv + \int_{v \in (r,b]} (1 - \mathbf{F}_n^{(2)}(v) ) dv \\
    &= b - r \mathbf{F}_n^{(2)}(r) - \int_{v=r}^{v=b} \mathbf{F}_n^{(2)}(v) dv \, .
\end{align*}

The third term is
\begin{align*}
    \bE[r \1(v^{(2)} &\leq r \leq v^{(1)})] = r \Pr(v^{(2)} \leq r \leq v^{(1)}) = r( \Pr(v^{(2)} \leq r) - \Pr(v^{(1)} < r)) \\&= r(\mathbf{F}_n^{(2)}(r) - \mathbf{F}_n^{(1)}(r^-)) \, .
\end{align*}
Together, we have, for $r \in [a,b]$,
\begin{align*}
R_{n}(\text{SPA}(r), \mathbf{F}_n) = -(1-\lambda) b + r \mathbf{F}_n^{(1)}(r^-) - \lambda  \int_{v=a}^{v=r} \mathbf{F}_n^{(1)}(v) dv + \int_{v=r}^{v=b} ( \mathbf{F}_n^{(2)}(v) - \lambda \mathbf{F}_n^{(1)}(v) ) dv \, ,
\end{align*}
as desired.

When $\mathbf{F}$ is $n$ i.i.d. $F$, we have $\mathbf{F}_{n}^{(1)}(r^-) = \Pr(\max(\mathbf{v}) < r) = \prod_{i=1}^{n} \Pr(v_i < r) = F(r^-)^{n}$ where the second-to-last equality uses the fact that the $v_i$'s are independent. We also have $\mathbf{F}_{n}^{(1)}(v) = F(v)^{n}$ and $\mathbf{F}_{n}^{(2)}(v) = n F(v)^{n-1} - (n-1) F(v)^{n}$.
\end{proof}

Now we are ready to prove the main result.

\begin{proof}[Proof of Proposition~\ref{app:prop:sup-regret-spa-r-lmbd}]

From Proposition~\ref{prop:reg-spa-det-a}, the $\lambda$-regret of $\textnormal{SPA}(r)$ is
\begin{align*}
    R_n(\textnormal{SPA}(r), F) = -(1-\lambda)b + r F_{-}(r)^n -  \lambda \int_{v=a}^{v=r} F(v)^n dv  + \int_{v=r}^{v=b} n F(v)^{n-1} - (n-1+\lambda) F(v)^n dv \, .
\end{align*}

We first assume that $n \geq 2$ and $r \in (a,b]$. Let $c = F_{-}(r)$. (Note here that we require $r > a$ in order for us to have the freedom to set the value of $c = \Pr(v < r)$, the mass \textit{strictly} below $r$. If $r = a$, i.e. there is no reserve, then $c = 0$ by definition. This is why we consider the case $r = a$, i.e. no reserve, separately.) Note that the integrand $n F(v)^{n-1} - (n-1+\lambda) F(v)^n$ is increasing for $F(v) \leq \frac{n-1}{n-1+\lambda}$ and is decreasing for $F(v) \geq \frac{n-1}{n-1+\lambda}$.  To minimize $\int_{v \in [a,r]^n} F(v)^n dv $ we must have $F(v) = 0$ for $v  \in [a,r-\epsilon]$ for arbitrarily small $\epsilon > 0$, and to maximize $\int_{v \in (r,b]} n F(v)^{n-1} - n F(v)^n dv$, the only constraint we have is $F(v) \geq c$ so for $v \in (r,b]$ we set $F(v) = \frac{n-1}{n-1+\lambda}$ if $c \leq \frac{n-1}{n-1+\lambda}$ and $F(v) = c$ otherwise. Note that the sup over first case of $c \leq \frac{n-1}{n-1+\lambda}$ is simply the second case with $c = \frac{n-1}{n-1+\lambda}$. Because we take the sup over $F$, we can let $\epsilon \downarrow 0$ and get that the worst-case regret is 
\begin{align*}
- (1-\lambda) b + \sup_{c \in [\frac{n-1}{n-1+\lambda},1] } rc^n +  (b-r)( nc^{n-1} - (n-1+\lambda) c^n) \, .
\end{align*}
Now, the derivative of this expression of $c$ is $nc^{n-2} \left[ rc+(b-r)(n-1-(n-1+\lambda)c) \right]$. The expression in $\left[ \cdots \right]$ is linear in $c$. At $c = \frac{n-1}{n-1+\lambda}$, the derivative expression is $nr \left( \frac{n-1}{n-1+\lambda} \right)^{n-1} \geq 0$. At $c = 1$, the expression is $n ( (1+\lambda)r - \lambda b  )$. So if $r \geq   \frac{\lambda}{1+\lambda}b$, the first derivative is always $\geq 0$, so the maximum is achieved at $c = 1$ and the value is $\lambda r$. If $r \leq \frac{\lambda}{1+\lambda}b$, the maximum is achieved at $c^* = \frac{(n-1)(b-r)}{(n-1+\lambda)(b-r)-r} \in \left[ \frac{n-1}{n-1+\lambda}, 1 \right]$ and the value is $-(1-\lambda) b + \frac{(b-r)^n (n-1)^{n-1}}{((n-1+\lambda)(b-r)-r)^{n-1}}$.

Now we consider the case $r = a$. In this case, by definition $c = 0$ and we have
\begin{align*}
R_n(\text{SPA}(a), F) = -(1-\lambda) b + \int_{v \in (a,b]} n F(v)^{n-1} - (n-1+\lambda) F(v)^n dv \, ,
\end{align*}
so
\begin{align*}
R_n(\text{SPA}(a), \mathcal{F}) = -(1-\lambda) b + (b-a) \sup_{z \in [0,1]} nz^{n-1} - (n-1+\lambda) z^n \, .
\end{align*}

The maximum occurs at $z = \frac{n-1}{n-1+\lambda}$ which gives
\begin{align*}
R_n(\text{SPA}(a), \mathcal{F})  = -(1-\lambda) b + (b-a) \left( \frac{n-1}{n-1+\lambda} \right)^{n-1} \, .
\end{align*}

Now we deal with the case $n=1$. The regret expression reduces to
\begin{align*}
    R_{1}(\text{SPA}(r), F) =  -(1-\lambda)b + r F_{-}(r) -  \lambda \int_{v \in [a,r]} F(v) dv  + \int_{v \in (r,b]} ( 1 - \lambda F(v) ) dv \, .
\end{align*}

For $r > a$, we have
\begin{align*}
R_{1}(\text{SPA}(r), \mathcal{F}) &= -(1-\lambda) b + \sup_{c \in [0,1]} rc + (b-r)(1-\lambda c) \\&= \max(\lambda b - r, \lambda r) = \begin{cases}
\lambda b - r &\text{ if } a < r \leq \frac{\lambda}{1+\lambda} b \\
\lambda r &\text{ if } r \geq \frac{\lambda}{1+\lambda} b \, ,
\end{cases}
\end{align*}
because the expression under sup is linear in $c$ so it achieves the extrema at one of the end points, either at $c = 0$ or $c = 1$. 

For $r = a$ we have
\begin{align*}
R_{1}(\text{SPA}(a), F) =  -(1-\lambda)b   + \int_{v \in (a,b]} ( 1 - \lambda F(v) ) dv \, .
\end{align*}
This is maximized when $F(v) = 0$ for all $v \in (a,b]$ and we get
\begin{align*}
R_1(\text{SPA}(a),\mathcal{F}) = -(1-\lambda) b + (b-a) = \lambda b - a \, .
\end{align*}

We therefore have
\begin{align*}
R_1(\text{SPA}(r), \mathcal{F}) = \begin{cases}
\lambda b - a &\text{ if } r = a \\
\lambda b - r &\text{ if } a < r \leq \frac{\lambda}{1+\lambda} b \\
\lambda r &\text{ if } r \geq \frac{\lambda}{1+\lambda} b \, .
\end{cases}
\end{align*}

Note that the second regime and the first regime are continuous whenever the second regime is applicable, but we will keep them separate for clarity (because the first regime $r=a$ is always applicable, whereas the second regime $r \in (a, \frac{\lambda}{1+\lambda} b]$ is applicable only when $\frac{a}{b} < \frac{\lambda}{1+\lambda}$.

Now we want to choose the optimal $r$ to minimize the worst-case regret. First consider the case $n \geq 2$. We note that
\begin{align*}
(b-a) \left( \frac{n-1}{n-1+\lambda} \right)^{n-1} \leq \frac{(b-a)^n (n-1)^{n-1}}{( (n-1+\lambda)(b-a)-a )^{n-1}} \, ,
\end{align*}
with equality if and only if $a = 0$. Therefore, if $\frac{a}{b} < \frac{\lambda}{1+\lambda}$, that is, the regime $r \in (a, \frac{\lambda}{1+\lambda} b ]$ is permissible, then the worst-case regret under $r = a$ is lower than under $r = a^+$, slightly above $a$. In contrast, the worst-case regret is continuous at $r = \frac{\lambda}{1+\lambda} b$. Given that the regret in the third regime $\lambda r$ is linear in $r$, the worst $r$ (lowest regret) occurs at $r = \frac{\lambda}{1+\lambda}b$ with regret $\frac{\lambda^2}{1+\lambda} b$.


First consider the case $\frac{a}{b} < \frac{\lambda}{1+\lambda}$, so all 3 regimes of $r$ are permissible. 
 
In the $r \in [\frac{\lambda}{1+\lambda}b, b]$ regime, the regret is $\lambda r$, so the lowest regret occurs at $r = \frac{\lambda}{1+\lambda} b$ and has value $\frac{\lambda^2}{1+\lambda} b$.
 
In the $r \in (a, \frac{\lambda}{1+\lambda} b]$ regime, the regret is
\begin{align*}
-(1-\lambda) b + (n-1)^{n-1} \exp \left\{ n \log(b-r)- (n-1) \log ( (n-1+\lambda)b - (n+\lambda) r ) \right\} \, .
\end{align*}
The derivative of the expression in $\left\{ \cdots \right\}$ is
\begin{align*}
-\frac{n}{b-r} + \frac{(n-1)(n+\lambda)}{(n-1+\lambda) b -(n+\lambda) r} = \frac{(n+\lambda) r - \lambda b}{(b-r) ((n-1+\lambda) b -(n+\lambda) r )} \, .
\end{align*}

Therefore, in this second regime, the worst-case regret is decreasing for $\frac{a}{b} \leq \frac{\lambda}{n+\lambda}$ and increasing for $\frac{a}{b} \geq \frac{\lambda}{n+\lambda}$. So if $\frac{a}{b} \leq \frac{\lambda}{n+\lambda}$, the $r$ that minimizes worst-case regret is $r = \frac{\lambda}{n+\lambda} b$, which gives the regret
\begin{align*}
-(1-\lambda) b + \left( \frac{n}{n+\lambda} \right)^{n} b \, .
\end{align*}
Therefore, the overall worst-case regret, including $r = a$ also, has regret
\begin{align*}
\min\left(    -(1-\lambda) b + \left( \frac{n}{n+\lambda} \right)^{n} b, -(1-\lambda) b + \left( \frac{n-1}{n-1+\lambda} \right)^{n-1} (b-a) \right) \, ,
\end{align*}
corresponding to $r = \frac{\lambda}{n+\lambda} b$ and $r = a$ respectively. 

We can show that 
\begin{align*}
\left( \frac{n-1}{n-1+\lambda} \right)^{n-1} \geq  \left( \frac{n}{n+\lambda} \right)^{n}  \, .
\end{align*}
So, for $\frac{a}{b} \leq 1 - \left( \frac{n}{n+\lambda} \right)^{n} \left( \frac{n-1+\lambda}{n-1} \right)^{n-1} $, $r = \frac{\lambda}{n+\lambda} b$ gives the lowest worst-case regret, and for $\frac{a}{b} \geq 1 - \left( \frac{n}{n+\lambda} \right)^{n} \left( \frac{n-1+\lambda}{n-1} \right)^{n-1} $, $r =a$ gives the lowest worst-case regret. We can show that
\begin{align*}
0 \leq 1 - \left( \frac{n}{n+\lambda} \right)^{n} \left( \frac{n-1+\lambda}{n-1} \right)^{n-1} \leq \frac{\lambda}{n+\lambda} \, ,
\end{align*}
so this threshold is always interior. 

For $ \frac{\lambda}{n+\lambda} \leq \frac{a}{b} \leq \frac{\lambda}{1+\lambda}$, the worst-case regret is increasing in $r$ for the second regime, so the worst-case in the second regime is when $r = a^+$, but we already know that the worst-case regret is lower under $r = a$ than under $r = a^+$, so the best $r$ is $r = a$ with regret 
\begin{align*}
-(1-\lambda) b + \left( \frac{n-1}{n-1+\lambda} \right)^{n-1} (b-a) \, .
\end{align*}

For $\frac{a}{b} \geq \frac{\lambda}{1+\lambda}$, the second regime is not possible, and the third regime's worst case is again $r = a^+$ which has highest regret than $r =a $, so again the best $r$ is $r = a$ with regret 
\begin{align*}
-(1-\lambda) b + \left( \frac{n-1}{n-1+\lambda} \right)^{n-1} (b-a) \, ,
\end{align*}
so the $r$ that minimizes worst-case regret is $r = \frac{\lambda}{n+\lambda} b$, which gives the regret
\begin{align*}
-(1-\lambda) b + \left( \frac{n}{n+\lambda} \right)^n b \, .
\end{align*}


If $\frac{a}{b} < \frac{\lambda}{n+\lambda}$, then the worst-case regret is minimized at $r = \frac{\lambda}{n+\lambda }b$ and the worst-case regret value is
\begin{align*}
-(1-\lambda) b + \frac{n^n}{(n+\lambda)^n} b \, .
\end{align*}

If $\frac{a}{b} \geq \frac{\lambda}{n+\lambda}$, then we always have $r > a \geq \frac{\lambda}{n+\lambda} b$, so the worst-case regret is minimized at $r = a^+$, but we have already shown that $r = a^+$ always has higher regret (worse) than $r = a$. 

Therefore, for $\frac{a}{b} < \frac{\lambda}{n+\lambda}$, the worst-case regret is
\begin{align*}
\min\left( - (1-\lambda) b + \frac{(b-a)(n-1)^{n-1}}{(n-1+\lambda)^{n-1}} ,  \left( \frac{n^n}{(n+\lambda)^n} - 1 + \lambda \right) b , \frac{\lambda^2}{1+\lambda} b \right) \, .
\end{align*}

The first, second, and third terms correspond to $r = a$, $r \in (a, \frac{\lambda}{1+\lambda }b]$ and $r \in [\frac{\lambda}{1+\lambda} b, b]$ respectively. 

The third one is higher than the second one because the third one, as we have already shown, is the second one with $r = \frac{\lambda}{1+\lambda} b$ which by our proof has higher regret than that at $r = \frac{\lambda}{1+\lambda} b$. So the worst-case regret becomes
\begin{align*}
\min\left(   \frac{(n-1)^{n-1}}{(n-1+\lambda)^{n-1}} (b-a) - (1-\lambda) b ,   \frac{n^n}{(n+\lambda)^n} b- (1 - \lambda)  b  \right) \, .
\end{align*}
We note that 
\begin{align*}
\left( \frac{n-1}{n-1+\lambda} \right)^{n-1} \geq \left( \frac{n}{n+\lambda} \right)^{n} \, .
\end{align*}
This is true because $\left( \frac{x}{x+\lambda} \right)^x = \exp \left\{ x \log(x) - x \log(x+\lambda) \right\}$ is a decreasing function of $x$: the derivative of the expression in the $\left\{ \cdots \right\}$ is
\begin{align*}
\left( x \cdot \frac{1}{x} + \log(x) \right) - \left( x \cdot \frac{1}{x+\lambda} + \log(x+\lambda) \right) = 1 - \frac{x}{x+\lambda} + \log \left( \frac{x}{x+\lambda} \right) \leq 0 \, ,
\end{align*}
because $1 + \log(u) \leq u$ for all $u$. (Let $u' = \log(u)$; this because the well-known $1+u' \leq \exp(u')$.)

We however have
\begin{align*}
\left( \frac{n-1}{n-1+\lambda} \right)^{n-1} \leq \left( \frac{n}{n+\lambda} \right)^{n-1} \, .
\end{align*}
Therefore, we have
\begin{align*}
  \left( \frac{n-1}{n-1+\lambda} \right)^{n-1} (b-a) - (1-\lambda) b  \geq   \left( \frac{n}{\lambda+n} \right)^{n} b- (1 - \lambda)  b  \, ,
\end{align*}
when $a = 0$ but 
\begin{align*}
  \left( \frac{n-1}{n-1+\lambda} \right)^{n-1} (b-a) - (1-\lambda) b  \leq   \left( \frac{n}{\lambda+n} \right)^{n} b- (1 - \lambda)  b \, ,
\end{align*}
when $a = \frac{\lambda}{n+\lambda b}$. The threshold to define which one is better is therefore always in the middle, at $1 - \left( \frac{n}{n+\lambda} \right)^{n} \left( \frac{n-1+\lambda}{n-1} \right)^{n-1}$.

 We now consider $\frac{\lambda}{n+\lambda} \leq \frac{a}{b} \leq \frac{\lambda}{1+\lambda}$. We have shown that the worst-case regret in the second regime occurs at $r = a^+$ which is always higher regret than $r = a$, so we only need to consider the first and third regime: the worst-case regret is
\begin{align*}
&\min\left(-(1-\lambda) b + (b-a) \left( \frac{n-1}{n-1+\lambda} \right)^{n-1}, \frac{\lambda^2}{1+\lambda} b  \right) \\&\quad= \min\left( (b-a) \left( \frac{n-1}{n-1+\lambda} \right)^{n-1}, \frac{1}{1+\lambda} b  \right) - (1-\lambda) b \, .
 \end{align*}

We know that $\left( \frac{n-1}{n-1+\lambda} \right)^{n-1} \leq \frac{1}{1+\lambda}$ because $n-1 \geq 1$ and $\left( \frac{x}{x+\lambda} \right)^{x}$ is a decreasing function of $x$. We also know that $b-a \leq b$. Therefore, the first expression in the min (first regime) is always lower than the second expression (third regime). So the worst case regret in this case is just
\begin{align*}
-(1-\lambda) b + (b-a) \left( \frac{n-1}{n-1+\lambda} \right)^{n-1} \, ,
\end{align*}
which is achieved at $r = a$. 

Lastly, we consider the case $\frac{a}{b} \geq \frac{\lambda}{1+\lambda}$. Then the second regime is never applicable, and the worst-case regret is
\begin{align*}
\min\left(-(1-\lambda) b + (b-a) \left( \frac{n-1}{n-1+\lambda} \right)^{n-1}, \frac{\lambda^2}{1+\lambda} b  \right) = -(1-\lambda) b + (b-a) \left( \frac{n-1}{n-1+\lambda} \right)^{n-1} \, ,
\end{align*}
where we know the first term in the min is less than the second term by what we just proved. This is also achieved when $r = a$.

We therefore conclude that for $n \geq 2$ the worst-case regret (and the corresponding optimal reserve $r^*$) as a function of $a$ and $b$ is as follows.
\begin{align*}
\inf_{r \in [a,b]} R_n(\text{SPA}(r), \mathcal{F}) = \begin{cases}
\left( \frac{n-1}{n-1+\lambda} \right)^{n-1} (b-a) - (1-\lambda) b &\text{ if } \frac{a}{b} \leq 1 - \left( \frac{n}{n+\lambda}\right)^{n} \left( \frac{n-1+\lambda}{n-1} \right)^{n-1}  \text{ or } \frac{a}{b} \geq \frac{\lambda}{n+\lambda} \\
\left( \frac{n}{n+\lambda} \right)^{n} b- (1-\lambda) b &\text{ if }  1 - \left( \frac{n}{n+\lambda}\right)^{n} \left( \frac{n-1+\lambda}{n-1} \right)^{n-1} \leq \frac{a}{b} < \frac{\lambda}{n+\lambda}  \, .
\end{cases}
\end{align*}
In the first case, $r^* = a$. In the second case, $r^* = \frac{\lambda}{n+\lambda} b$.

Now we calculate the optimal $r^*$ and the best worst-case regret for the case $n = 1$. In the case $\frac{a}{b} < \frac{\lambda}{1+\lambda}$, then all 3 regimes are possible. The lowest worst-case regret in the second regime is $\frac{\lambda^2}{1+\lambda} b$ when $r = \frac{\lambda}{1+\lambda} b$, which is the same as the lowest worst-case regret in the third regime. Therefore,
\begin{align*}
\inf_{r \in [a,b]} R_1(\text{SPA}(r), \mathcal{F}) = \min\left( \lambda b - a, \frac{\lambda^2}{1+\lambda } b \right) = \frac{\lambda^2}{1+\lambda} b \, ,
\end{align*}
where the last part is true because $\frac{a}{b} < \frac{\lambda}{1+\lambda}$ implies $\lambda b - a > \frac{\lambda^2}{1+\lambda } b$. Here, $r^* = \frac{\lambda}{1+\lambda} b$.

Now consider the case $\frac{a}{b} \geq \frac{\lambda}{1+\lambda} $, then the second regime is inapplicable, and the third case holds for any $r \in (a,b]$, and the lowest worst-case regret in this regime is $\lambda a$ at $r = a^+$, so
\begin{align*}
\inf_{r \in [a,b]} R_1(\text{SPA}(r), \mathcal{F}) = \min\left( \lambda b - a, \lambda a \right) = \lambda b - a \, ,
\end{align*}
where the last part is true because $\frac{a}{b} \geq \frac{\lambda}{1+\lambda} b$ implies $\lambda b - a \leq \lambda a$. Here, $r^* = a$.
\end{proof}

\jadelete{

\section{Structure and Performance of Optimal Mechanisms within Subclasses}\label{app:sec:other-mech-classes-discuss}

\subsection{Structure of Optimal Mechanisms within Subclasses}\label{app:subsec:other-mech-classes-structure}

We now use the analytical results previously derived in Section~\ref{sec:other-mech-classes} to show the structure and evaluate the performance of these mechanisms that are optimal in the subclasses $\mathcal{M}_{\textnormal{std}}$, $\mathcal{M}_{\textnormal{SPA-rand}}$, $\mathcal{M}_{\textnormal{SPA-det}}$ and $\mathcal{M}_{\textnormal{SPA-a}}$.

\begin{figure}[h!]
	\begin{subfigure}{0.48\linewidth}
    	\centering
\includegraphics[height=0.75\linewidth]{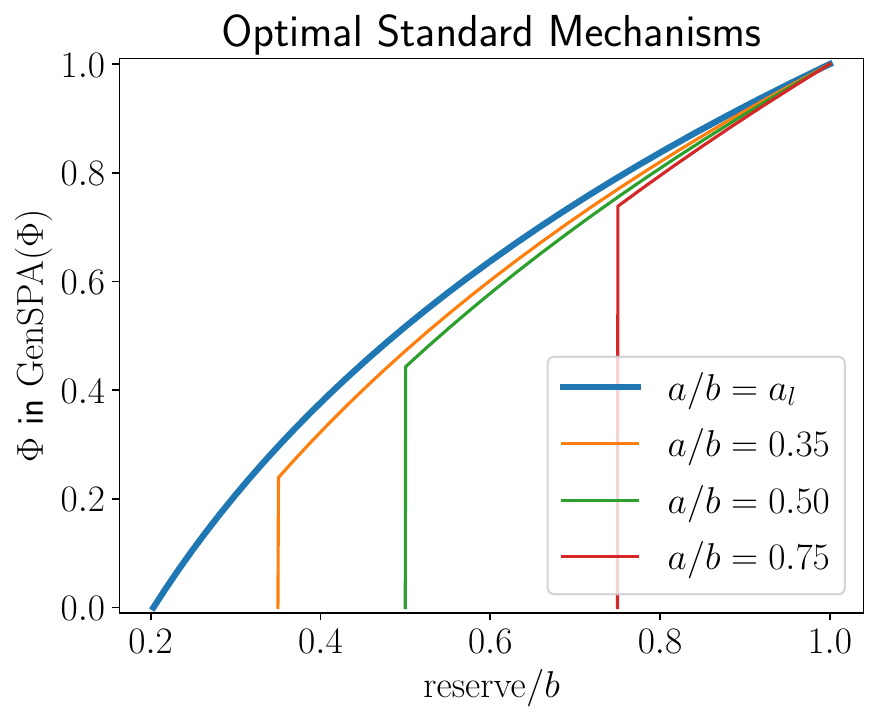}
    	\caption{$\Phi$ of $\textnormal{GenSPA}(\Phi)$ in $\mathcal{M}_{\textnormal{std}}$}
    	\label{fig:phi-std}
    \end{subfigure}%
	\begin{subfigure}{0.48\linewidth}
		\centering
		\includegraphics[height=0.75\linewidth]{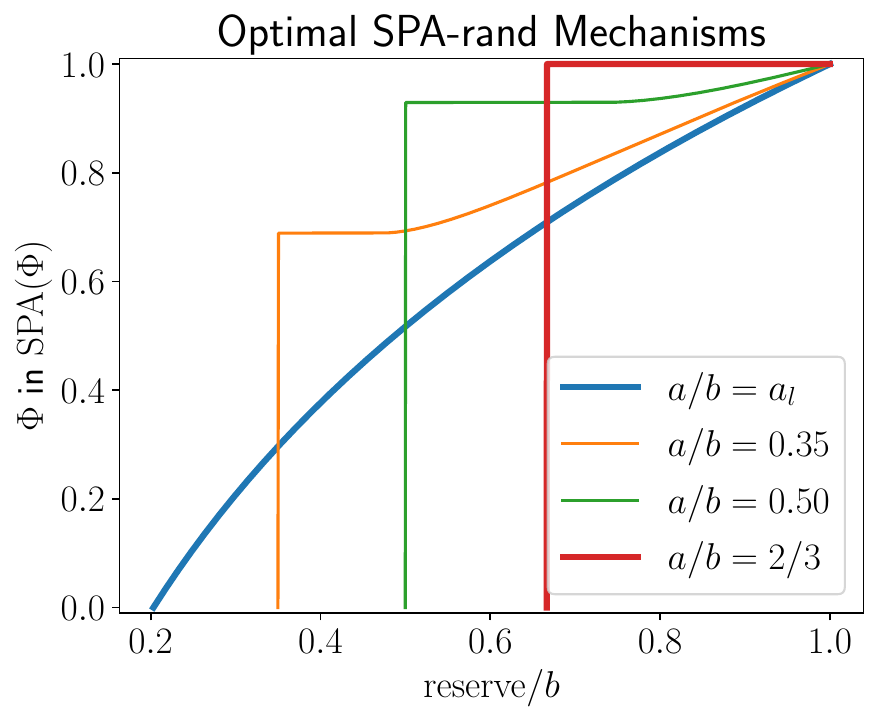}
		\caption{$\Phi$ of $\textnormal{SPA}(\Phi)$ in $\mathcal{M}_{\textnormal{SPA-rand}}$}
		\label{fig:phi-spa-rand}
	\end{subfigure}
	\caption{Structure of the minimax regret optimal mechanisms in the classes of standard mechanisms (generous SPAs) and SPAs with random reserve with $n = 2$ bidders.}
\label{fig:std-spa-rand-structure}
\end{figure}

Figure~\ref{fig:std-spa-rand-structure} shows the structure of the optimal mechanisms in $\mathcal{M}_{\textnormal{std}}$, the class of standard mechanisms, and $\mathcal{M}_{\textnormal{SPA-rand}}$, the class of SPAs with random reserve. Note that for $a/b \leq k_l$, the optimal mechanisms in both classes are the same as the optimal mechanism in $\mathcal{M}_{\textnormal{all}}$, namely $\textnormal{SPA}(\Phi)$. Henceforth we will focus on the case $a/b \geq k_l$. (Note also that the distribution $\Phi$s in Figures \ref{fig:phi-std} and \ref{fig:phi-spa-rand} are \textit{not} comparable, even for the same $\Phi$, because the mechanism structures are different.)

Theorem~\ref{thm:char-std-mech} states that the optimal mechanism in $\mathcal{M}_{\textnormal{std}}$ is $\textnormal{GenSPA}(\Phi)$, a \textit{generous} SPA with reserve distribution $\Phi$, meaning that the mechanism allocates like $\textnormal{SPA}(\Phi)$, except in the case when all but one bidder have the lowest value $a$, then it allocates to the highest (and only non-lowest) value with probability one. This distribution $\Phi$ has a point mass at $a$ and a density on $(a,b]$. Figure~\ref{fig:phi-std} shows the plot of $\Phi$ for $a/b \in \{k_l,0.35,0.50,0.75\}$. We see that as $a/b$ gets larger, the point mass at $a$ becomes bigger, and the density component of $\Phi$ for different $a/b$ are very close to each other.

Theorem~\ref{thm:char-spa-rand-mech} states that the optimal $\Phi$ of $\textnormal{SPA}(\Phi)$ in the class $\mathcal{M}_{\textnormal{SPA-rand}}$ has three regimes determined by the lower threshold $k_l$ and the upper threshold $k_h' = \lambda n/((1+\lambda)n-1)$. Here, we consider the minimax regret problem $\lambda = 1$ and $n = 2$ bidders, so the thresholds are $k_l \approx 0.2032$ and $2/3$. For $a/b \leq k_l$, $\Phi$ has a density in $[k_lb, b]$, while for $a/b \geq 2/3$, $\Phi$ is a point mass at $a$, namely, no reserve is optimal. Figure~\ref{fig:phi-spa-rand} shows the optimal $\Phi$ in the intermediate regime $k_l \leq a/b \leq 2/3$. In this regime, $\Phi$ has a point mass at $a$ and a density in $(a,b]$. We can see that the $\Phi$ in this regime interpolates between the smooth density (blue) and a single point mass (red). 

}

\end{document}